\documentclass[11pt]{article}
\usepackage{fullpage}

\usepackage{graphicx}
\usepackage{amsmath, bm, amsthm, amssymb, enumitem}
\usepackage{hyperref}

\title{The largest eigenvalue of real symmetric, Hermitian and Hermitian self-dual random matrix models with rank one external source, part I}

\author{Dong Wang\footnote{Department of Mathematics, University of Michigan, Ann Arbor, MI, 48105, USA \newline email: \href{mailto:dowang@umich.edu}{\texttt{dowang@umich.edu}}}}

\newcommand{\pdf}{p.d.f.}
\newcommand{\ie}{i.e.\@}
\newcommand{\eg}{e.g.}
\newcommand{\cf}{cf.}

\newcommand{\Brezin}{Br\'{e}zin}

\newcommand{\Prob}{\mathbb{P}}
\newcommand{\realR}{\mathbb{R}}
\newcommand{\compC}{\mathbb{C}}

\newcommand{\quatH}{\mathbb{H}}

\newcommand{\acc}{\mathbf{a}_c}
\newcommand{\redge}{\mathbf{e}}
\newcommand{\J}{\mathcal{J}_V}
\newcommand{\gfn}{\mathbf{g}}
\newcommand{\Gfn}{\mathbf{G}}
\newcommand{\Hfn}{\mathbf{H}}

\newcommand{\Av}{\mathcal{A}_V}
\newcommand{\ximax}[1][n]{\xi_{\max}{(#1)}}
\newcommand{\Cont}{\mathcal{C}}

\newcommand{\constZ}{\mathbf{Z}}
\newcommand{\M}{\mathcal{M}}
\newcommand{\tildeP}{\tilde{P}}
\newcommand{\bfepsilon}{\tilde{\epsilon}}
\newcommand{\bfC}{\mathbf{C}}
\newcommand{\erf}{\Phi}
\newcommand{\A}{\mathbf{A}}

\newcommand{\Virag}{Vir\'{a}g}

\DeclareMathOperator{\Tr}{Tr}
\DeclareMathOperator{\diag}{diag}

\DeclareMathOperator{\supp}{supp}

\DeclareMathOperator{\E}{E}

\DeclareMathOperator{\dist}{dist}
\DeclareMathOperator{\pv}{p.v.}

\newtheorem{prop}{Proposition}[section]
\newtheorem{thm}{Theorem}[section]

\theoremstyle{definition}
\newtheorem{condition}{Condition}

\theoremstyle{remark}
\newtheorem{rmk}{Remark}[section]

\begin{document}

\maketitle

\begin{abstract}
  We consider the limiting location and limiting distribution of the largest eigenvalue in real symmetric ($\beta = 1$), Hermitian ($\beta = 2$), and Hermitian self-dual ($\beta = 4$) random matrix models with rank $1$ external source. They are analyzed in a uniform way by a contour integral representation of the joint probability density function of eigenvalues. Assuming the ``one-band'' condition and certain regularities of the potential function, we obtain the limiting location of the largest eigenvalue when the nonzero eigenvalue of the external source matrix is not the critical value, and further obtain the limiting distribution of the largest eigenvalue when the nonzero eigenvalue of the external source matrix is greater than the critical value. When the nonzero eigenvalue of the external source matrix is less than or equal to the critical value, the limiting distribution of the largest eigenvalue will be analyzed in a subsequent paper. In this paper we also give a definition of the external source model for all $\beta > 0$.
\end{abstract}

\section{Introduction and statement of results}

\subsection{Introduction}

In this paper we will be concerned with the distribution of the largest eigenvalue $\ximax$ in the following ensembles of matrices $\{ M \}$:
\begin{itemize}
\item 
  The set of $n \times n$ real symmetric matrices, with the probability distribution function (\pdf)
  \begin{equation} \label{eq:pdf_of_matrix_beta=1}
    p_{n,1}(M)dM := \frac{1}{\tilde{C}_{n,1}} e^{-n\Tr(V_1(M)-\A_{n,1}M)} dM.
  \end{equation}
\item
  The set of $n \times n$ Hermitian matrices, with the \pdf
  \begin{equation} \label{eq:pdf_of_matrix_beta=2}
    p_{n,2}(M)dM := \frac{1}{\tilde{C}_{n,2}} e^{-n\Tr(V_2(M)-\A_{n,2}M)} dM.
  \end{equation}
\item
  The set of $2n \times 2n$ self-dual Hermitian matrices, with the \pdf
  \begin{equation} \label{eq:pdf_of_matrix_beta=4}
    \hat{p}_{n,4}(M)dM := \frac{1}{\tilde{C}_{n,4}} e^{-n\Tr(\hat{V}_4(M)-\hat{\A}_{n,4}M)} dM.
  \end{equation}
\end{itemize}
Here for each $\beta = 1,2,4$, $\tilde{C}_{n,\beta}$ is the normalization constant, $V_{\beta}(x)$ (or $\hat{V}_{\beta}(x)$) is a real-valued function which grows fast enough, $\A_{n,\beta}$ (or $\hat{\A}_{n,\beta}$) is a fixed $n \times n$ real symmetric matrix, $n \times n$ Hermitian matrix and $2n \times 2n$ self-dual Hermitian matrix respectively. The function $V_{\beta}(x)$ (or $\hat{V}_{\beta}(x)$) is called the potential function and $\A_{n,\beta}$ (or $\hat{\A}_{n,\beta}$) is called the external source. The rank of the external source $\A_{n,\beta}$ (or $\hat{\A}_{n,\beta}$) is defined to be the number of nonzero eigenvalues of $\A_{n,\beta}$ if $\beta = 1,2$, or half of the nonzero eigenvalues of $\hat{\A}_{n,\beta}$ if $\beta = 4$. These ensembles are called real symmetric, Hermitian and Hermitian self-dual random matrix models with external source respectively. Throughout this paper, we address the three types of ensembles as $\beta$-external source ensembles with $\beta = 1,2,4$ respectively.

Note that in these three external source models, the distributions of eigenvalues of $M$ are unchanged if $\A_{n,\beta}$ (or $\hat{\A}_{n,\beta}$) is changed into $Q\A_{n,\beta} Q^{-1}$ (or $Q\hat{\A}_{n,\beta} Q^{-1}$), where $Q$ is in the orthogonal group $O(n)$, unitary group $U(n)$ and compact symplectic group $Sp(n)$ for $\beta = 1,2,4$ respectively. Since we are only concerned with the distribution of eigenvalues of $M$, we assume $\A_{n,\beta}$ (or $\hat{\A}_{n,\beta}$) to be diagonal without loss of generality. To make our presentation uniform for all values of $\beta$, we let $V(x)$ be a fixed function and 
\begin{equation} \label{eq:defination_of_matrix_A}
\A_n := \diag(a_1, \dots, a_n) 
\end{equation}
be an $n \times n$ diagonal matrix. We assume that $V_{\beta}(x)$ (or $\hat{V}_{\beta}(x)$) and $\A_{n,\beta}$ (or $\hat{\A}_{n,\beta}$) are defined from $V(x)$ and $\A_n$ such that
\begin{gather}
  V_1(x) = \frac{1}{2}V(x), \quad V_2(x)=\hat{V}_4(x) = V(x), \label{eq:relation_of_V_beta_to_V} \\
  \A_{n,1} = \frac{1}{2}\A_n, \quad \A_{n,2} = \A_n, \label{eq:relation_of_A_beta_to_A} \\
  \hat{\A}_{n,4} = \diag(a_1,a_1, a_2,a_2, \dots, a_n,a_n) \label{eq:diag_formula_of_hat_A_n4}
\end{gather}

Writing a $2n \times 2n$ self-dual Hermitian matrix into $2 \times 2$ blocks $\left( \begin{smallmatrix} a_{st} & b_{st} \\ c_{st} & d_{st} \end{smallmatrix} \right)^n_{s,t=1}$, we can express it as a quaternionic Hermitian matrix $(q_{st})^n_{s,t=1}$ whose $s,t$-entry comes from the $s,t$-block by
\begin{equation}
  \begin{pmatrix}
    a+bi & c+di \\
    -c+di & a-bi
  \end{pmatrix}
  = a + bi + cj + dk.
\end{equation}
In the quaternion form, the \pdf\ of the Hermitian self-dual external source model is
\begin{equation} \label{eq:quaternion_form_of_pdf_of_beta=4}
  p_{n,4}(M)dM := \frac{1}{\tilde{C}_{n,4}}
  e^{-n\Re\Tr(V_4(M)-\A_{n,4}M)} dM,
\end{equation}
where
\begin{equation} \label{eq:diag_formula_of_A_n4}
  V_4(x) = 2\hat{V}_4(x) = 2V(x), \quad \A_{n,4} = 2\A_n
\end{equation}
with  $\A_n$ defined in \eqref{eq:defination_of_matrix_A}. The $\A_{n,4}$ defined in \eqref{eq:diag_formula_of_A_n4} corresponds to the $\hat{\A}_{n,4}$ defined in \eqref{eq:diag_formula_of_hat_A_n4}. In Appendix \ref{sec:external_source_model_with_odd_n_and_general_beta_external_source_model} we use the quaternion form \pdf\ \eqref{eq:quaternion_form_of_pdf_of_beta=4} of the $4$-external source model to streamline the derivations for all $\beta$.

In this paper we concentrate on the rank $1$ case, \ie, 
\begin{equation} \label{eq:formula_of_rank_1_A}
  \A_n = \diag(a, \underbrace{0, \dots, 0}_{n-1}).
\end{equation}

\bigskip

Random matrices are powerful tools to simulate Hamiltonians of complex systems. Different types of random matrices, namely the real symmetric (aka orthogonal), the Hermitian (aka unitary) and the Hermitian self-dual (aka symplectic) ensembles are used for physical systems with different properties of time-reversal invariance \cite{Mehta04}. For random matrix models without external source in all the three types of ensembles, \ie, rank $0$ $\beta$-external source models, the distribution of the largest eigenvalue has been studied extensively for all the three $\beta$. If $\beta = 2$, for all real analytic potentials $V_2(x)$ under mild regularity conditions, the largest eigenvalue $\ximax$ approaches $\redge$, the right-end point of the equilibrium measure of $V_2(x)$ (see \eqref{eq:defn_of_redge} below),  with probability $1$ as $n \to \infty$, and the limiting distribution is the GUE Tracy-Widom distribution. See \eg\ \cite{Deift-Kriecherbauer-McLaughlin-Venakides-Zhou99} and \cite{Deift-Gioev07a}. If $\beta = 1$, for real analytic potentials $V_1(x)$ satisfying the ``one-band'' condition (\cf\ Condition \ref{condition:2} in Subsection \ref{subsection: Assumptions_on_Function_V(x)} below) and mild regularity conditions, the largest eigenvalue with probability $1$ approaches $\redge$ as $n \to \infty$, and the limiting distribution is the GOE Tracy-Widom distribution. See \cite{Shcherbina09}. If $\beta = 4$, similar result can be obtained and the limiting distribution is the GSE Tracy-Widom distribution. See \cite{Deift-Gioev07a}.

The random matrix model with external source was proposed by \Brezin\ and Hikami \cite{Brezin-Hikami96}, \cite{Brezin-Hikami98} to simulate complex systems with both random part and deterministic part. Although in all three types of random matrix ensembles the random matrix model with external source can be defined, due to technical reasons, only the Hermitian ($\beta = 2$) type has been studied for general potential functions. See \eg\ \cite{Bleher-Delvaux-Kuijlaars10} and references therein.

In \cite{Baik-Wang10a}, the Hermitian random matrix model with rank $1$ external source was studied for all real analytic potentials $V_2(x)$ under mild regularity conditions. For convex potentials, the universality of phase transition was proved. Let $V_2(x)$ be defined by $V(x)$ as in \eqref{eq:relation_of_V_beta_to_V}. In the rank $1$ $2$-external source model, with probability $1$, as $n \to \infty$
\begin{equation}\label{eq:location_of_largest_eigenvalue_convex_potential}
  \ximax \to
  \begin{cases}
    \redge & \textnormal{if $a \leq \frac{1}{2}V'(\redge)$,} \\
    x_0(a) & \textnormal{if $a > \frac{1}{2}V'(\redge)$,}
  \end{cases}
\end{equation}
where $a$ is the unique nonzero eigenvalue of the external source $\A_{n,2} = \A_n$, and $x_0(a)$ is a continuous increasing function in $a \in (\frac{1}{2}V'(\redge), \infty)$ such that $x_0(a) \to \redge$ as $a \to \frac{1}{2}V'(\redge)$, see \eqref{eq:defn_of_redge} and \eqref{eq:definition_of_x_0}. If $a < \frac{1}{2}V'(\redge)$, the limiting distribution of $\ximax$ is the GUE Tracy-Widom distribution, and if $a > \frac{1}{2}V'(\redge)$ the limiting distribution is Gaussian. For the double scaling $a = \frac{1}{2}V'(\redge) + \frac{\alpha}{n^{1/3}}$, the limiting distribution is the generalized Tracy-Widom distribution. If the potential is not convex, then new phenomena may occur. The ``critical value'' may be less than $\frac{1}{2}V'(\redge)$, and there may be ``secondary critical values''. The largest eigenvalue $\ximax$ may converge to two or more points if $a$ takes such values. The results were also obtained by Bertola, Buckingham, Lee and Pierce in \cite{Bertola-Buckingham-Lee-Pierce11} and \cite{Bertola-Buckingham-Lee-Pierce11a} independently.

For real symmetric and Hermitian self-dual matrix models with external source, known results are limited to special potentials. Let $V_1(x)$ and $\hat{V}_4(x)$ be defined by $V(x)$ as in \eqref{eq:relation_of_V_beta_to_V}, the rank $1$ $1$-external source model with Gaussian potential ($V(x) = x^2$ on the real line) and Laguerre potential ($V(x) = x - c\log(x)$ on half of real line) are studied in \eg\ \cite{Baik-Silverstein06}, \cite{Paul07}, \cite{Feral-Peche07} and \cite{Capitaine-Donati_Martin-Feral09}. The limiting location of the largest eigenvalue is given by formula \eqref{eq:location_of_largest_eigenvalue_convex_potential}, the same as in the corresponding rank $1$ $2$-external source model, where $a$ is the nonzero eigenvalue of $\A_n$ and $\A_{n,1}$ is defined by \eqref{eq:relation_of_A_beta_to_A}. If $a > \frac{1}{2}V'(\redge)$, then the limiting distribution of $\ximax$ is Gaussian, with variance twice of that in the corresponding Hermitian ($\beta = 2$) external source model. If $a < \frac{1}{2}V'(\redge)$, the limiting distribution of $\ximax$ is the GOE Tracy-Widom distribution. The rank $1$ $4$-external source model with Laguerre potential is studied in \cite{Wang08}, where the limiting location of the largest eigenvalue is found to be given by formula \eqref{eq:location_of_largest_eigenvalue_convex_potential}, the same as in the corresponding rank $1$ $2$-external source model, where $a$ is the nonzero eigenvalue of $\A_n$ and $\hat{\A}_{n,4}$ is defined by \eqref{eq:diag_formula_of_hat_A_n4}. If $a > \frac{1}{2}V'(\redge)$, then the limiting distribution of $\ximax$ is Gaussian, with variance half of that in the corresponding Hermitian ($\beta = 2$) external source model. If $a < \frac{1}{2}V'(\redge)$, the limiting distribution of $\ximax$ is the GSE Tracy-Widom distribution. In \cite{Wang08} the limiting distribution of $\ximax$ when $a = \frac{1}{2}V'(\redge)$ is also obtained.

In the recent preprint \cite{Bloemendal-Virag11}, Bloemendal and \Virag\ obtained the limiting distribution of the largest eigenvalue $\ximax$ when the potential is Gaussian or Laguerre, for all $\beta$ and for all $a$. When $a$ is at or near $\frac{1}{2}V'(\redge)$, they described the limiting distribution function of $\ximax$ via the unique solution to a PDE. The recent preprint \cite{Mo11} by Mo indicates a new approach to study the limiting distribution of $\ximax$ in the rank $1$ $1$-external source model with Laguerre potential when $a$ is at or near $\frac{1}{2}V'(\redge)$, see also \cite{Mo11a}. The contour integral formula in \cite[Theorem 1]{Mo11} is equivalent to that of Proposition \ref{prop:contour_rep_of_pdf_of_largest_eigenvalue} in this paper with $\beta = 1$ and Laguerre potential  (\cf\ Remark \ref{rmk:relation_to_Mo}). In \cite{Mo11}, Mo further simplified the integrand in the contour integral formula, (see \cite[Theorem 3]{Mo11},) and he applied it in the asymptotic analysis in \cite{Mo11a} to obtain a result similar to that in \cite{Bloemendal-Virag11}. In this paper, we take a different approach to apply Proposition \ref{prop:contour_rep_of_pdf_of_largest_eigenvalue} in asymptotic analysis. The reader may also compare our paper with the paper \cite{Benaych_Georges-Nadakuditi11} by Benaych-Georges and Nadakuditi, where they considered a different kind of low rank perturbations of large random matrices.

In this paper, we consider the rank $1$ $\beta$-external source models with general potential $V_{\beta}(x)$ (or $\hat{V}_{\beta}(x)$) which are defined by $V(x)$. The conditions satisfied by $V(x)$ will be given in Subsection \ref{subsection: Assumptions_on_Function_V(x)}. We find that the ``critical value'' is independent of $\beta$, and for all $\beta = 1,2,4$ find the limiting location of the largest eigenvalue $\ximax$ when $a$, the nonzero eigenvalue of $\A_n$, is not equal to the critical value. When $a$ is greater than the critical value, we also find the limiting distribution of $\ximax$.

Besides the asymptotic results summarized above, in Appendix \ref{sec:external_source_model_with_odd_n_and_general_beta_external_source_model} we also have an algebraic result: the definition of the $\beta$-external source model with general $\beta > 0$. Here we note that the analytic method presented in this paper can be used to study the rank $1$ $\beta$-external source model with general $\beta$.

\subsection{Assumptions on $V(x)$} \label{subsection: Assumptions_on_Function_V(x)}

Throughout this paper, we assume four conditions on $V(x)$, the function in \eqref{eq:relation_of_V_beta_to_V} and \eqref{eq:diag_formula_of_A_n4}. The first is 
\begin{condition} \label{condition:1}
$V(x)$ is a polynomial of degree $2l$ and with positive leading coefficient.
\end{condition}
To state the other three conditions, we need a few definitions. Recall that if $V(x)$ is regarded as a potential function on $\realR$ itself, there is an equilibrium measure $\mu$ associated to $V(x)$ defined by a certain variational problem. See \eqref{eq:variational_condition_in_support} and \eqref{eq:variational_condition_out_of_support}, and the general references \cite{Saff-Totik97} and \cite{Deift-Kriecherbauer-McLaughlin98}. Since $V(x)$ is a polynomial, $\mu$ is supported on a finite union of intervals. In this paper we assume that $V(x)$ satisfies the ``one-band'' condition:
\begin{condition} \label{condition:2}
The equilibrium measure $\mu$ associated to $V(x)$ is supported on a single interval $J = [b_1, b_2]$.
\end{condition}
For the function $V(x)$ satisfying Conditions \ref{condition:1} and \ref{condition:2}, the equilibrium measure $\mu$ has the form
\begin{equation} \label{eq:formula_of_equilibrium_measure}
  d\mu :=
  \Psi(x) \chi_{J}dx = \frac{1}{2\pi}
  \sqrt{(b_2-x)(x-b_1)}h(x) \chi_{J}dx,
\end{equation}
where $\chi_{J}$ is the indicator function and $h(x)$ is a polynomial of degree $2l-2$. The next condition assumed on $V(x)$ is
\begin{condition} \label{condition:3}
The function $h(x)$ in the formula \eqref{eq:formula_of_equilibrium_measure} of the equilibrium measure $\mu$ of $V(x)$ has only non-real zeros.
\end{condition}
The equilibrium measure $d\mu = \Psi(x)dx$ is characterized by the conditions
\begin{align}
  2\int_J \log \lvert x-s \rvert \Psi(s)ds -V(x) = \ell \quad & \textnormal{for $x \in J$,}  \label{eq:variational_condition_in_support} \\
  2\int_J \log \lvert x-s \rvert \Psi(s)ds -V(x) \leq \ell \quad & \textnormal{for $x \in \realR \setminus J$.}  \label{eq:variational_condition_out_of_support}
\end{align}
The last condition assumed on $V(x)$ is
\begin{condition} \label{condition:4}
  The inequality \eqref{eq:variational_condition_out_of_support} is strict.
\end{condition}

\begin{rmk}
  Conditions \ref{condition:1}--\ref{condition:3} are assumed to apply Proposition \ref{prop:Johansson} in our paper, and they are not used anywhere else in this paper. If Proposition \ref{prop:Johansson} can be proved under weaker conditions, \eg\ the conditions assumed in \cite[Theorem 1]{Kriecherbauer-Shcherbina11} \footnote{Mariya Shcherbina informed the author that Proposition \ref{prop:Johansson} can be proved under the consitions assumed in \cite[Theorem 1]{Kriecherbauer-Shcherbina11} through private communication.}, these conditions can be weakened accordingly. 
\end{rmk}
\begin{rmk}
  Functions $V$ satisfying all Conditions \ref{condition:1}--\ref{condition:4} also satisfy the assumptions of $V$ in \cite[Formulas (6)--(8)]{Baik-Wang10a}. Thus all the results in \cite{Baik-Wang10a} on $V$ can be applied in this paper.
\end{rmk}
\begin{rmk}
  If $V(x)$ is a convex polynomial with positive leading coefficient, $V(x)$ satisfies Conditions \ref{condition:1}--\ref{condition:3} by \cite[Proposition 3.1]{Johansson98}, and it is straightforward to verify that $V(x)$ satisfies Condition \ref{condition:4}.
\end{rmk}

\subsection{Preliminary notations} \label{subsec:preliminary_notations}

To state the results in this paper, we need a few more notations. We follow the notational convention in \cite{Baik-Wang10a} to denote the right edge of the support of the equilibrium measure
\begin{equation} \label{eq:defn_of_redge}
  \textnormal{$\redge := b_2$, the right edge of $J = [b_1, b_2]$, the support of the equilibrium measure $\mu$.}
\end{equation}
The so called $\gfn$-function is defined by
\begin{equation} \label{eq:defn_of_g_function}
\gfn(z) := \int_J \log(z-s) \Psi(s) ds, \quad \textnormal{for $z \in \compC \setminus (-\infty, \redge)$.}
\end{equation}
For $a \in (0, \frac{1}{2}V'(\redge))$, define $c(a)$ as the unique point in $(\redge, \infty)$ such that
\begin{equation}
\gfn'(c(a)) = \int_J \frac{d\mu(x)}{c(a)-x} = a.
\end{equation}
The properties of $\gfn(x)$ used in this paper is summarized below (see \cite[Formula (30)]{Baik-Wang10a}).
\begin{equation}
  \begin{gathered}
    \gfn'(x) > 0, \quad \gfn''(x) < 0 \quad \textnormal{for $x \in (\redge, \infty)$,} \\
    \gfn(\redge) = \frac{V(\redge) + \ell}{2}, \quad \gfn'(\redge) = \frac{V'(\redge)}{2}, \quad \lim_{x \to \infty} \gfn'(x) = 0.
  \end{gathered}
\end{equation}
For $a \geq \frac{1}{2}V'(\redge)$, define $c(a) := \redge$. We also define two auxiliary functions
\begin{align}
\Gfn(z) = \Gfn(z;a) := & \gfn(z) - V(z) + az, \label{eq:defn_of_Gfn} \\
\Hfn(z) = \Hfn(z;a) := & -\gfn(z) + az + \ell, \label{eq:defn_of_Hfn}
\end{align}
for $z \in \compC \setminus (-\infty, \redge)$. Condition \ref{condition:4} of $V$ and the condition \eqref{eq:variational_condition_in_support} imply that for any $a$
\begin{gather}
  \Gfn(\redge;a) = \Hfn(\redge;a) = -\frac{1}{2}V(\redge) + a\redge + \frac{\ell}{2}, \\
  \Gfn(x;a) < \Hfn(x;a) \quad \textnormal{for $x \in (\redge, \infty)$.}
\end{gather}
The convexity of $\gfn(x)$ on $(\redge, \infty)$ yields that for $u > c(a)$,
\begin{equation} \label{eq:positivity_of_H'(u;a)}
  \Hfn'(u;a) = a - \int \frac{d\mu(x)}{u-x} > 0.
\end{equation}

Define the set 
\begin{equation} \label{eq:defn_of_A_v}
  \Av := \{ a \in (0, \infty) \mid  \textnormal{there exists $\bar{x} \in (c(a),\infty)$}  \textnormal{such that $\Gfn(\bar{x};a) > \Hfn(c(a);a)$} \}.
\end{equation} 
It is proved in \cite[Lemma 1.2]{Baik-Wang10a} that $\Av$ is an open, semi-infinite interval. From $\Av$ we define
\begin{equation} \label{eq:definition_of_acc}
\acc = \acc(V) := \inf \Av.
\end{equation}
It is also proved in \cite[Lemma 1.2]{Baik-Wang10a} that $\acc > 0$.

Let 
\begin{equation} \label{defn_of_G_max}
\Gfn_{\max}(a) := \max_{x \in[c(a),\infty)} \Gfn(x;a). 
\end{equation}
Define the discrete set
\begin{equation} \label{eq:definition_of_J_V}
  \begin{split}
    \J := \{ a \in [\acc, \infty) \mid & \textnormal{$\Gfn_{\max}(a)$ is attained at more than one point} \}.
  \end{split}
\end{equation}
For $a > \acc$ and $a \not\in \J$, there is a unique $x_0(a) \in (c(a), \infty)$ such that (\cf\ \cite[Lemma 1.3]{Baik-Wang10a})
\begin{equation} \label{eq:definition_of_x_0}
\Gfn_{\max}(a) = \Gfn(x_0(a);a).
\end{equation}
For $a > \acc$ and $a \in \J$, there are $r \geq 2$ and  $c(a) < x_1(a) < x_2(a) < \dots < x_r(a)$ such that
\begin{equation} \label{eq:definition_of_x_1_to_x_r}
  \Gfn_{\max}(a) = \Gfn(x_1(a);a) = \dots = \Gfn(x_r(a);a).
\end{equation}
We define the set of secondary critical values as $\J \setminus \{ \acc \}$ (\cf\ \cite[Definition 1.3]{Baik-Wang10a} ).

\begin{rmk}
  If the potential $V$ is convex for $x \geq \redge$, $\J = \emptyset$. See \cite[Remark 1.2]{Baik-Wang10a}.
\end{rmk}

\subsection{Statement of main results}

Let $V(x)$ be a function that satisfies Conditions \ref{condition:1}--\ref{condition:4}. For any $n$ and $\beta = 1,2,4$, let the $n$-dimensional $\beta$-external source models be defined by \pdf{s} \eqref{eq:pdf_of_matrix_beta=1}, \eqref{eq:pdf_of_matrix_beta=2} and \eqref{eq:pdf_of_matrix_beta=4} respectively, with potentials $V_{\beta}(x)$ (or $\hat{V}_{\beta}(x)$) given by \eqref{eq:relation_of_V_beta_to_V} and external sources $\A_{n,\beta}$ (or $\hat{\A}_{n,\beta}$) given by  \eqref{eq:defination_of_matrix_A}, \eqref{eq:relation_of_A_beta_to_A} and \eqref{eq:diag_formula_of_hat_A_n4}. We assume that $\A_n$ has only one nonzero eigenvalue $a$, as in \eqref{eq:formula_of_rank_1_A}. In each $n$-dimensional $\beta$-external source model, let $\ximax$ be the largest eigenvalue of the random matrix. The theorems below are stated uniformly for all $\beta$-external source models ($\beta = 1,2,4$). In the case $\beta = 1$, we assume that the dimension $n$ is even. For $\beta = 1$ and $n$ is odd, the theorems below also hold, and we discuss it briefly in Appendix \ref{sec:external_source_model_with_odd_n_and_general_beta_external_source_model}. 
First we show the limiting location of the largest eigenvalue.
\begin{thm} \label{thm:limiting_position}
  The following hold for each fixed $a$ as $n \to \infty$:
  \begin{enumerate}[label=(\alph*)]
  \item \label{enu:thm:limiting_position:1}
    If $0 < a < \acc$, or $a = \acc = \frac{1}{2}V'(\redge)$ and $\acc \notin \mathcal{J}_V$, $\xi_{\max}(n) \to \redge$ with probability $1$.

  \item \label{enu:thm:limiting_position:2} 
    If $a > \acc$ and $\acc \not\in \J$, then $\xi_{\max}(n) \to x_0(a)$ with probability $1$, where $x_0(a)$ is defined in \eqref{eq:definition_of_x_0}.

  \item \label{enu:thm:limiting_position:3} 
    If $a > \acc$ and $a \in \J$, then there exist $r \geq 2$ and $x_1(a)$, \dots, $x_r(a)$ defined in \eqref{eq:definition_of_x_1_to_x_r}. Under the assumption that $\Gfn''(x_j(a)) \neq 0$ for all $j=1, \dots, r$, then $\xi_{\max}(n)$ converges to $x_j(a)$ with nonzero probability $p_{j,\beta}(0)$ for $j=1, \dots, r$. Here $p_{j,\beta}(0)$ are defined in \eqref{eq:defn_of_p_i_two_normal_points} and $\sum^r_{j=1} p_{j,\beta}(0) = 1$.
  \end{enumerate}
\end{thm}

\begin{rmk}
  If $a < 0$, Theorem \ref{thm:limiting_position}\ref{enu:thm:limiting_position:1} still holds, and the method of proof is similar to that in the $0 < a < \acc$ case. Since when $a < 0$ there is no interesting phase transition phenomenon for the distribution of the largest eigenvalue (while there is a similar one of the smallest eigenvalue) and the proof is long and parallel to the $a > 0$ case, we skip further discussions about the $a < 0$ case.
\end{rmk}

If $a > \acc$, we have the limiting distribution of the largest eigenvalue. If $a$ is not at or near secondary critical values, we have the following result that strengthens Theorem \ref{thm:limiting_position}\ref{enu:thm:limiting_position:2}.
\begin{thm} \label{thm:limiting_distr_nondegenerate}
  The following hold for $a > \acc$ and $a \not\in \J$ as $n \to \infty$.
  \begin{enumerate}[label=(\alph*)]
  \item \label{enu:thm:limiting_distr_nondegenerate:1} If $\Gfn''(x_0(a);a) \neq 0$, then for any $T \in \realR$
    \begin{equation}
      \lim_{n \to \infty} \Prob \left( \xi_{\max}(n) < x_0(a) + \frac{T}{\sqrt{-\frac{\beta}{2}\Gfn''(x_0(a);a)n}} \right) = \erf(T),
    \end{equation}
    where $\erf(T) := \frac{1}{\sqrt{2\pi}} \int^T_{-\infty}  e^{-\frac{1}{2}\xi^2} d\xi$ denotes the cumulative distribution function of standard normal distribution.

  \item \label{enu:thm:limiting_distr_nondegenerate:2} If
    $\Gfn^{(2k)}(x_0(a);a) \neq 0$ and $\Gfn^{(j)}(x_0(a);a)= 0$ for $j = 1, \dots, 2k-1$ where $k > 1$, then for any $T \in \realR$
    \begin{equation}
      \lim_{n \to \infty} \Prob \left( \xi_{\max}(n) < x_0(a) + \left( \frac{-\frac{\beta}{2}\Gfn^{(2k)}(x_0(a);a)n}{(2k)!} \right)^{-1/(2k)} T \right) = \frac{\int^{T}_{-\infty} e^{-\xi^{2k}} d\xi}{\int^{\infty}_{-\infty}  e^{-\xi^{2k}} d\xi}.
    \end{equation}
  \end{enumerate}
\end{thm}

If $a>\acc$ is at or near a secondary critical value, we have the following result that shows the double scaling case and strengthens Theorem \ref{thm:limiting_position}\ref{enu:thm:limiting_position:3}.
\begin{thm} \label{thm:two_normal_max_points}
  Suppose that $a_0 > \acc$ and $a_0 \in \J$. Assume that $\Gfn(x;a_0)$ attains its maximum at $r \geq 2$ points $x_1(a_0) < x_2(a_0) < \dots < x_r(a_0)$ in $(c(a_0), \infty)$, and $\Gfn''(x_i(a_0);a_0) \neq 0$ for all $i = 1, 2, \dots, r$, then for
  \begin{equation}
    a = a_0 + \frac{\alpha}{n},
  \end{equation}
  where $\alpha$ is in a compact subset of $\realR$, we have
  \begin{equation}
    \lim_{n \to \infty} \Prob \left( \xi_{\max}(n) < x_i(a_0) + \frac{T}{\sqrt{-\frac{\beta}{2}\Gfn''(x_i(a_0))n}} \right) = \left( \sum^{i-1}_{j=1}p_{j,\beta}(\alpha) \right) + p_{i,\beta}(\alpha)\erf(T),
  \end{equation}
  where $p_{i,\beta}(\alpha)$ ($i=1, \dots, r$) are defined in \eqref{eq:defn_of_p_i_two_normal_points} and $\sum^r_{j=1} p_{j,\beta}(\alpha) = 1$. Furthermore, $p_{r,\beta}(\alpha) \to 1$ as $\alpha \to \infty$ and $p_{1,\beta}(\alpha) \to 1$ as $\alpha \to -\infty$.
\end{thm}
\begin{rmk}
  The phenomenon of Theorem \ref{thm:two_normal_max_points} occurs for some quartic potential $V$ that satisfies Conditions \ref{condition:1}--\ref{condition:4}. For example, $V(x) = 0.02093x^4 - 0.16736x^3 + 0.37448x^2 + 0.11418x$.
\end{rmk}
In the case that $a > \acc$, $a \in \J$ and $\Gfn''(x_j(a);a) = 0$ at at least one maximizer $x_j(a)$ of $\Gfn(x;a)$ in $(c(a), \infty)$, we show hereafter an example when the number of maximizers of $\Gfn(x;a)$ in $(c(a), \infty)$ is $r = 2$. The result for general case is similar.
\begin{thm} \label{thm:one_abnormal_max_points}
  Suppose that $a_0 > \acc$ and $a_0 \in \J$. Assume that $\Gfn(x;a_0)$ attains its maximum at two points $x_1(a_0) < x_2(a_0)$ in $(c(a_0), \infty)$, with $\Gfn''(x_1(a_0);a_0) \neq 0$, $\Gfn^{(2k)}(x_2(a_0);a_0) \neq 0$ and $\Gfn^{(j)}(x_2(a_0);a_0) = 0$ for $j = 1, \dots, 2k-1$. Then for
  \begin{equation} \label{eq:defn_of_q_and_a}
    a = a_0 -q_{\beta}\frac{\log n}{n} + \frac{\alpha}{n}, \quad \textnormal{where} \quad q_{\beta} := \frac{2}{\beta}\frac{\frac{1}{2} - \frac{1}{2k}}{x_2(a_0) - x_1(a_0)},
  \end{equation}
  and $\alpha$ is in a compact subset of $\realR$, we have
  \begin{align}
    \lim_{n \to \infty} \Prob \left( \xi_{\max}(n) < x_1(a_0) + \frac{T}{\sqrt{-\frac{\beta}{2}\Gfn''(x_1(a_0))n}} \right) = & \tilde{p}_{1,\beta}(\alpha)\erf(T), \\
    \lim_{n \to \infty} \Prob \left( \xi_{\max}(n) < x_2(a_0) + \left( \frac{-\frac{\beta}{2} \Gfn^{(2k)}(x_2(a_0);a_0)n}{(2k)!} \right)^{-1/(2k)}T \right) = & \tilde{p}_{1,\beta}(\alpha) + \tilde{p}_{2,\beta}(\alpha) \frac{\int^T_{-\infty} e^{-x^{2k}} dx}{\int^{\infty}_{-\infty} e^{-x^{2k}} dx},
  \end{align}
  where $\tilde{p}_{1,\beta}(\alpha)$ and $\tilde{p}_{2,\beta}(\alpha)$ are defined in \eqref{eq:defn_of_tilde_p}, and $\tilde{p}_{1,\beta}(\alpha) + \tilde{p}_{2,\beta}(\alpha) = 1$. Furthermore, $\tilde{p}_{2,\beta}(\alpha) \to 1$ as $\alpha \to \infty$ and $\tilde{p}_{1,\beta}(\alpha) \to 1$ as $\alpha \to -\infty$.
\end{thm}

\begin{rmk}
  When $\beta = 2$, the probabilities $p_{j,2}(\alpha)$ and $\tilde{p}_{j,2}(\alpha)$ should agree with the $p^{(j)}_{1,n}(\alpha)$ in \cite[Formula (52)]{Baik-Wang10a} and the $p^{(j)}_{1,n}(\alpha)$ in \cite[Formula (63)]{Baik-Wang10a} respectively. It is not obvious that they are the same, and we give the proof in Appendix \ref{sec:simplification_of_M_2(x)}.
\end{rmk}

The limiting distribution of the largest eigenvalue when $a \leq \acc$, as well as the limiting location of the largest eigenvalue when $a$ is at or near $\acc < \frac{1}{2}V'(\redge)$, will be analyzed in a subsequent paper.

\bigskip

The paper is organized as follows. In Section \ref{sec:the_pdf_of_the_largest_eigenvalue}, we calculate the limiting \pdf\ of the largest eigenvalue in the rank $1$ $\beta$-external source model as $n \to \infty$, based on Proposition \ref{cor:outlier_micro}. In Section \ref{sec:proof_of_theorems}, we prove Theorems \ref{thm:limiting_position}, \ref{thm:limiting_distr_nondegenerate} \ref{thm:two_normal_max_points} and \ref{thm:one_abnormal_max_points}. The proof of Proposition \ref{cor:outlier_micro} is in Section \ref{sec:proof_of_corollary_of_Johansson}. 

The starting point of the asymptotic analysis in this paper is Proposition \ref{prop:contour_rep_of_pdf_of_largest_eigenvalue}, the contour integral formula of the largest eigenvalue $\ximax$. Since its proof is combinatorial, we postpone it to Appendix \ref{sec:external_source_model_with_odd_n_and_general_beta_external_source_model}. In this appendix we also propose a definition of the $\beta$-external source model for any $\beta > 0$. In Appendix \ref{sec:simplification_of_M_2(x)} we show that the results in this paper agree with those in \cite{Baik-Wang10a} when $\beta = 2$.

\section{The \pdf\ of the largest eigenvalue} \label{sec:the_pdf_of_the_largest_eigenvalue}

In this section we compute $f_{\ximax}$, the \pdf\ of the largest eigenvalue $\ximax$ in the $n$-dimensional $\beta$-external source model with rank $1$, as $n \to \infty$. For $\beta = 1$, $n$ is assumed to be even. We also assume that the only nonzero eigenvalue of the external source matrix $\A_n$ is $a > 0$. Recall that $J$ is the support of the equilibrium measure $\mu$ associated to $V(x)$, $\redge$ is the right end of $J$, and $c(a)$ is defined in Subsection \ref{subsec:preliminary_notations}. In this section we compute/estimate $f_{\ximax}(u)$ for all real $u$. To be concrete, let $\bfepsilon$ be a small enough positive constant. In Subsection \ref{subsec:e<u<c(a)}, assuming that $c(a) > \redge$, we compute $f_{\ximax}(u)$ for $u \in [\redge+\bfepsilon, c(a)-\bfepsilon]$ up to a constant factor $\bfC_{n,\beta}$. In Subsection \ref{subsec:u>c(a)}, we compute $f_{\ximax}(u)$ for $u \in [c(a)+\bfepsilon, \redge+\bfepsilon^{-1}]$ up to the constant factor $\bfC_{n,\beta}$. For $u \in (\redge+\bfepsilon^{-1}, \infty)$, $u \in (-\infty, \redge+\bfepsilon)$ and $u \in (c(a)-\bfepsilon, c(a)+\bfepsilon)$ in case $c(a) > \redge$, we give an estimate of $f_{\ximax}(u)$ in Subsection \ref{subsec:other_u}. Note that throughout this section, $u$ is always a real number.

To facilitate the computation of $f_{\ximax}(u)$, we define some notations. For any $m$, define the probability measure on $\realR^m$
\begin{equation} \label{eq:defn_of_measure_mu_{n-1,beta}}
  d\mu_{m,\beta}(x_1, \dots, x_m) := \frac{1}{\constZ_{m,\beta}}
  \lvert \Delta(x_1, \dots, x_m) \rvert^{\beta} \prod^m_{j=1}
  e^{-\frac{\beta}{2}mV(x_j)} dx_1 \dots dx_m,
\end{equation}
where $\constZ_{m,\beta}$ is the normalization constant. Suppose $F(x_1, \dots, x_m)$ is an integrable function with respect to the measure $\mu_{m,\beta}$ defined in \eqref{eq:defn_of_measure_mu_{n-1,beta}}, define the expectation of $F$ with respect to $\mu_{m,\beta}$ by
\begin{equation} \label{eq:defn_of_E_n-1,beta}
\E_{m,\beta}(F(x_1, \dots, x_m)) :=  \int_{\realR^m} F(x_1, \dots, x_m) d\mu_{m,\beta}(x_1 \dots x_m).
\end{equation}
For $u \in \realR$ and $w \in \compC \setminus (-\infty, u)$, define the functions in $u$ and $w$
\begin{align}
  Z_{m,\beta}(u,w) := & \E_{m,\beta} \left( P_{m,\beta}(x_1, \dots, x_m; u,w) \prod^m_{j=1} \chi_{(-\infty,u)}(x_j) \right), \label{eq:definition_of_Z_n-1_beta} \\
  \hat{Z}_{m,\beta}(u,w) := & \E_{m,\beta} \left( \lvert P_{m,\beta}(x_1, \dots, x_m; u,w) \rvert \prod^m_{j=1} \chi_{(-\infty,u)}(x_j) \right), \label{eq:defn_of_Z_hat_(uwc)}
\end{align}
where
\begin{equation} \label{eq:definition_of_P_n-1_beta}
P_{m,\beta}(x_1, \dots, x_m;u,w) := \prod^m_{j=1} \frac{e^{-\frac{\beta}{2}V(x_j)}(u-x_j)^{\beta}}{(w-x_j)^{\beta/2}},
\end{equation}
and we take the principal branch of $(w-x_j)^{\beta/2}$ for $w \in \compC \setminus (-\infty, x_j)$. For $x < u$ and $w \in \compC \setminus (-\infty, u)$, define the function in $x$ with parameter $u$ and $w$
\begin{equation} \label{eq:def_of_tilde_p(xu)}
p(x;u,w) := -V(x)+2\log(u-x) -\log\lvert w-x \rvert.
\end{equation}
We have
\begin{equation} \label{eq:alternative_defn_of_Z_hat_(uwc)}
  \hat{Z}_{m,\beta}(u,w) = \E_{m,\beta} \left( e^{\frac{\beta}{2} \sum^m_{j=1} p(x_j;u,w)} \prod^m_{j=1} \chi_{(-\infty,u)}(x_j) \right).
\end{equation}
If $w=u$ we denote
\begin{equation} \label{eq:def_of_p(xu)}
p(x;u) := p(x;u,u) = -V(x)+\log(u-x).
\end{equation}
Then we can state the technical tool in the asymptotic analysis of this section:
\begin{prop} \label{cor:outlier_micro}
  Let $u > \redge$ and $w \in \compC \setminus (-\infty, u]$.

  \begin{enumerate}[label=(\alph*)]
  \item \label{enu:cor:outlier_micro:a} 
    Suppose
    \begin{equation} \label{eq:condition_enu:cor:outlier_micro:a}
      w = u+\frac{z}{n}, \quad \textnormal{where $z = s+it$ is in a compact subset of $\compC \setminus (-\infty,0]$,}
    \end{equation}
    we have
    \begin{equation} \label{eq:enu:cor:outlier_micro:a}
      Z_{n-1,\beta}(u,w) = e^{-\frac{\beta z}{2} \int
        \frac{d\mu(x)}{u-x}} R_{\beta}(u) \exp \left[ \frac{\beta}{2}n
        \int p(x;u)d\mu(x) \right] (1+o(1)),
    \end{equation}
    where $R_{\beta}(u)$ is defined in \eqref{eq:defn_of_R_beta(u)}.

  \item \label{enu:cor:outlier_micro:c}

    Suppose
    \begin{equation} \label{eq:condition_enu:cor:outlier_micro:c}
      w = w_0 + i\frac{t}{\sqrt{n}}, \quad \textnormal{where $w_0 > u$ and $t$ is in a compact subset of $\realR$,}
    \end{equation}
    we have
    \begin{equation} \label{eq:enu:cor:outlier_micro:c}
      Z_{n-1,\beta}(u,w) = e^{-\frac{\beta t^2}{4} \int \frac{d\mu(x)}{(w_0-x)^2} -i\frac{\beta t}{2} \sqrt{n} \int \frac{d\mu(x)}{w_0-x}} R_{\beta}(u,w_0) \exp \left[ \frac{\beta}{2}n \int p(x;u,w_0) d\mu(x) \right] (1+o(1)),
    \end{equation}
    where $R_{\beta}(u,w_0)$ is defined in \eqref{eq:defn_of_R_beta(u,w)}.

  \item \label{enu:cor:outlier_micro:b} 
    Let $\epsilon$ be a small positive constant. For all $w$ such that $\dist(w, (-\infty, \redge]) \geq \epsilon$ and $ \lvert w \rvert \leq \epsilon^{-1}$,
    \begin{equation} \label{eq:enu:cor:outlier_micro:b}
      \lvert Z_{n-1,\beta}(u,w) \lvert \leq \hat{Z}_{n-1,\beta}(u,w) = \exp \left[ \frac{\beta}{2}n \int p(x;u,w) d\mu(x) \right] O(1),
    \end{equation}
    where the factor $O(1)$ is bounded uniformly in $w$.
  \end{enumerate}
\end{prop}
This proposition is a corollary of a theorem of Johansson \cite[Theorem 2.4]{Johansson98}, and we put off its proof to Section \ref{sec:proof_of_corollary_of_Johansson}.

\bigskip

For the asymptotic analysis in this section, we define four types of contours: $\Sigma^x_{s_1,s_2}$, $\Pi^x_s$, $\Gamma^x_s$ and $I^x_s$, where $x$ is a real parameter and $s_1$, $s_2$ and $s$ are positive parameters. We assume $s_2 > s_1/\sqrt{2}$ for $\Sigma^x_{s_1,s_2}$ and allow $s = \infty$ in $\Gamma^x_s$. The contours $\Pi^x_s$ and $\Sigma^x_{s_1,s_2}$ will be used in Subsections \ref{subsec:e<u<c(a)} and \ref{subsec:u>c(a)} respectively. The contours $I^x_s$ and $\Gamma^x_s$ represent the local parts of $\Pi^x_s$ and $\Sigma^x_{s_1,s_2}$ around the point $x$ respectively, which will turn out to be the saddle point in the asymptotic analysis.
\begin{multline} \label{eq:parametrization_of_Sigma}
  \Sigma^x_{s_1,s_2} = \{ w(t) \mid t \in \realR \}, \quad \textnormal{where} \\
  w(t) =
  \begin{cases}
    x + e^{\frac{3\pi i}{4}}t & \textnormal{if $0 \leq t \leq s_1$,} \\
    x + e^{\frac{3\pi i}{4}}s_1 + i(t-s_1) & \textnormal{if $s_1 \leq t \leq s_2 + (1 - \frac{1}{\sqrt{2}})s_1$,} \\
    x - (\sqrt{2}-1)s_1 + (1+i)s_2 - t & \textnormal{if $t \geq s_2 + (1 - \frac{1}{\sqrt{2}})s_1$,} \\
    \overline{w(-t)} & \textnormal{if $t \leq 0$.}
  \end{cases}
\end{multline}
\begin{equation} \label{eq:parametrization_of_Pi}
  \Pi^x_s = \{ w(t) \mid t \in \realR \} \quad \textnormal{where} \quad w(t) =
  \begin{cases}
    x + it & \textnormal{if $0 \leq t \leq s$,} \\
    x + (1+i)s - t & \textnormal{if $t \geq s$,} \\
    \overline{w(-t)} & \textnormal{if $t \leq 0$.}
  \end{cases}
\end{equation}
\begin{align}
\Gamma^x_s = & \{ w(t) \mid t \in [-s,s] \} \quad \textnormal{where} \quad w(t) = x + (it - \lvert t \rvert)/\sqrt{2}. \\
I^x_s = & \{ w(t) \mid t \in [-s,s] \} \quad \textnormal{where} \quad w(t) = x + it.
\end{align}
See Figures \ref{figure_contour_Sigma}, \ref{figure_contour_Pi}, \ref{figure_contour_Gamma} and \ref{figure_contour_I} for these contours. For any real number $r$, we define
\begin{equation}
\Sigma^x_{s_1,s_2}(r) = \{ z \in \Sigma^x_{s_1,s_2} \mid \Re z \geq r \}, \quad \Pi^x_s(r) = \{ z \in \Pi^x_s \mid \Re z \geq r \}.
\end{equation}

\begin{figure}[htb]
\begin{minipage}[b]{0.45\textwidth}
\begin{centering}
\includegraphics{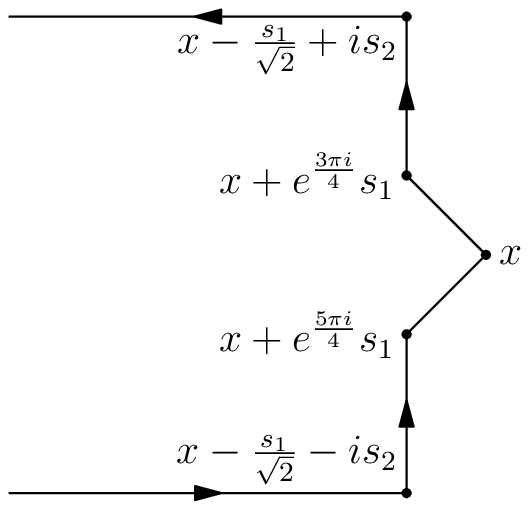}
\caption{The contour $\Sigma^x_{s_1,s_2}$} \label{figure_contour_Sigma}
\end{centering}
\end{minipage}
\begin{minipage}[b]{0.45\textwidth}
\begin{centering}
\includegraphics{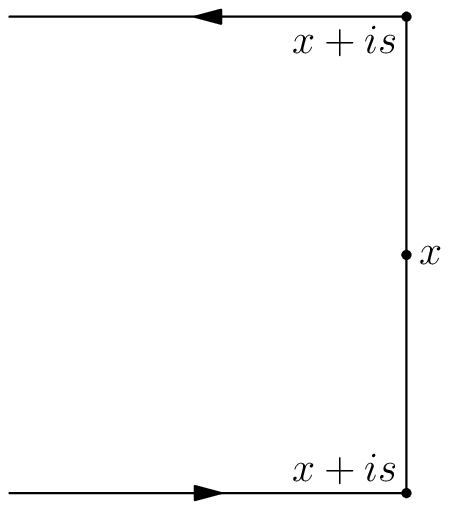}
\caption{The contour $\Pi^x_s$} \label{figure_contour_Pi}
\end{centering}
\end{minipage} \\
\begin{minipage}[b]{0.45\textwidth}
\begin{centering}
\includegraphics{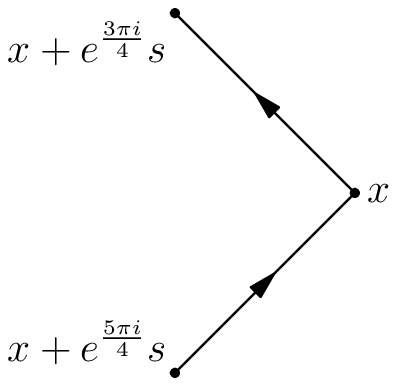}
\caption{The contour $\Gamma^x_s$} \label{figure_contour_Gamma}
\end{centering}
\end{minipage}
\begin{minipage}[b]{0.45\textwidth}
\begin{centering}
\includegraphics{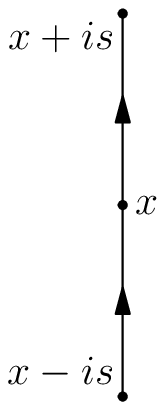}
\caption{The contour $I^x_s$} \label{figure_contour_I}
\end{centering}
\end{minipage}
\end{figure}

\bigskip

The asymptotic analysis in this section is based on the contour integral representation of the \pdf\ of the largest eigenvalue $\ximax$:
\begin{prop} \label{prop:contour_rep_of_pdf_of_largest_eigenvalue}
  Let $\ximax$ be the largest eigenvalue in the $n$-dimensional rank $1$ $\beta$-external source model for $\beta = 1,2,4$, where the potential $V_{\beta}(x)$ (or $\hat{V}_{\beta}(x)$) is defined by \eqref{eq:relation_of_V_beta_to_V} from $V(x)$, and the external source matrix $\A_{n,\beta}$ (or $\hat{\A}_{n,\beta}$) is defined by \eqref{eq:relation_of_A_beta_to_A} (or \eqref{eq:diag_formula_of_hat_A_n4}) from $\A_n$ in \eqref{eq:formula_of_rank_1_A} with $a > 0$. Then for any integer $n$ if $\beta = 2,4$ and for even integer $n$ if $\beta = 1$,
  \begin{equation} \label{eq:pdf_of_largest_eigen_supercritical}
    f_{\ximax}(u) = \frac{\hat{C}_{n,\beta}}{2\pi i} e^{-\frac{\beta}{2}nV(u)} \oint_{\Cont} Z_{n-1,\beta}(u,w)  \frac{e^{\frac{\beta}{2}anw}}{(w-u)^{\beta/2}} dw,
  \end{equation}
  where $\hat{C}_{n,\beta}$ is a constant, $Z_{n-1,\beta}(u,w)$ is defined in \eqref{eq:definition_of_Z_n-1_beta}, and $\Cont$ is either $\Sigma^x_{s_1,s_2}$ defined in \eqref{eq:parametrization_of_Sigma} or $\Pi^x_s$ defined in \eqref{eq:parametrization_of_Pi} with $x > u$.
\end{prop}

The proof of Proposition \ref{prop:contour_rep_of_pdf_of_largest_eigenvalue} is in Appendix \ref{sec:external_source_model_with_odd_n_and_general_beta_external_source_model}.

\subsection{Computation of $f_{\ximax}(u)$ when $u \in [\redge+\bfepsilon, c(a)-\bfepsilon]$} \label{subsec:e<u<c(a)}

Assuming $\redge+\bfepsilon \leq u \leq c(a)-\bfepsilon$, we use the contour integral formula \eqref{eq:pdf_of_largest_eigen_supercritical} of $f_{\ximax}(u)$ and take the contour $\Cont$ in \eqref{eq:pdf_of_largest_eigen_supercritical} as $\Pi^{c(a)}_{\pi_1}$, which is defined in \eqref{eq:parametrization_of_Pi}. Here $\pi_1$ is a large enough parameter such that the inequality \eqref{eq:inequality_of_second_decreasing_horizontal_real} holds. For $w \in \Pi^{c(a)}_{\pi_1}$, we parametrize it by $w=w(t)$ as in \eqref{eq:parametrization_of_Pi} with $x=c(a)$ and $s = \pi_1$.

Let $L$ be a positive number. For sufficiently large $n$, the contour $I^{c(a)}_{L/\sqrt{n}}$ is part of $\Pi^{c(a)}_{\pi_1}$. By Proposition \ref{cor:outlier_micro}\ref{enu:cor:outlier_micro:c}, for $w = c(a) + it/\sqrt{n} \in I^{c(a)}_{L/\sqrt{n}}$, the integrand in the contour integral of \eqref{eq:pdf_of_largest_eigen_supercritical} satisfies
\begin{equation} \label{eq:asy_formula_of_integrand_around_c(a)}
  Z_{n-1,\beta}(u,w)  \frac{e^{\frac{\beta}{2}anw}}{(w-u)^{\beta/2}} = e^{-\frac{\beta t^2}{4} \int \frac{d\mu(x)}{(c(a)-x)^2}} e^{\frac{\beta}{2}n \left( \int p(x;u,c(a)) d\mu(x) + ac(a) \right)} \frac{R_{\beta}(u,c(a))}{(c(a)-u)^{\beta/2}} (1+o(1)).
\end{equation}
Using the asymptotic formula \eqref{eq:asy_formula_of_integrand_around_c(a)}, we have the result 
\begin{equation} \label{eq:estimate_of_integral_essential_b_2<u<c(a)}
  \frac{1}{2\pi i} \int_{I^{c(a)}_{L/\sqrt{n}}} Z_{n-1,\beta}(u,w)  \frac{e^{\frac{\beta}{2}anw}}{(w-u)^{\beta/2}} dw = \frac{1}{\sqrt{n}} e^{\frac{\beta}{2}n \left( \int p(x;u,c(a)) d\mu(x) + ac(a) \right)} \tilde{\M}_{\beta}(u) (1+\epsilon_1(L,n)),
\end{equation}
where
\begin{equation}
  \tilde{\M}_{\beta}(u) = \frac{R_{\beta}(u,c(a))}{\sqrt{\pi \beta \int \frac{d\mu(x)}{(c(a)-x)^2}} (c(a)-u)^{\beta/2}}
\end{equation}
and $\epsilon_1(L,n)$ is small if $L$ and $n$ are large. To be precise, for all $\epsilon > 0$, there is an $L_1 > 0$ such that for all $L > L_1$, $\lvert \epsilon_1(L,n) \rvert < \epsilon$ for $n$ large enough. 

For $w$ in a bounded subset of $\Pi^{c(a)}_{\pi_1} \setminus I^{c(a)}_{L/\sqrt{n}}$, we use Proposition \ref{cor:outlier_micro}\ref{enu:cor:outlier_micro:b} and estimate the integrand of the contour integral of \eqref{eq:pdf_of_largest_eigen_supercritical}
\begin{equation} \label{eq:estimate_of_absolute_value_of_integrand}
\left\lvert Z_{n-1,\beta}(u,w)  \frac{e^{\frac{\beta}{2}anw}}{(w-u)^{\beta/2}} \right\rvert = e^{\frac{\beta}{2}n \tildeP(u,w)} \frac{1}{\lvert w-u \rvert^{\beta/2}} O(1),
\end{equation}
where 
\begin{equation} \label{eq:defn_of_tildeP}
  \tildeP(u,w) = \int
  p(x;u,w) d\mu(x) + a\Re w,
\end{equation}
and the $O(1)$ factor is uniformly bounded. Below we show that $\tildeP(u,w(t))$ decreases fast enough as $t$ increases and $t \geq 0$. By the symmetry of the contour $\Sigma^{u+n^{-1}}_{\sigma_1, \sigma_2}$ about the real axis, we see that $\tildeP(u,w(t))$ decreases fast enough as $t$ decreases and $t \leq 0$.

If  $w(t) \in \Pi^{c(a)}_{\pi_1} \setminus I^{c(a)}_{L/\sqrt{n}}$ with $t \geq \pi_1$, \ie, $w$ is in the ray from $c(a)+i\pi_1$ to $-\infty$, we have
\begin{equation} \label{eq:derivative_of_tildeP_new_crrection}
  \frac{d \tildeP(u,w(t))}{dt} = \int \frac{c(a)-x+\pi_1-t}{\lvert w(t) -x \rvert^2} d\mu(x) - a.
\end{equation}
If $\pi_1$ is large enough, we have that for all $t \in [\pi_1, \infty)$ and all $x \in J$, there exists $c_1 > 0$ such that
\begin{equation} \label{eq:inequality_of_second_decreasing_horizontal_real}
  \frac{c(a)-x+\pi_1-t}{\lvert w(t) -x \rvert^2} < \frac{1}{c(a)-x} - c_1.
\end{equation}
Hence we have
\begin{equation} \label{eq:inequality_of_second_decreasing_horizontal}
  \frac{d \tildeP(u,w(t))}{dt} < \int \frac{d\mu(x)}{c(a)-x} - c_1 + a = -c_1.
\end{equation}
For $0 < t \leq \pi_1$, like \eqref{eq:derivative_of_tildeP_new_crrection} and \eqref{eq:inequality_of_second_decreasing_horizontal_real}, we have
\begin{equation} \label{eq:inequality_of_second_decreasing_vertical}
  \frac{d \tildeP(u,w(t))}{dt} = \int \frac{t}{\lvert w(t) - x \rvert^2} d\mu(x) < -c_2t < 0,
\end{equation}
where $c_2$ is a positive constant depending on $\pi_1$.

Let $\tilde{L}_1$ be a large enough positive number such that $-\tilde{L}_1 < c(a)$ and the inequality \eqref{eq:inequality_for_tildeL_1_small_u} holds. By Proposition \ref{cor:outlier_micro}\ref{enu:cor:outlier_micro:b} and inequalities \eqref{eq:inequality_of_second_decreasing_horizontal} and \eqref{eq:inequality_of_second_decreasing_vertical}, 
\begin{equation} \label{eq:integral_over_Pi_middle}
  \left\lvert \frac{1}{2\pi i} \int_{(\Pi^{c(a)}_{\pi_1}(-\tilde{L}_1) \setminus I^{c(a)}_{L/\sqrt{n}}) \cap \compC_+} Z_{n-1,\beta}(u,w)  \frac{e^{\frac{\beta}{2}anw}}{(w-u)^{\beta/2}} dw \right\rvert = \frac{1}{\sqrt{n}} \exp \left[ \frac{\beta}{2}n \tildeP(u, w(L/\sqrt{n})) \right] O(1),
\end{equation}
where the factor $O(1)$ is bounded uniformly in $L$. Substituting the Taylor expansion \eqref{eq:relation_between_tilde_p_and_tilde_p_w_0} of $p(x;u,w)$ into \eqref{eq:integral_over_Pi_middle}, we find 
\begin{equation} \label{eq:another_formula_of_tildeP}
e^{\frac{\beta}{2} n\tildeP(u, w(L/\sqrt{n}))} = e^{\frac{\beta}{2} n \left( \int p(x;u,c(a)) d\mu(x) + ac(a) \right)} \left( e^{-\frac{\beta L^2}{4} \int \frac{d\mu(x)}{(c(a)-x)^2}} + o(1) \right).
\end{equation}
We write like \eqref{eq:estimate_of_integral_essential_b_2<u<c(a)}
\begin{equation} \label{eq:estimate_of_integral_middle_b_2<u<c(a)}
\frac{1}{2\pi i} \int_{(\Pi^{c(a)}_{\pi_1}(-\tilde{L}_1) \setminus I^{c(a)}_{L/\sqrt{n}}) \cap \compC_+} Z_{n-1,\beta}(u,w)  \frac{e^{\frac{\beta}{2}anw}}{(w-u)^{\beta/2}} dw = \frac{1}{\sqrt{n}} e^{\frac{\beta}{2}n \left( \int p(x;u,c(a)) d\mu(x) + ac(a) \right)} \epsilon_2(L,n).
\end{equation}
By \eqref{eq:integral_over_Pi_middle} and \eqref{eq:another_formula_of_tildeP}, we find that $\epsilon_2(L,n)$ is small if $L$ and $n$ are large. To be precise, for any $\epsilon > 0$, there exists $L_2 > 0$ such that for all $L > L_2$, $\lvert \epsilon_2(L,n) \rvert < \epsilon$ for sufficiently large $n$.

If $w \in (\Pi^{c(a)}_{\pi_1} \setminus \Pi^{c(a)}_{\pi_1}(-\tilde{L}_1)) \cap \compC_+$, \ie, $w$ is in the ray from $c(a) + i\pi_1$ to $-\infty$ and $\Re w < -\tilde{L}_1$, since $\lvert w - \lambda_j \rvert \geq \pi_1$ for all $\lambda_j \in \realR$, by \eqref{eq:definition_of_P_n-1_beta} we have 
\begin{equation} \label{eq:estimate_of_P_n-1,beta_far_to_left}
  \lvert P_{n-1,\beta}(x_1, \dots, x_{n-1};u,w) \rvert \leq \pi^{-\frac{\beta}{2}(n-1)}_1 \prod^{n-1}_{j=1} e^{-\frac{\beta}{2}V(x_j)}(u-x_j)^{\beta}.
\end{equation}
Hence substituting \eqref{eq:estimate_of_P_n-1,beta_far_to_left} into \eqref{eq:definition_of_Z_n-1_beta}, we have
\begin{equation}
\lvert Z_{n-1,\beta}(u,w) \rvert \leq \sigma^{-\frac{\beta}{2}(n-1)}_2 \E_{n-1,\beta}(\hat{F}_{n-1,\beta}(x_1, \dots, x_{n-1};u)),
\end{equation}
where
\begin{equation} \label{eq:defn_of_hat_F}
  \hat{F}_{n-1,\beta}(x_1, \dots, x_{n-1};u) = 
  \begin{cases}
    e^{\frac{\beta}{2} \sum^{n-1}_{j=1} (-V(x_j)+2\log(u-x_j))} & \textnormal{if $\max_{1 \leq j < n} x_j < u$,} \\
    0 & \textnormal{otherwise.}
  \end{cases}
\end{equation}
Similar to \eqref{eq:average_pf_F_beta_by sandwich} and \eqref{eq:1st_support_of_cor:outlier_micro_2_c}, we have (see Remark \ref{rmk:evaluation_of_E(hat_F)})
\begin{gather}
  \E_{n-1,\beta}(\hat{F}_{n-1,\beta}(x_1, \dots, x_{n-1};u)) = e^{\frac{\beta}{2}n \int -V(x)+ 2\log(u-x) d\mu(x)} O(1), \\
  Z_{n-1,\beta}(u,w) = \sigma^{-\frac{\beta}{2}n}_2
  e^{\frac{\beta}{2}n \int -V(x)+ 2\log(u-x) d\mu(x)}
  O(1), \label{eq:rough_estimate_of_Z(uw)_unbounded}
\end{gather}
where in \eqref{eq:rough_estimate_of_Z(uw)_unbounded} $w \in (\Pi^{c(a)}_{\pi_1}\setminus \Pi^{c(a)}_{\pi_1}(-\tilde{L}_1)) \cap \compC_+$, and the $O(1)$ factor is bounded uniformly in $w$. Taking $s = -\Re w$, by substituting the estimate \eqref{eq:rough_estimate_of_Z(uw)_unbounded} of $Z_{n-1,\beta}(u,w)$ into \eqref{eq:pdf_of_largest_eigen_supercritical}, we have
\begin{equation} \label{eq:estimate_of_integral_far_away_u>c(a)_pre}
  \begin{split}
    & \left\lvert \frac{1}{2\pi i} \int_{(\Pi^{c(a)}_{\pi_1}\setminus \Pi^{c(a)}_{\pi_1}(-\tilde{L}_1)) \cap \compC_+} Z_{n-1,\beta}(u,w)  \frac{e^{\frac{\beta}{2}anw}}{(w-u)^{\beta/2}} dw \right\rvert \\
    \leq & \frac{1}{2\pi} \int^{\infty}_{\tilde{L}_1} \pi^{-\frac{\beta}{2}n}_1 e^{\frac{\beta}{2}n \int -V(x)+ 2\log(u-x) d\mu(x)} O(1) \frac{e^{-\frac{\beta}{2}ans}}{\lvert -s-u+i\pi_1 \rvert^{\beta/2}} ds \\
    = & e^{\frac{\beta}{2}n \left( \int -V(x)+ 2\log(u-x) d\mu(x) - \log(\pi_1) - a\tilde{L}_1 \right)} O(n^{-1}),
  \end{split}
\end{equation}
where the $O(n^{-1})$ factor is uniform in $\tilde{L}_1$. For large enough $\tilde{L}_1$,
\begin{equation} \label{eq:inequality_for_tildeL_1_small_u}
\int -V(x) + 2\log(u-x) d\mu(x) - a\tilde{L}_1 - \log \pi_1 < \int p(x;u,c(a)) d\mu(x) + ac(a).
\end{equation}
Substituting \eqref{eq:inequality_for_tildeL_1_small_u} into \eqref{eq:estimate_of_integral_far_away_b_2<u<c(a)}, we have 
\begin{equation} \label{eq:estimate_of_integral_far_away_b_2<u<c(a)_final}
  \frac{1}{2\pi i} \int_{(\Pi^{c(a)}_{\pi_1} \setminus (\Pi^{c(a)}_{\pi_1}(-\tilde{L}_1)) \cap \compC_+} Z_{n-1,\beta}(u,w)  \frac{e^{\frac{\beta}{2}anw}}{(w-u)^{\beta/2}} dw = \frac{1}{\sqrt{n}} e^{\frac{\beta}{2}n \left( \int p(x;u,c(a)) d\mu(x) + ac(a) \right)} o(1).
\end{equation}

The results of \eqref{eq:estimate_of_integral_middle_b_2<u<c(a)} and \eqref{eq:estimate_of_integral_far_away_b_2<u<c(a)_final} give an estimate of the contour integral \eqref{eq:pdf_of_largest_eigen_supercritical} on $(\Pi^{c(a)}_{\pi_1} \setminus I^{c(a)}_{L/\sqrt{n}}) \cap \compC_+$. The integral on $(\Pi^{c(a)}_{\pi_1} \setminus I^{c(a)}_{L/\sqrt{n}}) \cap \compC_-$ is similar, since both the contour and the integrand are symmetric about the real axis.

Therefore, by \eqref{eq:estimate_of_integral_essential_b_2<u<c(a)}, \eqref{eq:estimate_of_integral_middle_b_2<u<c(a)} and \eqref{eq:estimate_of_integral_far_away_b_2<u<c(a)_final}, we have that for $u \in [\redge + \bfepsilon, c(a) - \bfepsilon]$, 
\begin{equation} \label{eq:asy_of_f_ximax_near}
\begin{split}
  f_{\ximax}(u) = & \hat{C}_{n,\beta} e^{-\frac{\beta}{2}nV(u)} \frac{1}{\sqrt{n}} e^{\frac{\beta}{2}n \left( \int p(x;u,c(a)) d\mu(x) + ac(a) \right)} \left( \tilde{\M}_{\beta}(u) + o(1) \right) \\
  = & \hat{C}_{n,\beta} \frac{1}{\sqrt{n}} e^{-\frac{\beta}{2}n \int V(x) d\mu(x)} e^{\frac{\beta}{2}n (-V(u) + 2\gfn(u) - \ell)}e^{\frac{\beta}{2}n \Hfn(c(a);a)} \left( \tilde{\M}_{\beta}(u) + o(1) \right) \\
  = & \bfC_{n,\beta} n^{(1-\beta)/2} e^{\frac{\beta}{2}n(\Hfn(c(a);a) - V(u) + 2 \gfn(u) - \ell)} \left( \tilde{\M}_{\beta}(u) + o(1) \right),
\end{split}
\end{equation}
where $\gfn(u)$ is defined in \eqref{eq:defn_of_g_function}, $\Hfn(c(a);a)$ is defined in \eqref{eq:defn_of_Hfn}, and
\begin{equation} \label{eq:defn_of_bfC_newcorrection}
\bfC_{n,\beta} := \hat{C}_{n,\beta} n^{\beta/2 - 1} e^{-\frac{\beta}{2}n \int V(x) d\mu(x)}.
\end{equation}

\subsection{Computation of $f_{\ximax}(u)$ when $u \in [c(a)+\bfepsilon, \redge+\bfepsilon^{-1}]$} \label{subsec:u>c(a)}

Assuming $c(a)+\bfepsilon \leq u \leq \redge+\bfepsilon^{-1}$, we use the contour integral formula \eqref{eq:pdf_of_largest_eigen_supercritical} of $f_{\ximax}(u)$ and take the contour $\Cont$ in \eqref{eq:pdf_of_largest_eigen_supercritical} as $\Sigma^{u+n^{-1}}_{\sigma_1,\sigma_2}$, which is defined in \eqref{eq:parametrization_of_Sigma}. Here $\sigma_1$ is a small enough parameter and $\sigma_2$ is a large enough parameter, such that the inequalities \eqref{eq:derivative_on_Sigma_III} and \eqref{eq:decreasing_speed_of_w_not_far} hold. For $w \in \Sigma^{u+n^{-1}}_{\sigma_1,\sigma_2}$, we parametrize it by $w = w(t)$ as in \eqref{eq:parametrization_of_Sigma} with $x = u+n^{-1}$, $s_1 = \sigma_1$ and $s_2 = \sigma_2$.

Let $L$ be a positive number. For sufficiently large $n$, the contour $\Gamma^{u+n^{-1}}_{L/n}$ is part of $\Sigma^{u+n^{-1}}_{\sigma_1, \sigma_2}$. By Proposition \ref{cor:outlier_micro}\ref{enu:cor:outlier_micro:a}, for $w = u+z/n \in \Gamma^{u+n^{-1}}_{L/n}$ the integrand in the contour integral of \eqref{eq:pdf_of_largest_eigen_supercritical} satisfies
\begin{equation}
  Z_{n-1,\beta}(u,w)  \frac{e^{\frac{\beta}{2}anw}}{(w-u)^{\beta/2}} = 
  e^{\frac{\beta}{2}n \left( \int p(x;u) d\mu(x) + au \right)}
  R_{\beta}(u) \frac{e^{\frac{\beta z}{2} \left( a - \int
        \frac{d\mu(x)}{u-x} \right)}}{(z/n)^{\beta/2}} (1+o(1)).
\end{equation}
Hence
\begin{multline} \label{eq:estimate_of_integral_over_finite_Gamma_contour}
  \frac{1}{2\pi i} \int_{\Gamma^{u+n^{-1}}_{L/n}} Z_{n-1,\beta}(u,w)  \frac{e^{\frac{\beta}{2}anw}}{(w-u)^{\beta/2}} dw = \\
  n^{\beta/2 - 1} e^{\frac{\beta}{2}n \left( \int p(x;u) d\mu(x) + au
    \right)} R_{\beta}(u) \frac{1}{2\pi i} \int_{\Gamma^{1}_L}
  \frac{e^{\frac{\beta z}{2} \left( a - \int \frac{d\mu(x)}{u-x}
      \right)}}{z^{\beta/2}} dz (1+o(1)).
\end{multline}
Using Hankel's contour integral expression of Gamma function (See \cite[6.1.4]{Abramowitz-Stegun64}), we find
\begin{equation} \label{eq:model_integral_over_infinite_Gamma_contour}
\frac{1}{2\pi i} \int_{\Gamma^{1}_{\infty}} \frac{e^{\frac{\beta z}{2} \left( a - \int \frac{d\mu(x)}{u-x} \right)}}{z^{\beta/2}} dz = \left[ \frac{\beta}{2}\left( a - \int \frac{d\mu(x)}{u-x} \right) \right]^{\frac{\beta}{2}-1} \Gamma\left( \frac{\beta}{2} \right).
\end{equation}
Comparing the integral on the right-hand side of \eqref{eq:estimate_of_integral_over_finite_Gamma_contour} with the left-hand side of \eqref{eq:model_integral_over_infinite_Gamma_contour}, we write analogous to \eqref{eq:estimate_of_integral_essential_b_2<u<c(a)} that
\begin{equation} \label{eq:formula_define_epsilon_1(Ln)}
  \frac{1}{2\pi i} \int_{\Gamma^{u+n^{-1}}_{L/n}} Z_{n-1,\beta}(u,w)  \frac{e^{\frac{\beta}{2}anw}}{(w-u)^{\beta/2}} dw = n^{\beta/2 - 1} e^{\frac{\beta}{2}n \left( \int p(x;u) d\mu(x) + au \right)} \M_{\beta}(u) (1+\epsilon_3(L,n)),
\end{equation}
where 
\begin{equation} \label{eq:defn_of_M_beta}
\M_{\beta}(u) = R_{\beta}(u) \left[ \frac{\beta}{2}\left( a - \int \frac{d\mu(x)}{u-x} \right) \right]^{\frac{\beta}{2}-1} \Gamma\left( \frac{\beta}{2} \right),
\end{equation}
and the term $\epsilon_3(L,n)$ is small if $L$ and $n$ are large. To be precise, for any $\epsilon > 0$, there is an $L_3 > 0$ such that for all $L > L_3$, $\lvert \epsilon_3(L,n) \rvert < \epsilon$ for sufficiently large $n$.

Let $\tilde{L}_2$ be a large enough positive number such that $-\tilde{L}_2 < u-\sigma_1/\sqrt{2}$ and the inequality \eqref{eq:inequality_for_tildeL_1} holds. For $w \in \Sigma^{u+n^{-1}}_{\sigma_1,\sigma_2}(-\tilde{L}_2) \setminus \Gamma^{u+n^{-1}}_{L/n}$, we use Proposition \ref{cor:outlier_micro}\ref{enu:cor:outlier_micro:b} and find that \eqref{eq:estimate_of_absolute_value_of_integrand} still holds.

If $w = w(t) \in \Sigma^{u+n^{-1}}_{\sigma_1,\sigma_2}$ and $0 \leq t \leq \sigma_1$, \ie, $w$ is in the line segment between $u+n^{-1}$ and $u+n^{-1}+e^{\frac{3\pi i}{4}}\sigma_1$, we have
\begin{equation} \label{eq:derivative_of_intergrand_Sigma_I}
\frac{d \tildeP(u,w(t))}{dt} = \int \frac{(u-x+n^{-1})/\sqrt{2}-t}{\lvert w(t) -x \rvert^2} d\mu(x) - \frac{a}{\sqrt{2}}.
\end{equation}
By \eqref{eq:positivity_of_H'(u;a)}, we know that $\Hfn'(u;a)$ is a positive number. If $\sigma_1$ is small enough, for all $t \in [0, \sigma_1]$ and all $x \in J$
\begin{equation} \label{eq:derivative_on_Sigma_III}
\frac{(u-x+n^{-1})/\sqrt{2}-t}{\lvert w(t) - x \rvert^2} < \frac {1}{\sqrt{2}} \left( \frac{1}{u-x} + \left( 1-\frac{1}{\sqrt{2}} \right) \Hfn'(u;a) \right).
\end{equation}
Hence for $0 \leq t \leq \sigma_1$, substituting \eqref{eq:derivative_on_Sigma_III} in \eqref{eq:derivative_of_intergrand_Sigma_I}, we find with the help of \eqref{eq:positivity_of_H'(u;a)}
\begin{equation} \label{eq:quantitative_est_of_dervative_on_Sigma_1}
  \begin{split}
    \frac{d \tildeP(u,w(t))}{dt} < & \frac{1}{\sqrt{2}} \int \frac{d\mu(x)}{u-x} + \left( \frac{1}{\sqrt{2}}-\frac{1}{2} \right) \Hfn'(u;a) \int d\mu(x) - \frac{a}{\sqrt{2}} \\
    = & -\frac{1}{\sqrt{2}} \Hfn'(u;a) + \left(
      \frac{1}{\sqrt{2}}-\frac{1}{2} \right) \Hfn'(u;a) = -\frac{1}{2}
    \Hfn'(u;a).
  \end{split}
\end{equation}

If $w = w(t) \in \Sigma^{u+n^{-1}}_{\sigma_1,\sigma_2}$ and $t \geq \sigma_2 + (1-\sqrt{2}/2)\sigma_1$, \ie, $w$ is in the ray from $u+n^{-1}-\sigma_1/\sqrt{2}+i\sigma_2$ to $-\infty$, like \eqref{eq:derivative_of_tildeP_new_crrection} we have
\begin{equation} \label{eq:derivative_on_Sigma_III_far} \frac{d
    \tildeP(u,w(t))}{dt} = \int \frac{u-x+n^{-1} + \sigma_2 - (\sqrt{2}-1)\sigma_1 - t}{\lvert w(t) -x \rvert^2} d\mu(x) - a.
\end{equation}
If $\sigma_2$ is large enough, like \eqref{eq:inequality_of_second_decreasing_horizontal_real} we have that for all $t \in [\sigma_2 + (1-\sqrt{2}/2)\sigma_1, \infty)$ and all $x \in J$, if $n$ is large enough
\begin{equation} \label{eq:decreasing_speed_of_w_not_far}
  \frac{u-x+n^{-1} + \sigma_2 - (\sqrt{2}-1)\sigma_1 - t}{\lvert w(t) -x \rvert^2} < \frac{1}{u-x} + \frac{1}{2} \Hfn'(u;a).
\end{equation}
Substituting \eqref{eq:decreasing_speed_of_w_not_far} into \eqref{eq:derivative_on_Sigma_III_far}, we find that like \eqref{eq:decreasing_speed_of_w_not_far}, for $t \geq \sigma_2 + (1-\sqrt{2}/2)\sigma_1$
\begin{equation} \label{eq:quantitative_est_of_dervative_on_Sigma_3}
\frac{d \tildeP(u,w(t))}{dt} < \int \frac{d\mu(x)}{u-x} + \frac{1}{2}\Hfn'(u;a) -a = -\frac{1}{2}\Hfn'(u;a).
\end{equation}

If $w = w(t) \in \Sigma^{u+n^{-1}}_{\sigma_1,\sigma_2}$ and $\sigma_1 \leq t \leq \sigma_2 + (1-\sqrt{2}/2)\sigma_1$, \ie, $w$ is in the line segment between $u+n^{-1}+e^{\frac{3\pi i}{4}}\sigma_1$ and $u+n^{-1}-\sigma_1/\sqrt{2}+i\sigma_2$ we have
\begin{equation} \label{eq:est_of_dervative_on_Sigma_2}
\frac{d \tildeP(u,w(t))}{dt} = -\int \frac{d}{dt} \log \lvert w(t)-x \rvert d\mu(x) < 0.
\end{equation}

Thus by \eqref{eq:estimate_of_absolute_value_of_integrand}, \eqref{eq:quantitative_est_of_dervative_on_Sigma_1}, \eqref{eq:quantitative_est_of_dervative_on_Sigma_3} and \eqref{eq:est_of_dervative_on_Sigma_2}, we find that for a fixed $\tilde{L}_2$, similar to \eqref{eq:integral_over_Pi_middle},
\begin{equation} \label{eq:estimate_of_finite_integral_upto_tildeL}
  \begin{split}
    & \left\lvert \frac{1}{2\pi i} \int_{(\Sigma^{u+n^{-1}}_{\sigma_1,\sigma_2}(-\tilde{L}_2) \setminus \Gamma^{u+n^{-1}}_{L/n}) \cap \compC_+} Z_{n-1,\beta}(u,w)  \frac{e^{\frac{\beta}{2}anw}}{(w-u)^{\beta/2}} dw \right\rvert \\
    \leq & \frac{1}{2\pi} \int^{\tilde{L}+u+n^{-1} + \sigma_2 + (\sqrt{2}-1)\sigma_1}_{L/n} \left\lvert Z_{n-1,\beta}(u,w(t))  \frac{e^{\frac{\beta}{2}anw(t)}}{(w(t)-u)^{\beta/2}} \right\rvert dt \\
    = & \frac{1}{2\pi} \int^{\tilde{L}+u+n^{-1} + \sigma_2 +
      (\sqrt{2}-1)\sigma_1}_{L/n} e^{\frac{\beta}{2}n \tildeP(u,w(t))}
    \frac{1}{\lvert w(t)-u \rvert^{\beta/2}} O(1) dt \\
    = & n^{-1} \exp \left[ \frac{\beta}{2}n \tildeP(u, w(L/n)) \right] \frac{1}{\lvert w(L/n)-u \rvert^{\beta/2}} O(1) \\
    = & n^{\frac{\beta}{2}-1} \exp \left[ \frac{\beta}{2}n \tildeP(u,
      w(L/n)) \right] \frac{1}{\lvert 1 + e^{\frac{3\pi i}{4}}L
      \rvert^{\beta/2}} O(1),
  \end{split}
\end{equation}
where the last factor $O(1)$ is bounded uniformly in $L$. Substituting the Taylor expansion \eqref{eq:relation_between_p_and_tilde_p} of $p(x;u,w)$ into \eqref{eq:defn_of_tildeP}, we find that like \eqref{eq:another_formula_of_tildeP}
\begin{equation} \label{eq:another_formula_of_tildeP_u}
e^{\frac{\beta}{2} n\tildeP(u, w(L/n))} = e^{\frac{\beta}{2} n\left( \int p(x;u) d\mu(x) + au \right)} \left( e^{-\frac{\beta L}{2\sqrt{2}} \int \frac{d\mu(x)}{u-x} - \frac{aL}{\sqrt{2}}} + o(1) \right).
\end{equation}
Like \eqref{eq:estimate_of_integral_middle_b_2<u<c(a)}, we write
\begin{equation} \label{eq:estimate_of_cont_int_Sigma_middle}
  \frac{1}{2\pi i} \int_{(\Sigma^{u+n^{-1}}_{\sigma_1,\sigma_2}(-\tilde{L}_2) \setminus \Gamma^{u+n^{-1}}_{L/n}) \cap \compC_+} Z_{n-1,\beta}(u,w) \frac{e^{\frac{\beta}{2}anw}}{(w-u)^{\beta/2}} dw =  n^{\frac{\beta}{2}-1} e^{\frac{\beta}{2}n \left( \int p(x;u)d\mu(x) + au \right)} \epsilon_4(L,n).
\end{equation}
By \eqref{eq:estimate_of_finite_integral_upto_tildeL} and \eqref{eq:another_formula_of_tildeP_u}, we find that $\epsilon_4(L,n)$ is small if $L$ and $n$ are large. To be precise, for any $\epsilon > 0$, there is an $L_4 > 0$ such that for all $L > L_4$, $\lvert \epsilon_4(L,n) \rvert < \epsilon$ for sufficiently large $n$.

Like \eqref{eq:estimate_of_integral_far_away_u>c(a)_pre}, we have
\begin{multline} \label{eq:estimate_of_integral_far_away_b_2<u<c(a)}
  \left\lvert \frac{1}{2\pi i} \int_{(\Sigma^{u+n^{-1}}_{\sigma_1,\sigma_2} \setminus \Sigma^{u+n^{-1}}_{\sigma_1,\sigma_2}(-\tilde{L}_2)) \cap \compC_+} Z_{n-1,\beta}(u,w)  \frac{e^{\frac{\beta}{2}anw}}{(w-u)^{\beta/2}} dw \right\rvert = \\
  e^{\frac{\beta}{2}n \left( \int -V(x) + 2\log(u-x) d\mu(x) - \log(\sigma_2) - a\tilde{L}_2 \right)} O(n^{-1}),
\end{multline}
where the $O(n^{-1})$ factor is uniform in $\tilde{L}_2$. For large enough $\tilde{L}_2$,
\begin{equation} \label{eq:inequality_for_tildeL_1}
\int -V(x)+ 2\log(u-x) d\mu(x) - \log(\sigma_2) -a\tilde{L}_2 < \int p(x;u) d\mu(x) + au.
\end{equation}
Substituting \eqref{eq:inequality_for_tildeL_1} into \eqref{eq:estimate_of_integral_far_away_b_2<u<c(a)}, we obtain
\begin{equation} \label{eq:estimate_of_integral_far_away_u>c(a)}
  \frac{1}{2\pi i} \int_{(\Sigma^{u+n^{-1}}_{\sigma_1,\sigma_2} \setminus \Sigma^{u+n^{-1}}_{\sigma_1,\sigma_2}(-\tilde{L}_2)) \cap \compC_+} Z_{n-1,\beta}(u,w)  \frac{e^{\frac{\beta}{2}anw}}{(w-u)^{\beta/2}} dw = n^{\beta/2 - 1} e^{\frac{\beta}{2}n \left( \int p(x;u) d\mu(x) + au \right)} o(1).
\end{equation}

The results of \eqref{eq:estimate_of_cont_int_Sigma_middle} and \eqref{eq:estimate_of_integral_far_away_u>c(a)} give an estimate of the contour integral in \eqref{eq:pdf_of_largest_eigen_supercritical} on $(\Sigma^{u+n^{-1}}_{\sigma_1,\sigma_2} \setminus \Gamma^{u+n^{-1}}_{L/n}) \cap \compC_+$. The contour integral on $(\Sigma^{u+n^{-1}}_{\sigma_1,\sigma_2} \setminus \Gamma^{u+n^{-1}}_{L/n}) \cap \compC_-$ is similar, since both the contour and the integrand in \eqref{eq:pdf_of_largest_eigen_supercritical} are symmetric about the real axis.

Therefore, by \eqref{eq:pdf_of_largest_eigen_supercritical}, \eqref{eq:formula_define_epsilon_1(Ln)}, \eqref{eq:estimate_of_cont_int_Sigma_middle} and \eqref{eq:estimate_of_integral_far_away_u>c(a)}, we have that for $u \in [c(a)+\bfepsilon, \redge+\bfepsilon^{-1}]$,
\begin{equation} \label{eq:f_ximax_right_to_c(a)}
\begin{split}
f_{\ximax}(u) = & \hat{C}_{n,\beta} e^{-\frac{\beta}{2}nV(u)} n^{\beta/2 - 1} e^{\frac{\beta}{2}n \left( \int p(x;u) d\mu(x) + au \right)} (\M_{\beta}(u) + o(1)) \\
= & \bfC_{n,\beta} e^{\frac{\beta}{2}n\Gfn(u;a)}(\M_{\beta}(u) + o(1)),
\end{split}
\end{equation}
where $\Gfn(u;a)$ is defined in \eqref{eq:defn_of_Gfn}, $\M_{\beta}(u)$ is defined in \eqref{eq:defn_of_M_beta} and $\bfC_{n,\beta}$ is defined in \eqref{eq:defn_of_bfC_newcorrection}.

\subsection{Estimation of $f_{\ximax}(u)$ when $u \in (c(a)-\bfepsilon, c(a)+\bfepsilon)$, $u > \redge + \bfepsilon^{-1}$ or $u < \redge + \bfepsilon$} \label{subsec:other_u}

\subsubsection{$u \in [c(a)-\bfepsilon, c(a)+\bfepsilon]$ or $u \in [\redge-\bfepsilon, \redge+\bfepsilon]$}

In this subsubsection we use the inequality that if $u_1 < u_2$ and $w \in \compC \setminus (-\infty,u_2)$, then
\begin{equation} \label{eq:inequality_involving_u_1_u_2}
  Z_{n-1,\beta}(u_1,w) \leq \hat{Z}_{n-1,\beta}(u_1,w) \leq \hat{Z}_{n-1,\beta}(u_2,w).
\end{equation}
The inequality \eqref{eq:inequality_involving_u_1_u_2} is a straightforward consequence of the definitions \eqref{eq:definition_of_Z_n-1_beta} and \eqref{eq:defn_of_Z_hat_(uwc)} of $Z_{n-1,\beta}(u,w)$ and $\hat{Z}_{n-1,\beta}(u,w)$.

For $c(a)-\bfepsilon \leq u \leq c(a)+\bfepsilon$, we use the contour integral formula \eqref{eq:pdf_of_largest_eigen_supercritical} of $f_{\ximax}(u)$ and take the contour $\Cont$ in \eqref{eq:pdf_of_largest_eigen_supercritical} as $\Pi^{c(a)+2\bfepsilon}_{\pi_2}$, which is defined in \eqref{eq:parametrization_of_Pi}. Here $\pi_2$ is a large enough parameter such that the inequality \eqref{eq:inequality_of_decreasing_horizontal_u_around_c(a)} holds. For $w \in \Pi^{c(a)+2\bfepsilon}_{\pi_2}$, we parametrize it by $w=w(t)$ as in \eqref{eq:parametrization_of_Pi} with $x = c(a)+2\bfepsilon$ and $s = \pi_2$.
From \eqref{eq:inequality_involving_u_1_u_2} we have for all $u \in [c(a)-\bfepsilon, c(a)+\bfepsilon]$ and $\tilde{L}_3 > -(c(a)+2\bfepsilon)$
\begin{multline} \label{eq:basic_inequality_for_u_around_c(a)_and_b_2}
  \left\lvert \frac{1}{2\pi i} \int_{\Pi^{c(a)+2\bfepsilon}_{\pi_2}(-\tilde{L}_3) \cap \compC_+} Z_{n-1,\beta}(u,w) \frac{e^{\frac{\beta}{2}anw}}{(w-u)^{\beta/2}} dw \right\rvert \leq \\
  \frac{1}{2\pi} \int^{c(a)+2\bfepsilon + \tilde{L}_3 + \pi_2}_{0} \hat{Z}_{n-1,\beta}(c(a)+\bfepsilon, w(t)) \frac{e^{\frac{\beta}{2}an \Re(w(t))}}{\lvert w(t)-u \rvert^{\beta/2}} dt.
\end{multline}

For $w \in \Pi^{c(a)+2\bfepsilon}_{\pi_2}$, we use Proposition \ref{cor:outlier_micro}\ref{enu:cor:outlier_micro:b} and find like \eqref{eq:estimate_of_absolute_value_of_integrand}
\begin{equation} \label{eq:est_like_estimate_of_absolute_value_of_integrand}
  \hat{Z}_{n-1,\beta}(c(a)+\bfepsilon, w(t))
  \frac{e^{\frac{\beta}{2}an \Re(w(t))}}{\lvert w(t)-u
    \rvert^{\beta/2}} = e^{\frac{\beta}{2}n
    \tildeP(c(a)+\bfepsilon,w(t))} \frac{1}{\lvert w(t)-u
    \rvert^{\beta/2}} O(1).
\end{equation}
Like \eqref{eq:inequality_of_second_decreasing_horizontal}, we have that for $\pi_2$ large enough, for all $t > \pi_2$ and $x \in J$, there exists $c'_1 > 0$ such that (\cf\ \eqref{eq:inequality_of_second_decreasing_horizontal_real} and \eqref{eq:inequality_of_second_decreasing_horizontal})
\begin{equation} \label{eq:inequality_of_decreasing_horizontal_u_around_c(a)}
\frac{d \tildeP(c(a)+\bfepsilon,w(t))}{dt} < -c'_1.
\end{equation}
For all $0 < t \leq \pi_1$ and $x \in J$, we have like \eqref{eq:inequality_of_second_decreasing_vertical} and \eqref{eq:est_of_dervative_on_Sigma_2} that
\begin{equation} \label{eq:inequality_of_decreasing_vertical_u_around_c(a)}
\frac{d \tildeP(c(a)+\bfepsilon,w(t))}{dt}< 0.
\end{equation}

Let $\tilde{L}_3$ be a large enough positive number such that the inequality \eqref{eq:inequality_for_tildeL_1_not_so_small_u} holds. By \eqref{eq:est_like_estimate_of_absolute_value_of_integrand}, \eqref{eq:inequality_of_decreasing_horizontal_u_around_c(a)} and \eqref{eq:inequality_of_decreasing_vertical_u_around_c(a)}, we find like \eqref{eq:integral_over_Pi_middle}
\begin{multline} \label{eq:estimate_of_integral_not_far_middle_u}
  \frac{1}{2\pi} \int^{\tilde{L}_3+c(a)+2\bfepsilon+\pi_2}_0 \hat{Z}_{n-1,\beta}(c(a)+\bfepsilon, w(t)) \frac{e^{\frac{\beta}{2}an \Re(w(t))}}{\lvert w(t)-u \rvert^{\beta/2}} dt = \\
  e^{\frac{\beta}{2}n \left( \int p(x; c(a)+\bfepsilon, c(a)+2\bfepsilon) d\mu(x) + a(c(a)+2\bfepsilon) \right)} O(1).
\end{multline}
Like \eqref{eq:estimate_of_integral_far_away_u>c(a)_pre} and \eqref{eq:estimate_of_integral_far_away_b_2<u<c(a)}, we also have
\begin{multline} \label{eq:estimate_of_integral_far_middle_u}
  \left\lvert \frac{1}{2\pi i} \int_{(\Pi^{c(a)+2\bfepsilon}_{\pi_2} \setminus (\Pi^{c(a)+2\bfepsilon}_{\pi_2}(-\tilde{L}_3)) \cap \compC_+} Z_{n-1,\beta}(u,w) \frac{e^{\frac{\beta}{2}anw}}{(w-u)^{\beta/2}} dw \right\rvert = \\
  e^{\frac{\beta}{2}n \left( \int -V(x) + 2\log(u-x) d\mu(x) - \log(\pi_2) - a\tilde{L}_3 \right)} O(n^{-1}),
\end{multline}
where the $O(n^{-1})$ factor is uniform in $\tilde{L}_3$. For large enough $\tilde{L}_3$,
\begin{equation} \label{eq:inequality_for_tildeL_1_not_so_small_u}
  \int -V(x) + 2\log(u-x) d\mu(x) - \log \pi_2 - a\tilde{L}_3 < \int p(x;c(a)+\bfepsilon, c(a)+2\bfepsilon) d\mu(x) + a(c(a)+2\bfepsilon).
\end{equation}
Substituting \eqref{eq:inequality_for_tildeL_1_not_so_small_u} into \eqref{eq:estimate_of_integral_far_middle_u}, we have 
\begin{multline} \label{eq:estimate_of_integral_far_away_u_around_c(a)_final}
  \frac{1}{2\pi i} \int_{(\Pi^{c(a)+2\bfepsilon}_{\pi_2} \setminus (\Pi^{c(a)+2\bfepsilon}_{\pi_2}(-\tilde{L}_3)) \cap \compC_+} Z_{n-1,\beta}(u,w)  \frac{e^{\frac{\beta}{2}anw}}{(w-u)^{\beta/2}} dw = \\
   e^{\frac{\beta}{2}n \left( \int p(x; c(a)+\bfepsilon, c(a)+2\bfepsilon) d\mu(x) + a(c(a)+2\bfepsilon) \right)} o(1).
\end{multline}

The results of \eqref{eq:estimate_of_integral_not_far_middle_u} and \eqref{eq:estimate_of_integral_far_away_u_around_c(a)_final} give an estimate of the contour integral \eqref{eq:pdf_of_largest_eigen_supercritical} on $\Pi^{c(a)+2\bfepsilon}_{\pi_2} \cap \compC_+$. The integral on $\Pi^{c(a)+2\bfepsilon}_{\pi_2} \cap \compC_-$ is similar, since both the contour and the integrand are symmetric about the real axis. Thus we obtain
\begin{equation} \label{eq:f_ximax_left_to_c(a)}
  \begin{split}
    f_{\ximax}(u) = & \hat{C}_{n,\beta} e^{-\frac{\beta}{2}V(u)} e^{\frac{\beta}{2}n \left( \int p(x;c(a)+\bfepsilon,c(a)+2\bfepsilon) d\mu(x) +a(c(a)+2\bfepsilon) \right)} O(1) \\
    = & \bfC_{n,\beta} n^{1-\beta/2} e^{\frac{\beta}{2}n \Hfn(c(a)+2\bfepsilon;a)} e^{\frac{\beta}{2}n ((-V(c(a)+\bfepsilon) + 2\gfn(c(a)+\bfepsilon) - \ell) + (V(c(a)+\bfepsilon)-V(u)))} O(1).
  \end{split}
\end{equation}

For $\redge-\bfepsilon \leq u \leq \redge+\bfepsilon$ we use the contour integral formula \eqref{eq:pdf_of_largest_eigen_supercritical} of $f_{\ximax}(u)$ and take the contour $\Cont$ in \eqref{eq:pdf_of_largest_eigen_supercritical} as $\Pi^{c(a)+2\bfepsilon}_{\pi_3}$, the same as in the $c(a)-\bfepsilon \leq u \leq c(a)+\bfepsilon$ case. We also apply the inequality \eqref{eq:inequality_involving_u_1_u_2}, and like \eqref{eq:basic_inequality_for_u_around_c(a)_and_b_2} have for all $u \in [\redge-\bfepsilon, \redge+\bfepsilon]$ 
\begin{multline} \label{eq:basic_inequality_for_u_around_c(a)_and_b_2}
  \left\lvert \frac{1}{2\pi i} \int_{\Pi^{c(a))+2\bfepsilon}_{\pi_3}(-\tilde{L}_3) \cap \compC_+} Z_{n-1,\beta}(u,w) \frac{e^{\frac{\beta}{2}anw}}{(w-u)^{\beta/2}} dw \right\rvert \leq \\
  \frac{1}{2\pi} \int^{c(a)+2\bfepsilon +\tilde{L}_3 + \pi_2}_{0} \hat{Z}_{n-1,\beta}(\redge+\bfepsilon, w(t)) \frac{e^{\frac{\beta}{2}an \Re(w(t))}}{\lvert w(t)-u \rvert^{\beta/2}} dt.
\end{multline}
Then we can find estimates similar to \eqref{eq:estimate_of_integral_not_far_middle_u} and \eqref{eq:estimate_of_integral_far_away_u_around_c(a)_final}, and obtain the estimate of $f_{\ximax}(u)$ similar to \eqref{eq:f_ximax_left_to_c(a)}. We only state the result that for $u \in [\redge-\bfepsilon, \redge+\bfepsilon]$
\begin{equation} \label{eq:est_of_f_ximax_around_b_2}
  f_{\ximax}(u) = \bfC_{n,\beta} n^{1-\beta/2} e^{\frac{\beta}{2}n \Hfn(c(a)+2\bfepsilon;a)} e^{\frac{\beta}{2}n ((-V(\redge+\bfepsilon) + 2\gfn(\redge+\bfepsilon) - \ell) + (V(\redge+\bfepsilon)-V(u)))} O(1),
\end{equation}
and skip details.

\subsubsection{$u \leq \redge-\bfepsilon$ or $u \geq \redge+\bfepsilon^{-1}$}

In this subsubsection we first consider $f_{\ximax}(u)$ for $u \leq \redge-\bfepsilon$. We use the contour integral formula \eqref{eq:pdf_of_largest_eigen_supercritical} of $f_{\ximax}(u)$ and take the contour $\Cont$ as $\Pi^{\redge+\bfepsilon}_{\pi_4}$. Here $\pi_4$ is a large enough parameter such that the inequality \eqref{eq;condition_of_pi_1_u_tobe_far_left} holds.

We let $C_{\bfepsilon}$ be any positive number, and define
\begin{equation}
C_V = \max_{x < \redge} (-V(x) + \log(\redge-x)).
\end{equation}
Let $f_{\bfepsilon}(x)$ be a function on $\realR$ such that
\begin{enumerate}[label=(\arabic*)]
\item $f_{\bfepsilon}(x)$ satisfies conditions \ref{enu:Johansson_condition_1}--\ref{enu:Johansson_condition_3} mentioned in Proposition \ref{prop:Johansson}.

\item On $(-\infty,\redge-\bfepsilon]$
  \begin{equation} \label{eq:linear_growth_of_f_epsilon}
    f_{\bfepsilon}(x) = C_V
  \end{equation}
and $f_{\bfepsilon}(x)$ is decreasing on $(\redge-\bfepsilon, \infty)$.

\item 
  \begin{equation} \label{eq:bound_of_intergral_f_epsilon}
    \int f_{\bfepsilon}(x) d\mu(x) < -C_{\bfepsilon}.
  \end{equation} 
\end{enumerate}
For $w$ in the line segment of $\Pi^{\redge+\bfepsilon}_{\pi_4}$ from $\redge+\bfepsilon-i\pi_4$ to $\redge+\bfepsilon+i\pi_4$, for all $x < u$
\begin{equation} \label{eq:p_less_than_C_V_vertical_part}
  \begin{split}
    p(x;u,w) = & -V(x) + 2\log(u-x) - \log\lvert w-x \rvert \\
    < & -V(x) + 2\log(\redge-x) - \log(\redge-x) \leq C_V.
  \end{split}
\end{equation}
On the other hand, we assume that $\pi_4$ is large enough such that for all $x < \redge$,
\begin{equation} \label{eq;condition_of_pi_1_u_tobe_far_left}
-V(x) + 2\log(\redge-x) - \log \pi_4 < C_V.
\end{equation}
By \eqref{eq:linear_growth_of_f_epsilon} and \eqref{eq;condition_of_pi_1_u_tobe_far_left}, it is straightforward to check that for all $w$ in the two rays of $\Pi^{\redge+\bfepsilon}_{\pi_1}$, from $\redge+\bfepsilon+i\pi_4$ to $-\infty$ and from $-\infty$ to $\redge+\bfepsilon-i\pi_4$ respectively,  and for all $x < u$
\begin{equation} \label{eq:p_less_than_C_V_horizontal_part}
  p(x;u,w) \leq C_V.
\end{equation}
Then by \eqref{eq:p_less_than_C_V_vertical_part}, \eqref{eq:p_less_than_C_V_horizontal_part}, \eqref{eq:definition_of_Z_n-1_beta} and \eqref{eq:alternative_defn_of_Z_hat_(uwc)}, we have for all $w \in \Pi^{\redge+\bfepsilon}_{\pi_1}$
\begin{equation} \label{eq:estimate_of_Z_uw_for_u_far_left}
  \lvert Z_{n-1,\beta}(u,w) \rvert \leq \hat{Z}_{n-1,\beta}(u,w) \leq \E_{n-1,\beta} \left( e^{\frac{\beta}{2} \sum^{n-1}_{j=1} f_{\bfepsilon}(x_j)} \right).
\end{equation}
By Proposition \ref{prop:Johansson} and \eqref{eq:bound_of_intergral_f_epsilon},
\begin{equation} \label{eq:estimate_of_Z_uw_for_u_far_left_continued}
  \E_{n-1,\beta} \left( e^{\frac{\beta}{2} \sum^{n-1}_{j=1} f_{\bfepsilon}(x_j)} \right) = e^{-\frac{\beta}{2}n C_{\bfepsilon}} O(1).
\end{equation}
Using the estimate \eqref{eq:estimate_of_Z_uw_for_u_far_left} and \eqref{eq:estimate_of_Z_uw_for_u_far_left_continued} of $Z_{n-1,\beta}(u,w)$, we find by direct calculation
\begin{equation} \label{eq:rough_estimate_first}
  \left\lvert \frac{1}{2\pi i} \int_{\Pi^{\redge+\bfepsilon}_{\pi_4}} Z_{n-1,\beta}(u,w) \frac{e^{\frac{\beta}{2}naw}}{(w-u)^{\beta/2}} dw \right\rvert = e^{\frac{\beta}{2}n (a(\redge+\bfepsilon) - C_{\bfepsilon})} O(1).
\end{equation}
Thus by \eqref{eq:pdf_of_largest_eigen_supercritical}
\begin{equation} \label{eq:estimate_of_prob_ximax_far_left}
  \begin{split}
    f_{\ximax}(u) = & \hat{C}_{n,\beta} e^{-\frac{\beta}{2}nV(u)} e^{\frac{\beta}{2}n (a(\redge+\bfepsilon) - C_{\bfepsilon})} O(1) \\
    = & \bfC_{n,\beta} n^{1-\beta/2} e^{\frac{\beta}{2} n\left( - V(u) + a(\redge+\bfepsilon) -C_{\bfepsilon} + \int V(x) d\mu(x) \right)} O(1).
  \end{split}
\end{equation}
Note that $C_{\bfepsilon}$ can be any number, and the last $O(1)$ factor in \eqref{eq:estimate_of_prob_ximax_far_left} is bounded uniformly for all $u < \redge-\bfepsilon$.

\bigskip

Next we consider $f_{\ximax}(u)$ for $u \geq \redge+\bfepsilon^{-1}$. We use the contour integral formula \eqref{eq:pdf_of_largest_eigen_supercritical} of $f_{\ximax}(u)$ and take the contour $\Cont$ as $\Pi^{u+1}_1$, which is defined in \eqref{eq:parametrization_of_Pi}.

Let 
\begin{equation}
  V_{\min} = \min_{x \in \realR} V(x),
\end{equation}
and denote
\begin{equation}
  \tilde{C}_V := \int ( -V_{\min }- 2x) d\mu(x).
\end{equation}
For all $w \in \Pi^{u+1}_1$ and $x < u$,
\begin{equation} \label{eq:Z_hat_over_Z_in_average_expression}
  \begin{split}
    p(x;u,w) = & -V(x) + 2\log(u-x) - \log\lvert w-x \rvert \\
    \leq & -V(x) + 2\log(u-x) \\
    < & -V(x) + 2(u-x) \\
    \leq & -V_{\min} -2x + 2u.
  \end{split}
\end{equation} 
Thus similar to \eqref{eq:estimate_of_Z_uw_for_u_far_left}, we have by \eqref{eq:Z_hat_over_Z_in_average_expression} that for $w \in \Pi^{u+1}_1$
\begin{equation} \label{eq:estimate_of_Z_uw_u_far_right}
  \lvert Z_{n-1,\beta}(u,w) \rvert \leq \hat{Z}_{n-1,\beta}(u,w) \leq \E_{n-1,\beta} \left( e^{\sum^{n-1}_{j=1} (-V_{\min} - 2x_j +2u)} \right).
\end{equation}
By Proposition \ref{prop:Johansson},
\begin{equation} \label{eq:estimate_of_Z_uw_u_far_right_continued}
  \E_{n-1,\beta} \left( e^{\sum^{n-1}_{j=1} (-V_{\min} - 2x_j +2u)} \right) = e^{\frac{\beta}{2}n \int (-V_{\min} - 2x +2u) d\mu(x)} O(1) = e^{\frac{\beta}{2}n(\tilde{C}_V + 2u)} O(1).
\end{equation}
Using the estimates \eqref{eq:estimate_of_Z_uw_u_far_right} and \eqref{eq:estimate_of_Z_uw_u_far_right_continued} of $Z_{n-1,\beta}(u,w)$, we have like \eqref{eq:rough_estimate_first}
\begin{equation}
  \left\lvert \frac{1}{2\pi i} \int_{\Pi^{u+1}_1} Z_{n-1,\beta}(u,w) \frac{e^{\frac{\beta}{2}naw}}{(w-u)^{\beta/2}} dw \right\rvert = e^{\frac{\beta}{2}n ((2+a)u + a + \tilde{C}_V)} O(1).
\end{equation}
Thus by \eqref{eq:pdf_of_largest_eigen_supercritical}
\begin{equation} \label{eq:estimate_of_prob_ximax_far_right}
  \begin{split}
    f_{\ximax}(u) = & \hat{C}_{n,\beta} e^{-\frac{\beta}{2}V(u)} e^{\frac{\beta}{2}n ((2+a)u + a + \tilde{C}_V)} \\
    = & \bfC_{n,\beta} n^{1-\beta/2} e^{\frac{\beta}{2} n\left( -V(u) + (2+a)u + a + \tilde{C}_V + \int V(x) d\mu(x) \right)} O(1).
  \end{split}
\end{equation}
Note that the last $O(1)$ factor in \eqref{eq:estimate_of_prob_ximax_far_right} is bounded uniformly for all $u > \redge+\bfepsilon^{-1}$.

\section{Proofs of Theorems \ref{thm:limiting_position}, \ref{thm:limiting_distr_nondegenerate}, \ref{thm:two_normal_max_points} and \ref{thm:one_abnormal_max_points}} \label{sec:proof_of_theorems}

In this section we prove the main theorems in this paper. We divide the proofs into three subsections. In Subsection \ref{subsec:0<a<acc}, we consider the case that $0 < a < \acc$ and the case that $a = \acc = \frac{1}{2}V'(\redge)$ and $\acc \notin \mathcal{J}_V$, and prove Theorem \ref{thm:limiting_position}\ref{enu:thm:limiting_position:1}. In Subsection \ref{subsec:a>acc_and_a_notin_J_V}, we consider the case that $a > \acc$ and $a \not\in \mathcal{J}_V$, and prove Theorems \ref{thm:limiting_distr_nondegenerate} and \ref{thm:limiting_position}\ref{enu:thm:limiting_position:2}. In Subsection \ref{subsec:a>acc_and_a_in_J_V}, we consider the case that $a > \acc$ and $a \in \mathcal{J}_V$ and prove Theorems \ref{thm:two_normal_max_points}, \ref{thm:limiting_position}\ref{enu:thm:limiting_position:3} and \ref{thm:one_abnormal_max_points}. 

\subsection{Proof of Theorem \ref{thm:limiting_position}\ref{enu:thm:limiting_position:1} when $0 < a < \acc$, or $a = \acc = \frac{1}{2}V'(\redge)$ and $\acc \notin \mathcal{J}_V$} \label{subsec:0<a<acc}

First we consider the case that $0 < a < \acc$. Let $\epsilon$ be a small positive number, such that $\redge + \epsilon
< c(a)$ and $\redge+\epsilon^{-1} > c(a)$. Furthermore we assume that $\epsilon$ is small enough such that the inequalities \eqref{eq:inequality_of_epsilon_sub_condition_1}, \eqref{eq:inequality_of_epsilon_sub_condition_2} and \eqref{eq:added_last} hold. 

The condition $0 < a < \acc$ implies the inequality $\Hfn(c(a);a) > \Gfn_{\max}(a)$, see \eqref{eq:defn_of_A_v}--\eqref{eq:definition_of_x_1_to_x_r} and \cite[Lemma 1.2(d)]{Baik-Wang10a}. We assume that 
\begin{equation} \label{eq:inequality_of_epsilon_sub_condition_1}
  \Hfn(c(a);a) - \Gfn_{\max}(a) > \epsilon,
\end{equation}
where $\Gfn_{\max}(a)$ is defined in \eqref{defn_of_G_max}. Then by \eqref{eq:f_ximax_right_to_c(a)} with $\bfepsilon = \epsilon$, for all $u \in [c(a)+\epsilon, \redge+\epsilon^{-1}]$ we have 
\begin{equation} \label{eq:pdf_of_u_sub_outer_near}
f_{\ximax}(u) = \bfC_{n,\beta} e^{\frac{\beta}{2}n \Gfn_{\max}(a)} O(1) = \bfC_{n,\beta} e^{\frac{\beta}{2} n\Hfn(c(a);a)} O(e^{-\frac{\beta}{2}\epsilon n}).
\end{equation}
We assume that for $u, v \in [c(a)-\epsilon, c(a)+\epsilon]$
\begin{equation} \label{eq:inequality_of_epsilon_sub_condition_2}
  \Hfn(v+\epsilon;a) + (-V(u)+2\gfn(v)-\ell) < \Hfn(c(a);a) - \epsilon.
\end{equation}
Then by \eqref{eq:f_ximax_left_to_c(a)} with $\bfepsilon = \epsilon$, for $u \in [c(a)-\epsilon, c(a)+\epsilon]$
\begin{equation} \label{eq:pdf_of_u_sub_around_c(a)}
f_{\ximax}(u) = \bfC_{n,\beta} e^{\frac{\beta}{2} n\Hfn(c(a);a)} O(e^{-\frac{\beta}{2}\epsilon n}).
\end{equation}
We assume that for $u \geq \redge + \epsilon^{-1}$, 
\begin{equation} \label{eq:added_last}
-V(u) + (2+a)u + a + \tilde{C}_V + \int V(x) d\mu(x) < \Hfn(c(a);a) + (\redge + \epsilon^{-1} -u) - 1.
\end{equation}
Then by \eqref{eq:estimate_of_prob_ximax_far_right} with $\bfepsilon = \epsilon$, for $u \geq \redge + \epsilon^{-1}$ we have uniformly in $u$
\begin{equation} \label{eq:pdf_of_u_sub_outer_far}
f_{\ximax}(u) = \bfC_{n,\beta} e^{\frac{\beta}{2} n\Hfn(c(a);a)} O(e^{\frac{\beta}{2}n(\redge + \epsilon^{-1}-u-1)}).
\end{equation}
Let $\tilde{C}$ be large enough such that
\begin{equation}
a(\redge+\epsilon) -\tilde{C} - V(u) + \int V(x) d\mu(x) < \Hfn(c(a);a) + \redge -\epsilon -u -1.
\end{equation}
Then by \eqref{eq:estimate_of_prob_ximax_far_left} with $\bfepsilon = \epsilon$ and $C_{\bfepsilon} = \tilde{C}$, for $u \leq \redge -\epsilon$ we have uniformly in $u$ 
\begin{equation} \label{eq:pdf_of_u_sub_left}
f_{\ximax}(u) = \bfC_{n,\beta} e^{\frac{\beta}{2} n\Hfn(c(a);a)} O(e^{\frac{\beta}{2}n(\redge -\epsilon - u - 1)}).
\end{equation}
By \eqref{eq:pdf_of_u_sub_outer_near}, \eqref{eq:pdf_of_u_sub_around_c(a)}, \eqref{eq:pdf_of_u_sub_outer_far}, \eqref{eq:pdf_of_u_sub_left}, we find that
\begin{equation} \label{eq:sub_Prob_other}
\Prob(\ximax \geq c(a)-\epsilon \quad \textnormal{or} \quad \ximax \leq \redge - \epsilon) = \bfC_{n,\beta} e^{\frac{\beta}{2}n \Hfn(c(a);a)} O(e^{-\frac{\beta}{2}\epsilon n}).
\end{equation}

Let $\epsilon' < \epsilon/2$ be a small positive number such that
\begin{itemize}
\item $-V(u)+2\gfn(u)-\ell$ is decreasing on $[\redge, \redge+2\epsilon']$.
\item $-V(u)+2\gfn(u)-\ell$ attains its maximum on $[\redge+2\epsilon', c(a)]$ at $\redge+2\epsilon'$.
\item Let $\epsilon'' := -(-V(\redge+2\epsilon') + 2\gfn(\redge+2\epsilon')-\ell)$. Then $\epsilon'' < \epsilon$.
\end{itemize}
Then by \eqref{eq:asy_of_f_ximax_near} with $\bfepsilon = \epsilon'$, we have that 
\begin{equation} \label{eq:sub_Prob_closer}
\Prob(\ximax \in [\redge+2\epsilon', c(a)-\epsilon]) = \bfC_{n,\beta} e^{\frac{\beta}{2}n \Hfn(c(a);a)} O(e^{\frac{\beta}{2}\epsilon'' n}),
\end{equation}
and as $n \to \infty$
\begin{equation} \label{eq:sub_Prob_essential}
\frac{\Prob(\ximax \in [\redge+\epsilon', \redge+2\epsilon'])}{\bfC_{n,\beta} e^{\frac{\beta}{2}n \Hfn(c(a);a)} e^{\frac{\beta}{2}\epsilon'' n}} \to \infty.
\end{equation}

The probabilities \eqref{eq:sub_Prob_other}, \eqref{eq:sub_Prob_closer} and \eqref{eq:sub_Prob_essential} imply that the conditional probability
\begin{equation} \label{eq:sub_conditional_probability_convergence}
  \Prob(\ximax \in [\redge+\epsilon', \redge+2\epsilon'] \mid \ximax \in \realR \setminus (\redge-\epsilon, \redge+\epsilon')) \to 1.
\end{equation}
Since $[\redge+\epsilon', \redge+\epsilon'] \cup (\redge-\epsilon, \redge+2\epsilon') \in [\redge-\epsilon, \redge+\epsilon]$, \eqref{eq:sub_conditional_probability_convergence} implies that
\begin{equation} \label{eq:final_for_proof_of_thm_1.1a_first}
\Prob(\ximax \in [\redge-\epsilon, \redge+\epsilon]) \to 1.
\end{equation}
Taking $\epsilon$ arbitrarily small, we prove Theorem \ref{thm:limiting_position}\ref{enu:thm:limiting_position:1} when $0 < a < \acc$.

\bigskip

The case when $a = \acc = \frac{1}{2}V'(\redge)$ and $\acc \notin \mathcal{J}_V$ is similar. Let $\epsilon$ be a small enough positive number. Since $a = \frac{1}{2}V'(\redge)$, we have $c(a) = \redge$. Since $a = \acc \notin \mathcal{J}_V$, there exists $\epsilon' > 0$ depending on $\epsilon$ such that for all $u > \redge + \epsilon$, $\Gfn(u;a) < \Hfn(\redge;a) - \epsilon'$. Thus like \eqref{eq:pdf_of_u_sub_outer_near}, for $u \in [c(a)+\epsilon, \redge+\epsilon^{-1}]$ we have
\begin{equation}
  f_{\ximax}(u) = \bfC_{n,\beta} e^{\frac{\beta}{2} n\Hfn(\redge;a)} O(e^{-\frac{\beta}{2}\epsilon' n}).
\end{equation}
When $\epsilon$ is small enough, \eqref{eq:pdf_of_u_sub_outer_far} and \eqref{eq:pdf_of_u_sub_left} also hold. Then by arguments similar to \eqref{eq:sub_Prob_other}--\eqref{eq:final_for_proof_of_thm_1.1a_first}, we prove Theorem \ref{thm:limiting_position}\ref{enu:thm:limiting_position:1} when $a = \acc = \frac{1}{2}V'(\redge)$ and $\acc \notin \mathcal{J}_V$.

\subsection{Proof of Theorems \ref{thm:limiting_distr_nondegenerate} and \ref{thm:limiting_position}\ref{enu:thm:limiting_position:2} when $a > \acc$ and $a \not\in \mathcal{J}_V$} \label{subsec:a>acc_and_a_notin_J_V}

Let $\epsilon$ be a small enough positive constant such that the maximizer $x_0(a)$ of $\Gfn(x;a)$ in $[c(a), \infty)$ is less than $\redge + \epsilon^{-1}$, and the inequalities \eqref{eq:epsilon_less_than_G-H}, \eqref{eq:epsilon_and_G,H,V} and \eqref{eq:not_last_probably} are satisfied.

First we consider the case that $\acc < a < \frac{1}{2}V'(\redge)$, \ie, $c(a) > \redge$. We assume that
\begin{equation} \label{eq:epsilon_less_than_G-H}
  \Gfn_{\max}(a) - \Hfn(c(a);a) > \epsilon. 
\end{equation}
Then by \eqref{eq:asy_of_f_ximax_near} with $\bfepsilon = \epsilon$, for $u \in (\redge+\epsilon, c(a)-\epsilon)$ we have
\begin{equation} \label{eq:pdf_of_u_super_inner}
  f_{\ximax}(u) = \bfC_{n,\beta} n^{(1-\beta)/2} e^{\frac{\beta}{2}n \Hfn(c(a);a)} O(1) = \bfC_{n,\beta} e^{\frac{\beta}{2}n \Gfn_{\max}(a)} O(e^{-\frac{\beta}{2}\epsilon n}).
\end{equation}
We assume that for all $u \in [\redge-\epsilon, \redge+\epsilon]$
\begin{equation}
  \label{eq:4}
  \Hfn(c(a)+2\epsilon;a) + (-V(\redge+\epsilon) + 2\gfn(\redge+\epsilon) - \ell) + (V(\redge+\epsilon)-V(u)) < \Gfn_{\max}(a) - \epsilon.
\end{equation}
Then by \eqref{eq:est_of_f_ximax_around_b_2} with $\bfepsilon = \epsilon$, for $u \in [\redge-\epsilon, \redge+\epsilon]$ we have
\begin{equation} \label{eq:pdf_of_u_super_around_redge}
  f_{\ximax}(u) = \bfC_{n,\beta} e^{\frac{\beta}{2}n \Gfn_{\max}(a)} O(e^{-\frac{\beta}{2}\epsilon n}).
\end{equation}
We assume that for all $u \in [c(a)-\epsilon, c(a)+\epsilon]$
\begin{equation} \label{eq:epsilon_and_G,H,V}
  \Hfn(c(a)+2\epsilon;a) + (-V(c(a)+\epsilon) + 2\gfn(c(a)+\epsilon) - \ell) + (V(c(a)+\epsilon)-V(u)) < \Gfn_{\max} - \epsilon.
\end{equation}
Then by \eqref{eq:f_ximax_left_to_c(a)} with $\bfepsilon = \epsilon$, for $u \in [c(a)-\epsilon, c(a)+\epsilon]$ we have
\begin{equation} \label{eq:pdf_of_u_super_around_c(a)}
  f_{\ximax}(u) = \bfC_{n,\beta} e^{\frac{\beta}{2}n \Gfn_{\max}(a)} O(e^{-\frac{\beta}{2}\epsilon n}).
\end{equation}
We assume that for  $u \geq \redge + \epsilon^{-1}$, 
\begin{equation} \label{eq:not_last_probably}
-V(u) + (2+a)u + a + \tilde{C}_V + \int V(x) d\mu(x) < \Gfn_{\max}(a) + (\redge + \epsilon^{-1} -u) - 1.
\end{equation}
Then by \eqref{eq:estimate_of_prob_ximax_far_right} with $\bfepsilon = \epsilon$, for $u \geq \redge + \epsilon^{-1}$ we have uniformly in $u$
\begin{equation} \label{eq:pdf_of_u_super_outer_far}
f_{\ximax}(u) = \bfC_{n,\beta} e^{\frac{\beta}{2} n\Gfn_{\max}(n)} O(e^{\frac{\beta}{2}n(\redge + \epsilon^{-1}-u-1)}).
\end{equation}
Let $\tilde{C}$ be large enough such that
\begin{equation}
a(\redge+\epsilon) -\tilde{C} - V(u) + \int V(x) d\mu(x) < \Gfn_{\max}(a) + \redge -\epsilon -u -1.
\end{equation}
Then by \eqref{eq:estimate_of_prob_ximax_far_left} with $\bfepsilon = \epsilon$ and $C_{\bfepsilon} = \tilde{C}$, for $u \leq \redge -\epsilon$ we have uniformly in $u$ 
\begin{equation} \label{eq:pdf_of_u_super_left}
f_{\ximax}(u) = \bfC_{n,\beta} e^{\frac{\beta}{2} n\Gfn_{\max}(a)} O(e^{\frac{\beta}{2}n(\redge -\epsilon - u - 1)}).
\end{equation}
By \eqref{eq:pdf_of_u_super_inner}, \eqref{eq:pdf_of_u_super_around_redge}, \eqref{eq:pdf_of_u_super_around_c(a)}, \eqref{eq:pdf_of_u_super_outer_far} and \eqref{eq:pdf_of_u_super_left}, we find that

\begin{equation} \label{eq:super_Prob_other}
\Prob(\ximax \geq \redge+\epsilon^{-1} \textnormal{ or } \ximax \leq c(a)+\epsilon) = \bfC_{n,\beta} e^{\frac{\beta}{2}n \Gfn_{\max}(a)} O(e^{-\frac{\beta}{2}\epsilon n}).
\end{equation}
In the case that $a \geq \frac{1}{2}V'(\redge)$, \ie, $c(a)=\redge$, we find that inequalities \eqref{eq:pdf_of_u_super_around_redge}, \eqref{eq:pdf_of_u_super_outer_far} and \eqref{eq:pdf_of_u_super_left} still hold, and the estimate \eqref{eq:super_Prob_other} holds with $c(a) = \redge$.

For $u \in [c(a)+\epsilon, \redge+\epsilon^{-1}]$, we have the asymptotic formula \eqref{eq:f_ximax_right_to_c(a)} and $x_0(a)$ is the unique maximum of $\Gfn(x;a)$ in $[c(a)+\epsilon, \redge+\epsilon^{-1}]$. If we further assume that $\Gfn''(x_0(a);a) \neq 0$, by the standard Laplace's method we have that
\begin{equation} \label{eq:super_Prob_essential}
  \Prob(\ximax \in [c(a)+\epsilon, \redge+\epsilon^{-1}]) = \sqrt{\frac{2\pi}{-\frac{\beta}{2}\Gfn''(x_0(a);a)n}} \M_{\beta}(x_0(a)) \bfC_{n,\beta}e^{\frac{\beta}{2}n\Gfn_{\max}(a)} (1+o(1)),
\end{equation}
and for any $T \in \realR$
\begin{multline} \label{eq:super_Prob_essential_part}
  \Prob \left( \ximax \in \left[ c(a)+\epsilon, x_0(a)+ \frac{T}{\sqrt{-\frac{\beta}{2}\Gfn''(x_0(a);a)n}} \right] \right) = \\
  \sqrt{\frac{2\pi}{-\frac{\beta}{2}\Gfn''(x_0(a);a)n}} \M_{\beta}(x_0(a)) \bfC_{n,\beta}e^{\frac{\beta}{2}n\Gfn_{\max}(a)} (\erf(T)+o(1)).
\end{multline}
The probabilities \eqref{eq:super_Prob_other}, \eqref{eq:super_Prob_essential} and \eqref{eq:super_Prob_essential_part} imply Theorem \ref{thm:limiting_distr_nondegenerate}\ref{enu:thm:limiting_distr_nondegenerate:1}.

If the second derivative of $\Gfn(x;a)$ vanishes at $x_0(a)$, due to the analyticity of $\Gfn(x;a)$, there exists $k > 1$ such that $\Gfn^{(j)}(x_0(a);a) = 0$ for $j=1, \dots, 2k-1$ and $\Gfn^{(2k)}(x_0(a);a) \neq 0$. By the Laplace's method we have 
\begin{multline} \label{eq:super_Prob_essential_exceptional}
  \Prob(\ximax \in [c(a)+\epsilon, \redge+\epsilon^{-1}]) = \left( \frac{(2k)!}{-\frac{\beta}{2}\Gfn^{(2k)}(x_0(a);a)n} \right)^{1/(2k)} \\
  \times \int^{\infty}_{-\infty} e^{-x^{2k}} dx \M_{\beta}(x_0(a)) \bfC_{n,\beta}e^{\frac{\beta}{2}n\Gfn_{\max}(a)} (1+o(1)),
\end{multline}
and for any $T \in \realR$
\begin{multline} \label{eq:super_Prob_part_essential_exceptional}
  \Prob \left( \ximax \in \left[ c(a)+\epsilon, x_0 + \left( \frac{-\frac{\beta}{2}\Gfn^{(2k)}(x_0(a);a)n}{(2k)!} \right)^{-1/(2k)} T \right] \right) \\
  = \left( \frac{(2k)!}{-\frac{\beta}{2}\Gfn^{(2k)}(x_0(a);a)n} \right)^{1/(2k)} \int^{T}_{-\infty} e^{-x^{2k}} dx \M_{\beta}(x_0(a)) \bfC_{n,\beta}e^{\frac{\beta}{2}n\Gfn_{\max}(a)} (1+o(1)).
\end{multline}
The probabilities \eqref{eq:super_Prob_other}, \eqref{eq:super_Prob_essential_exceptional} and \eqref{eq:super_Prob_part_essential_exceptional} imply Theorem \ref{thm:limiting_distr_nondegenerate}\ref{enu:thm:limiting_distr_nondegenerate:2}. Finally, Theorem \ref{thm:limiting_distr_nondegenerate} implies Theorem \ref{thm:limiting_position}\ref{enu:thm:limiting_position:2}.

\subsection{Proofs of Theorem \ref{thm:two_normal_max_points}, \ref{thm:limiting_position}\ref{enu:thm:limiting_position:3} and \ref{thm:one_abnormal_max_points} when $a > \acc$ and $a \in \mathcal{J}_V$} \label{subsec:a>acc_and_a_in_J_V}

Let $a_0 > \acc$, and $\epsilon$ be a small positive constant such that the inequalities \eqref{eq:pdf_of_u_super_inner}, \eqref{eq:pdf_of_u_super_around_redge}, \eqref{eq:pdf_of_u_super_around_c(a)}, \eqref{eq:pdf_of_u_super_outer_far} and \eqref{eq:pdf_of_u_super_left} hold with $a = a_0$. It is easy to verify that there exists a positive number $\bar{\epsilon}$ depending on $\epsilon$ such that if we take $a = a_0 + \epsilon'$ with $\epsilon' \in [-\bar{\epsilon}, \bar{\epsilon}]$, the inequalities  \eqref{eq:pdf_of_u_super_inner}, \eqref{eq:pdf_of_u_super_around_redge}, \eqref{eq:pdf_of_u_super_around_c(a)}, \eqref{eq:pdf_of_u_super_outer_far} and \eqref{eq:pdf_of_u_super_left} still hold with the same $\epsilon$. Thus the estimate of probability \eqref{eq:super_Prob_other} still holds with $a = a_0 + \epsilon'$. If we further assume that
\begin{align}
  c(a) > & c(a_0) - \frac{\epsilon}{2}, \\
  \Gfn_{\max}(a) > & \Gfn_{\max}(a_0) - \frac{\epsilon}{2},
\end{align}
by \eqref{eq:super_Prob_other} we obtain that 
\begin{equation} \label{eq:super_Prob_other_two_normal_point}
  \Prob(\ximax \geq \redge+\epsilon^{-1} \textnormal{ or } \ximax \leq c(a_0)+\frac{\epsilon}{2}) = \bfC_{n,\beta} e^{\frac{\beta}{2}n \Gfn_{\max}(a_0)} O(e^{-\frac{\beta}{4}\epsilon n}).
\end{equation}

First we assume that $a_0 \in \mathcal{J}_V$ and $\Gfn(x;a_0)$ has $r \geq 2$ maximizers $x_1(a_0) < x_2(a_0) < \dots < x_r(a_0)$ in $(c(a_0), \infty)$, and all of them are less than $\redge+\epsilon^{-1}$. Further we assume that for all $i = 1, \dots, r$
\begin{equation}
  \Gfn''(x_i(a_0)) \neq 0.
\end{equation}
We take
\begin{equation} \label{eq:double_scaling_of_a}
  a = a_0 + \frac{\alpha}{n},
\end{equation}
where $\alpha$ is in a compact subset of $\realR$. Since $a \in [a_0-\bar{\epsilon}, a_0+\bar{\epsilon}]$ for sufficiently large $n$, the estimate \eqref{eq:super_Prob_other_two_normal_point} is applicable to $a$.

For $x$ around $x_i(a_0)$ ($i=1, \dots, r$), we denote
\begin{equation} \label{eq:relation_between_x_and_xi_two_normal_max}
  x = x_i(a_0) + \frac{\xi_i}{\sqrt{-\frac{\beta}{2} \Gfn''(x_i(a_0);a_0) n}}.
\end{equation}
For $\xi_i$ in a compact subset of $\realR$ and $x$ given by \eqref{eq:relation_between_x_and_xi_two_normal_max}, we have
\begin{equation} \label{eq:Taylor_expansion_two_normal_points}
  \begin{split}
    \frac{\beta}{2}\Gfn(x;a) = & \frac{\beta}{2}\Gfn(x_i(a);a) - \frac{\xi^2_i}{2}\frac{1}{n} + O(n^{-2}) \\
    = & \frac{\beta}{2}\Gfn_{\max}(a_0) + \frac{\beta}{2}\alpha x_i(a_0) \frac{1}{n} - \frac{\xi^2_i}{2}\frac{1}{n} + O(n^{-2}).
  \end{split}
\end{equation}
Let $\epsilon_i$ ($i=1, \dots, r$) be small enough constant numbers such that $x_i(a_0)$ is the unique maximum of $\Gfn(x;a_0)$ in $[x_i(a_0)-\epsilon_i, x_i(a_0)+\epsilon_i]$. Applying the standard Laplace's method to \eqref{eq:f_ximax_right_to_c(a)}, near $x_i(a_0)$ ($i=1,2$), we obtain that
\begin{multline} \label{eq:super_Prob_around_x_i_two_normal_point}
  \Prob(\ximax \in [x_i(a_0)-\epsilon_i, x_i(a_0)+\epsilon_i]) = e^{\frac{\beta}{2}\alpha x_i(a_0) + \frac{\xi^2_i}{2}} \\
  \times \sqrt{\frac{2\pi}{-\frac{\beta}{2} \Gfn''(x_i(a_0);a_0)n}} \M_{\beta}(x_i(a_0)) \bfC_{n,\beta} e^{\frac{\beta}{2}n\Gfn_{\max}(a_0)} (1+o(1)),
\end{multline}
and for any $T$ in a compact subset of $\realR$,
\begin{multline} \label{eq:super_Prob_around_x_i_two_normal_point_partial}
  \Prob \left( \ximax \in \left[ x_i(a_0)-\epsilon_i, x_i(a_0)+\frac{T}{\sqrt{-\frac{\beta}{2}\Gfn''(x_i(a_0);a_0)n}} \right] \right) = \\
  e^{\frac{\beta}{2}\alpha x_i(a_0) + \frac{\xi^2_i}{2}} \sqrt{\frac{2\pi}{-\frac{\beta}{2} \Gfn''(x_i(a_0);a_0)n}} \M_{\beta}(x_i(a_0)) \bfC_{n,\beta} e^{\frac{\beta}{2}n\Gfn_{\max}(a_0)} (\erf(T)+o(1)),
\end{multline}
There exists $\epsilon'' > 0$ depending on $\epsilon_1, \dots, \epsilon_r$ such that for sufficiently large $n$
\begin{equation} \label{eq:super_Prob_around_x_i_two_normal_point_not_near}
  \Gfn(x;a) < \Gfn_{\max}(a_0) - \epsilon'' \textnormal{ for $x \in [c(a_0)+\frac{\epsilon}{2}, \redge+\epsilon^{-1}]$ but $x \not\in (x_i(a_0)-\epsilon_i, x_i(a_0)+\epsilon_i)$.}
\end{equation}
Then we find that the probability
\begin{equation} \label{eq:small_prob_when_x_not_near_x_1_x_2_normal}
  \Prob \left( \ximax \in [c(a_0)+\frac{\epsilon}{2}, \redge+\epsilon^{-1}] \setminus \bigcup^r_{i=1} (x_i(a_0)-\epsilon_i, x_i(a_0)+\epsilon_i) \right) = \bfC_{n,\beta} e^{\frac{\beta}{2}n\Gfn_{\max}(a_0)} O(e^{-\frac{\beta}{2}\epsilon''n}).
\end{equation}

For $i=1, \dots, r$, let 
\begin{equation} \label{eq:defn_of_B_i_two_normal_points}
  B_{i,\beta}(\alpha) := e^{\frac{\beta}{2}\alpha x_i(a_0)} \sqrt{\frac{2\pi}{-\frac{\beta}{2} \Gfn''(x_i(a_0);a_0)}} \M_{\beta}(x_i(a_0))
\end{equation}
and
\begin{equation} \label{eq:defn_of_p_i_two_normal_points}
  p_{i,\beta}(\alpha) := \frac{B_{i,\beta}(\alpha)}{\sum^r_{j=1}B_{j,\beta}(\alpha)}.
\end{equation}
From \eqref{eq:defn_of_p_i_two_normal_points} we immediately find $\sum^r_{j=1} p_{j,\beta}(\alpha) = 1$. By \eqref{eq:defn_of_B_i_two_normal_points} and \eqref{eq:defn_of_p_i_two_normal_points} we find that $\lim_{\alpha \to \infty} p_{r,\beta}(\alpha) = 1$ and  $\lim_{\alpha \to -\infty} p_{1,\beta}(\alpha) = 1$.

The probabilities \eqref{eq:super_Prob_other_two_normal_point}, \eqref{eq:super_Prob_around_x_i_two_normal_point}, \eqref{eq:super_Prob_around_x_i_two_normal_point_partial} and \eqref{eq:small_prob_when_x_not_near_x_1_x_2_normal} show that for any $T \in \realR$ and $i = 1, \dots, r$
\begin{equation}
  \Prob \left( \ximax \leq x_i(a_0)+\frac{T}{\sqrt{-\frac{\beta}{2}\Gfn''(x_i(a_0);a_0)n}} \right) = \left( \sum^{i-1}_{j=1} p_{j,\beta}(\alpha) \right) + p_{i,\beta}(\alpha) (\erf(T)+o(1)).
\end{equation}
Therefore Theorem \ref{thm:two_normal_max_points} is proved. Theorem \ref{thm:limiting_position}\ref{enu:thm:limiting_position:3} is a consequence of Theorem \ref{thm:two_normal_max_points} with $\alpha = 0$.

\bigskip

Next we consider the case that $r = 2$ and for $k > 1$
\begin{gather}
  \Gfn''(x_1(a_0)) \neq 0, \\
  \Gfn^{(2k)}(x_2(a_0)) \neq 0, \quad \Gfn^{(j)}(x_2(a_0)) = 0 \quad j=1, \dots, 2k-1.
\end{gather}
and take (see \eqref{eq:defn_of_q_and_a})
\begin{equation}
  a = a_0 - q_{\beta}\frac{\log n}{n} + \frac{\alpha}{n}, \quad \textnormal{where} \quad q_{\beta} := \frac{2}{\beta}\frac{\frac{1}{2} - \frac{1}{2k}}{x_2(a_0) - x_1(a_0)},
\end{equation}
where $\alpha$ is in a compact subset of $\realR$. For $x$ around $x_1(a)$, we denote
\begin{equation} \label{eq:x_by_xi_around_x_1_one_abnormal}
  x = x_1(a_0) + \frac{\xi_1}{\sqrt{-\frac{\beta}{2}\Gfn''(x_1(a_0);a_0)n}},
\end{equation}
and for $x$ around $x_2(a)$, we denote
\begin{equation} \label{eq:x_by_xi_around_x_2_one_abnormal}
  x = x_2(a_0) + \left( \frac{-\frac{\beta}{2} \Gfn^{(2k)}(x_2(a_0);a_0) n}{(2k)!} \right)^{-1/(2k)} \xi_2.
\end{equation}
For $\xi_1$ in a compact subset of $\realR$ and $x$ given by \eqref{eq:x_by_xi_around_x_1_one_abnormal}, we have like \eqref{eq:Taylor_expansion_two_normal_points}
\begin{equation}
  \begin{split}
    \frac{\beta}{2}\Gfn(x;a) = & \frac{\beta}{2}\Gfn(x_1(a_0);a) - \frac{\xi^2_1}{2} \frac{1}{n} + o(n^{-1}) \\
    = & \frac{\beta}{2}\Gfn_{\max}(a_0) + \frac{\beta}{2}x_1(a_0) \frac{-q_{\beta}\log n + \alpha}{n} - \frac{\xi^2_1}{2} \frac{1}{n} + o(n^{-1}).
  \end{split}
\end{equation}
For $\xi_2$ in a compact subset of $\realR$ and $x$ given by \eqref{eq:x_by_xi_around_x_2_one_abnormal}, we have
\begin{equation}
  \begin{split}
    \frac{\beta}{2}\Gfn(x;a) = & \frac{\beta}{2}\Gfn(x_2(a_0);a) - \xi^{2k}_2 \frac{1}{n} + o(n^{-1}) \\
    = & \frac{\beta}{2}\Gfn_{\max}(a_0) + \frac{\beta}{2}x_2(a_0) \frac{-q_{\beta}\log n + \alpha}{n} - \xi^{2k}_2 \frac{1}{n} + o(n^{-1}).
  \end{split}
\end{equation}
Let $\epsilon_i$ ($i=1,2$) be small enough constant numbers such that $x_i(a_0)$ is the unique maximum of $\Gfn(x;a_0)$ in $[x_i(a_0)-\epsilon_i, x_i(a_0)+\epsilon_i]$. Applying the standard Laplace's method to \eqref{eq:f_ximax_right_to_c(a)}, near $x_1(a_0)$, we obtain similar to \eqref{eq:super_Prob_around_x_i_two_normal_point} and \eqref{eq:super_Prob_around_x_i_two_normal_point_partial} that

\begin{multline} \label{eq:super_Prob_around_x_1_one_bnormal_point}
  \Prob(\ximax \in [x_1(a_0)-\epsilon_1, x_1(a_0)+\epsilon_1]) = \frac{e^{-\frac{\beta}{2}x_1(a_0) q _{\beta}\log n}}{\sqrt{n}} e^{\frac{\beta}{2}x_1(a_0)\alpha} \\
  \times \sqrt{\frac{2\pi}{-\frac{\beta}{2} \Gfn''(x_1(a_0);a_0)}} \M_{\beta}(x_1(a_0)) \bfC_{n,\beta} e^{\frac{\beta}{2}n\Gfn_{\max}(a_0)} (1+o(1)),
\end{multline}
and for any $T$ in a compact subset of $\realR$,
\begin{multline} \label{eq:super_Prob_around_x_1_one_bnormal_point_partial}
  \Prob \left( \ximax \in \left[ x_1(a_0)-\epsilon_1, x_1(a_0)+\frac{T}{\sqrt{-\frac{\beta}{2}\Gfn''(x_i(a_0);a_0)n}} \right] \right) = \\
   \frac{e^{-\frac{\beta}{2}x_1(a_0) q_ {\beta}\log n}}{\sqrt{n}} e^{\frac{\beta}{2}x_1(a_0)\alpha} \sqrt{\frac{2\pi}{-\frac{\beta}{2} \Gfn''(x_2(a_0);a_0)}} \M_{\beta}(x_2(a_0)) \bfC_{n,\beta} e^{\frac{\beta}{2}n\Gfn_{\max}(a_0)} (\erf(T)+o(1)).
\end{multline}
Applying the Laplace's method to \eqref{eq:f_ximax_right_to_c(a)}, near $x_2(a_0)$, we obtain
\begin{multline} \label{eq:super_Prob_around_x_2_one_bnormal_point}
  \Prob(\ximax \in [x_2(a_0)-\epsilon_2, x_2(a_0)+\epsilon_2]) = \frac{e^{-\frac{\beta}{2}x_2(a_0) q_{\beta}\log n}}{n^{1/(2k)}} e^{\frac{\beta}{2}x_2(a_0)\alpha} \\
  \times \left( \frac{(2k)!}{-\frac{\beta}{2} \Gfn^{(2k)}(x_2(a_0);a_0)} \right)^{1/(2k)} \M_{\beta}(x_2(a_0)) \bfC_{n,\beta} e^{\frac{\beta}{2}n\Gfn_{\max}(a_0)} \left( \int^{\infty}_{-\infty} e^{\xi^{2k}_2} d\xi_2 +o(1) \right),
\end{multline}
and for any $T$ in a compact subset of $\realR$,
\begin{multline} \label{eq:super_Prob_around_x_2_one_bnormal_point_partial}
  \Prob \left( \ximax \in \left[ x_2(a_0)-\epsilon_2, x_2(a) + \left( \frac{-\frac{\beta}{2}\Gfn^{(2k)}(x_0(a);a)n}{(2k)!} \right)^{-1/(2k)} T \right] \right) = \\
  \frac{e^{-\frac{\beta}{2}x_2(a_0) q_{\beta}\log n}}{n^{1/(2k)}} e^{\frac{\beta}{2}x_2(a_0)\alpha} \left( \frac{(2k)!}{-\frac{\beta}{2} \Gfn^{(2k)}(x_1(a_0);a_0)} \right)^{1/(2k)} \M_{\beta}(x_2(a_0)) \bfC_{n,\beta} e^{\frac{\beta}{2}n\Gfn_{\max}(a_0)} \\
  \times \left( \int^{T}_{-\infty} e^{\xi^{2k}_2} d\xi_2 +o(1) \right).
\end{multline}
Also there exists $\epsilon'' > 0$ depending on $\epsilon_1$ and $\epsilon_2$ such that the estimate \eqref{eq:small_prob_when_x_not_near_x_1_x_2_normal} holds. The probabilities \eqref{eq:super_Prob_other_two_normal_point}, \eqref{eq:small_prob_when_x_not_near_x_1_x_2_normal}, \eqref{eq:super_Prob_around_x_1_one_bnormal_point}, \eqref{eq:super_Prob_around_x_1_one_bnormal_point_partial}, \eqref{eq:super_Prob_around_x_2_one_bnormal_point}, \eqref{eq:super_Prob_around_x_2_one_bnormal_point_partial} show that the probability that $\ximax$ is in $[x_1-\epsilon_1, x_1+\epsilon_1]$ or $[x_2-\epsilon_2, x_2+\epsilon_2]$ approaches $1$ as $n \to \infty$.

Let
\begin{align}
  D_{1,\beta} := & e^{\frac{\beta}{2}x_1(a_0)\alpha} \sqrt{\frac{2\pi}{-\frac{\beta}{2}\Gfn''(x_1(a_0);a_0)}} \M_{\beta}(x_1(a_0)), \label{eq:formula_of_D_1,beta} \\
  D_{2,\beta} := & e^{\frac{\beta}{2}x_2(a_0)\alpha} \left( \frac{(2k)!}{-\frac{\beta}{2} \Gfn^{(2k)}(x_2(a_0);a_0)} \right)^{1/(2k)} \M_{\beta}(x_2(a_0)) \int^{\infty}_{-\infty} e^{\xi^{2k}_2} d\xi_2, \label{eq:formula_of_D_2,beta}
\end{align}
and for $i = 1,2$
\begin{equation} \label{eq:defn_of_tilde_p}
  \tilde{p}_{i,\beta}(\alpha) := \frac{D_{i,\beta}(\alpha)}{D_{1,\beta}(\alpha) + D_{2,\beta}(\alpha)}.
\end{equation}
From \eqref{eq:defn_of_tilde_p} we immediately find that $\tilde{p}_{1,\beta}(\alpha) + \tilde{p}_{2,\beta}(\alpha) = 1$. By \eqref{eq:formula_of_D_1,beta}, \eqref{eq:formula_of_D_2,beta} and \eqref{eq:defn_of_tilde_p}, we find that $\lim_{\alpha \to \infty} \tilde{p}_{2,\beta}(\alpha) = 1$ and $\lim_{\alpha \to -\infty} \tilde{p}_{1,\beta}(\alpha) = 1$.

Because
\begin{equation}
  \frac{e^{-\frac{\beta}{2}x_1(a_0) q_{\beta} \log n}}{\sqrt{n}} = \frac{e^{-\frac{\beta}{2}x_2(a_0) q_{\beta} \log n}}{n^{1/(2k)}},
\end{equation}
We further find from \eqref{eq:super_Prob_around_x_1_one_bnormal_point}, \eqref{eq:super_Prob_around_x_1_one_bnormal_point_partial}, \eqref{eq:super_Prob_around_x_2_one_bnormal_point}, \eqref{eq:super_Prob_around_x_2_one_bnormal_point_partial} that
for any $T \in \realR$
\begin{equation}
  \Prob \left( \ximax \leq x_1(a_0)+\frac{T}{\sqrt{-\frac{\beta}{2}\Gfn''(x_1(a_0);a_0)n}} \right) = \tilde{p}_{1,\beta}(\alpha) (\erf(T)+o(1)),
\end{equation}
\begin{multline}
  \Prob \left( \ximax \leq x_2(a) + \left( \frac{-\frac{\beta}{2}\Gfn^{(2k)}(x_0(a);a)n}{(2k)!} \right)^{-1/(2k)} T \right) = \\
  \tilde{p}_{1,\beta}(\alpha) + \tilde{p}_{2,\beta}(\alpha) \left( \frac{\int^T_{-\infty} e^{-x^{2k}} dx}{\int^{\infty}_{-\infty} e^{-x^{2k}} dx} + o(1) \right).
\end{multline}
Therefore Theorem \ref{thm:one_abnormal_max_points} is proved.

\section{Proof of Proposition \ref{cor:outlier_micro}} \label{sec:proof_of_corollary_of_Johansson}

The proof of Proposition \ref{cor:outlier_micro} is based on a theorem of Johansson \cite{Johansson98}. For the convenience of readers we state it bellow.
\begin{prop}{\cite[Theorem 2.4]{Johansson98}} \label{prop:Johansson}
  Suppose the function $V$ satisfies Conditions \ref{condition:1}--\ref{condition:3}, and $\mu$ is the equilibrium measure associated to $V$. Let $f$ be a real function that satisfies conditions \ref{enu:Johansson_condition_1}--\ref{enu:Johansson_condition_3} below, with $s = 2$ if $\beta = 2$ and $s = 17/2$ if $\beta \neq 2$. Then there are a quadratic functional $A$ on $f$ and a signed measure $\nu$ on $\supp(\mu) = J = [b_1, b_2]$ which do not depend on $n$, such that as $n \to \infty$
  \begin{equation}
    \log \E_{n-1,\beta}(e^{\sum^{n-1}_{j=1} f(x_j)}) = (n-1) \int f(x) d\mu(x) + \left( \frac{2}{\beta} - 1 \right) \int f(x) d\nu(x) + \frac{2}{\beta}A(f) + o(1).
  \end{equation}
\end{prop}
The quadratic functional $A$ is defined by
\begin{equation} \label{eq:defn_of_quadratic_A}
A(f) = \frac{1}{2} \int_J f(x)\delta^f(x) dx
\end{equation}
where $\delta^f$ is given by
\begin{equation}
\delta^f(x) = -\frac{1}{2\pi^2} \frac{1}{\sqrt{(x-b_1)(b_2-x)}} \pv \int_J \frac{f'(s) \sqrt{(s-b_1)(b_2-s)}}{s-x} ds.
\end{equation}
From the quadratic functional $A$, we define the inner product $\langle \cdot, \cdot \rangle_A$ by
\begin{equation}
\langle f, g \rangle_A := \frac{1}{2} (A(f+g)-A(f)-A(g)).
\end{equation}
The explicit formula of $\nu(x)$ is more complicated and is given in \cite[Formula (3.54)]{Johansson98}. The conditions mentioned in Proposition \ref{prop:Johansson} are (see \cite[Page 157]{Johansson98})
\begin{enumerate}[label=(\roman*)]
\item \label{enu:Johansson_condition_1}
$f(x) \leq C(V(x) + 1)$ for some constant $C$, all $x \in \realR$.

\item $\lvert f'(x) \rvert \leq q(x)$ for some polynomial $q(x)$ and all $x \in \realR$.

\item \label{enu:Johansson_condition_3}
For any $x_0 > 0$, there is an $\alpha > 0$ such that $h\psi_{x_0} \in H^{s+\alpha}$, where $H^s$, $s>0$, is the standard $L^2$ Sobolev space, and $\psi_{x_0} \in C^{\infty}$ is the function such that $\psi_{x_0}(x) = 1$ if $\lvert x \rvert \leq x_0$, $\psi_{x_0}(x) = 0$ if $\lvert x \rvert \geq x_0+1$ and $0 \leq \psi_{x_0}(x) \leq 1$.
\end{enumerate}

The function $R_{\beta}(u,w)$ appearing in Proposition \ref{cor:outlier_micro} is defined by
\begin{equation} \label{eq:defn_of_R_beta(u,w)}
  R_{\beta}(u,w) := \exp \left[ \frac{\beta}{2}\left( \left( \frac{2}{\beta}-1 \right) \int p(x;u,w)d\nu(x) - \int p(x;u,w)d\mu(x) + A(p(x;u,w)) \right) \right]
\end{equation}
for $u \in (\redge, \infty)$ and $w \in \compC \setminus (-\infty, u)$. If $w=u$, we denote
\begin{equation} \label{eq:defn_of_R_beta(u)}
R_{\beta}(u) := R_{\beta}(u,u).
\end{equation}

To facilitate the proof of Proposition \ref{cor:outlier_micro}, we define
\begin{align}
  Z_{m,\beta}(u,w;c) := & \E_{m,\beta} \left( P_{m,\beta}(x_1, \dots, x_m; u,w) \prod^m_{j=1} \chi_{(-\infty,c)}(x_j) \right), \label{eq:definition_of_Z_n-1_beta_with_c} \\
  \hat{Z}_{m,\beta}(u,w;c) := & \E_{m,\beta} \left( \lvert P_{m,\beta}(x_1, \dots, x_m; u,w) \rvert \prod^m_{j=1} \chi_{(-\infty,c)}(x_j) \right), \label{eq:defn_of_Z_hat_(uwc)_with_c}
\end{align}
where $c$ is a parameter no greater than $u$. When $c=u$, $Z_{m,\beta}(u,w;c)$ and $\hat{Z}_{m,\beta}(u,w;c)$ become $Z_{m,\beta}(u,w)$ and $\hat{Z}_{m,\beta}(u,w)$.

The proof of Proposition \ref{cor:outlier_micro} is as follows. Recall that $\redge$ is the right edge of $J$, the support of the equilibrium measure. Let $\bar{c} = (\redge + u)/2$. We write
\begin{equation}
  Z_{n-1,\beta}(u,w) = \hat{Z}_{n-1,\beta}(u,w;\bar{c}) \left( \frac{Z_{n-1,\beta}(u,w;\bar{c})}{\hat{Z}_{n-1,\beta}(u,w;\bar{c})} + \frac{Z_{n-1,\beta}(u,w)-Z_{n-1,\beta}(u,w;\bar{c})}{\hat{Z}_{n-1,\beta}(u,w;\bar{c})} \right).
\end{equation}

In case \ref{enu:cor:outlier_micro:a} where $w$ is given by \eqref{eq:condition_enu:cor:outlier_micro:a}, we assume the results
\begin{align} 
  \hat{Z}_{n-1,\beta}(u,w) = & e^{-\frac{\beta s}{2} \int \frac{d\mu(x)}{u-x}} R_{\beta}(u) e^{\frac{\beta}{2}n \int p(x;u)d\mu(x)} (1+o(1)), \label{eq:1st_support_of_cor:outlier_micro} \\
  \hat{Z}_{n-1,\beta}(u,w;\bar{c}) = & e^{-\frac{\beta s}{2} \int \frac{d\mu(x)}{u-x}} R_{\beta}(u) e^{\frac{\beta}{2}n \int p(x;u)d\mu(x)} (1+o(1)), \label{eq:1st_support_of_cor:outlier_micro_2} \\
  \frac{Z_{n-1,\beta}(u,w;\bar{c})}{\hat{Z}_{n-1,\beta}(u,w;\bar{c})} = & e^{-\frac{i\beta t}{2} \int \frac{d\mu(x)}{u-x}} + o(1).  \label{eq:2nd_support_of_cor:outlier_micro}
\end{align}
Then we have
\begin{equation} \label{eq:1st_support_of_cor:outlier_micro_derived}
\left\lvert \frac{Z_{n-1,\beta}(u,w)-Z_{n-1,\beta}(u,w;\bar{c})}{\hat{Z}_{n-1,\beta}(u,w;\bar{c})} \right\rvert \leq \frac{\hat{Z}_{n-1,\beta}(u,w)-\hat{Z}_{n-1,\beta}(u,w;\bar{c})}{\hat{Z}_{n-1,\beta}(u,w;\bar{c})} = o(1),
\end{equation}
and find that \eqref{eq:enu:cor:outlier_micro:a} is the consequence of \eqref{eq:1st_support_of_cor:outlier_micro}, \eqref{eq:1st_support_of_cor:outlier_micro_2}, \eqref{eq:2nd_support_of_cor:outlier_micro} and \eqref{eq:1st_support_of_cor:outlier_micro_derived}. 

In case \ref{enu:cor:outlier_micro:c} where $w$ is given by \eqref{eq:condition_enu:cor:outlier_micro:c}, we assume the results
\begin{align} 
  \hat{Z}_{n-1,\beta}(u,w) = & e^{-\frac{\beta t^2}{4} \int \frac{d\mu(x)}{(w_0-x)^2}} R_{\beta}(u,w_0)  e^{\frac{\beta}{2}n \int p(x;u,w_0) d\mu(x)} (1+o(1)), \label{eq:1st_support_of_cor:outlier_micro_c} \\
  \hat{Z}_{n-1,\beta}(u,w;\bar{c}) = & e^{-\frac{\beta t^2}{4} \int \frac{d\mu(x)}{(w_0-x)^2}} R_{\beta}(u,w_0)  e^{\frac{\beta}{2}n \int p(x;u,w_0) d\mu(x)} (1+o(1)), \label{eq:1st_support_of_cor:outlier_micro_2_c} \\
  \frac{Z_{n-1,\beta}(u,w;\bar{c})}{\hat{Z}_{n-1,\beta}(u,w;\bar{c})} = & e^{-i\frac{\beta t}{2} \sqrt{n} \int \frac{d\mu(x)}{w_0-x}} + o(1).  \label{eq:2nd_support_of_cor:outlier_micro_c}
\end{align}
We still have \eqref{eq:1st_support_of_cor:outlier_micro_derived}, and \eqref{eq:enu:cor:outlier_micro:c} is the consequence of \eqref{eq:1st_support_of_cor:outlier_micro_c}, \eqref{eq:1st_support_of_cor:outlier_micro_2_c}, \eqref{eq:2nd_support_of_cor:outlier_micro_c} and \eqref{eq:1st_support_of_cor:outlier_micro_derived}.

In case \ref{enu:cor:outlier_micro:b}, we assume
\begin{equation} \label{eq:1st_support_of_lemma:outlier_macro}
  \hat{Z}_{n-1,\beta}(u,w) = R_{\beta}(u,w) e^{\frac{\beta}{2}n \int p(x;u,w)d\mu(x)} (1+o(1)).
\end{equation}
Then we immediately obtain \eqref{eq:enu:cor:outlier_micro:b}.

Below we prove the asymptotic formulas \eqref{eq:1st_support_of_cor:outlier_micro}, \eqref{eq:1st_support_of_cor:outlier_micro_2}, \eqref{eq:2nd_support_of_cor:outlier_micro}, \eqref{eq:1st_support_of_cor:outlier_micro_c}, \eqref{eq:1st_support_of_cor:outlier_micro_2_c} and \eqref{eq:1st_support_of_lemma:outlier_macro}.

\begin{proof}[Proof of \eqref{eq:1st_support_of_cor:outlier_micro}, \eqref{eq:1st_support_of_cor:outlier_micro_2}, \eqref{eq:1st_support_of_cor:outlier_micro_c}, \eqref{eq:1st_support_of_cor:outlier_micro_2_c} and \eqref{eq:1st_support_of_lemma:outlier_macro}]
In the proof, $c$ stands for $\bar{c} = (\redge+u)/2$ or $u$.

For fixed $u$, $w$ and $c$, let $f_1(x)$ and $f_2(x)$ be two functions on $\realR$, such that
\begin{enumerate}[label=(\arabic*)]
\item $f_1(x)$ and $f_2(x)$ satisfy the conditions \ref{enu:Johansson_condition_1}--\ref{enu:Johansson_condition_3} mentioned in Proposition \ref{prop:Johansson}.
\item $f_1(x) = f_2(x) = p(x;u,w)$ for $x \leq \redge$.
\item $f_1(x) \geq p(x;u,w)$ for $x \in (\redge,c)$.
\item There exists $x_0 \in (\redge,c)$ such that $f_2(x) \leq p(x;u,w)$ for $x \in (\redge,x_0)$ and $f_1(x) - f_2(x) \geq \log 2$ for $x \geq x_0$.
\end{enumerate}
As a consequence of the properties of $f_1(x)$ and $f_2(x)$, we have
\begin{equation} \label{eq:sandwich_inequality}
  \begin{split}
    & 2\E_{n-1,\beta} \left( e^{\frac{\beta}{2} \sum^{n-1}_{j=1}f_2(\lambda_1)} \right) - \E_{n-1,\beta}(e^{\frac{\beta}{2} \sum^{n-1}_{j=1}f_1(\lambda_1)}) \\
    < & \E_{n-1,\beta} \left( \lvert P_{n-1,\beta}(x_1, \dots, x_{n-1}; u,w) \rvert \prod^{n-1}_{j=1} \chi_{(-\infty,c)}(x_j) \right) \\
    < & \E_{n-1,\beta} \left( e^{\frac{\beta}{2} \sum^{n-1}_{j=1}f_1(\lambda_1)} \right).
  \end{split}
\end{equation}
By Proposition \ref{prop:Johansson}, we have for both $i=1,2$ that
\begin{equation}
\begin{split}
  \E_{n-1,\beta}(e^{\frac{\beta}{2} \sum^{n-1}_{j=1}f_i(\lambda_1)}) = & \exp \left[ (n-1) \int \frac{\beta}{2} p(x;u,w) d\mu(x) + (\frac{2}{\beta} - 1) \int \frac{\beta}{2} p(x;u,w) d\nu(x) \right. \\
  & \left. + \frac{2}{\beta} A \left( \frac{\beta}{2} p(x;u,w) \right) \right] (1+o(1))\\
  = & R_{\beta}(u,w) \exp \left[ \frac{\beta}{2}n \int p(x;u,w) d\mu(x) \right] (1+o(1)).
\end{split}
\end{equation}
Thus by the sandwich inequality \eqref{eq:sandwich_inequality} and \eqref{eq:defn_of_Z_hat_(uwc)_with_c} we obtain
\begin{equation} \label{eq:average_pf_F_beta_by sandwich}
  \begin{split}
      \hat{Z}_{n-1,\beta}(u,w;c) = & \E_{n-1,\beta} \left( \lvert P_{n-1,\beta}(x_1, \dots, x_{n-1}; u,w) \rvert \prod^{n-1}_{j=1} \chi_{(-\infty,c)}(x_j) \right) \\
      = & R_{\beta}(u,w) \exp \left[ \frac{\beta}{2}n \int p(x;u,w) d\mu(x) \right] (1+o(1)).
  \end{split}
\end{equation}
By \eqref{eq:average_pf_F_beta_by sandwich}, we complete the proof of \eqref{eq:1st_support_of_lemma:outlier_macro} with $c=u$. Let $w$ be given in \eqref{eq:condition_enu:cor:outlier_micro:a}, we have uniformly for all $x \leq \redge$ that
\begin{equation} \label{eq:relation_between_p_and_tilde_p}
p(x;u,w) = p(x;u) - \frac{s}{u-x}n^{-1} + O(n^{-2}).
\end{equation}
Thus
\begin{equation}
R_{\beta}(u,w) = R_{\beta}(u) + O(n^{-1})
\end{equation}
and
\begin{equation}
\exp \left[ \frac{\beta}{2}n \int p(x;u,w) d\mu(x) \right] = e^{-\frac{\beta s}{2} \int \frac{d\mu(x)}{u-x}} \exp \left[ \frac{\beta}{2}n \int p(x;u) d\mu(x) \right] (1+O(n^{-1})),
\end{equation}
and we obtain the proof of \eqref{eq:1st_support_of_cor:outlier_micro_2} with $c = \bar{c}$ and the proof of \eqref{eq:1st_support_of_cor:outlier_micro} with $c=u$. Let $w$ be given in \eqref{eq:condition_enu:cor:outlier_micro:c}, we have uniformly for all $x \leq \redge$ that
\begin{equation} \label{eq:relation_between_tilde_p_and_tilde_p_w_0}
p(x;u,w) = p(x;u;w_0) - \frac{t^2}{2(w_0-x)^2}n^{-1} + O(n^{-2}).
\end{equation}
Thus
\begin{equation}
R_{\beta}(u,w) = R_{\beta}(u,w_0) + O(n^{-1})
\end{equation}
and
\begin{equation}
\exp \left[ \frac{\beta}{2}n \int p(x;u,w) d\mu(x) \right] = e^{-\frac{\beta t^2}{4} \int \frac{d\mu(x)}{(w_0-x)^2}} \exp \left[ \frac{\beta}{2}n \int p(x;u,w_0) d\mu(x) \right] (1+O(n^{-1})),
\end{equation}
and we obtain the proof of \eqref{eq:1st_support_of_cor:outlier_micro_2_c} with $c = \bar{c}$ and the proof of \eqref{eq:1st_support_of_cor:outlier_micro_c} with $c=u$.
\end{proof}

\begin{rmk} \label{rmk:evaluation_of_E(hat_F)}
  By the same method, we can evaluate $\E_{n-1,\beta}(\hat{F}_{n-1,\beta}(x_1, \dots, x_{n-1};u))$ where $\hat{F}_{n-1,\beta}$ is defined in \eqref{eq:defn_of_hat_F}.
\end{rmk}

\begin{proof}[Proof of \eqref{eq:2nd_support_of_cor:outlier_micro} and \eqref{eq:2nd_support_of_cor:outlier_micro_c}]
We consider $(-\infty, \bar{c})^{n-1}$ as a probability space with the probability measure
\begin{equation}
\frac{1}{\hat{Z}_{n-1,\beta}(u,w;\bar{c})} \lvert P_{n-1,\beta}(x_1, \dots, x_{n-1};u,w) \rvert d\mu_{n-1,\beta}(x_1, \dots, x_{n-1}),
\end{equation}
where $d\mu_{n-1,\beta}(x_1, \dots, x_{n-1})$ is defined in \eqref{eq:defn_of_measure_mu_{n-1,beta}}. Let $S^w_{n-1,\beta}$ be a random variable on $(-\infty, \bar{c})^{n-1}$ such that
\begin{equation}
S^w_{n-1,\beta}(x_1, \dots, x_{n-1}) = \left( \sum^{n-1}_{j=1} \arg \frac{1}{(w-\lambda_j)^{\beta/2}} \right),
\end{equation}
where the range of the argument is taken to be $(-\pi, \pi]$. 

We define 
\begin{equation}
\sigma_v := A \left( \frac{1}{v-x} \right)^{-1/2}
\end{equation}
for any $v > \redge$. For $w$ given in \eqref{eq:condition_enu:cor:outlier_micro:a}, we will show 
\begin{align}
  \E(S^w_{n-1,\beta}) = & -\frac{\beta t}{2} \int \frac{d\mu(x)}{u-x} + o(1), \label{eq:formula_of_mean_argument} \\
  \Prob(n(S^w_{n-1,\beta} - \E(S^w_{n-1,\beta})) < T) = & \frac{1}{\sqrt{2\pi/\beta}t\sigma_u} \int^T_{-\infty} e^{-\frac{x^2}{2t^2\sigma^2_u/\beta}} dx +o(1), \label{eq:formula_of_mean_fluctuation_argument}
\end{align}
and for $w$ given in \eqref{eq:condition_enu:cor:outlier_micro:c}, we will show
\begin{align}
\E(S^w_{n-1,\beta}) = & -\frac{\beta t}{2} \sqrt{n} \int \frac{d\mu(x)}{w_0-x} + o(1), \label{eq:formula_of_mean_argument_w_0} \\
\Prob(\sqrt{n}(S^w_{n-1,\beta} - \E(S^w_{n-1,\beta})) < T) = & \frac{1}{\sqrt{2\pi/\beta}t\sigma_{w_0}} \int^T_{-\infty} e^{-\frac{x^2}{2t^2\sigma^2_{w_0}/\beta}} dx + o(1). \label{eq:formula_of_mean_fluctuation_argument_w_0}
\end{align}
Assuming \eqref{eq:formula_of_mean_argument} and \eqref{eq:formula_of_mean_fluctuation_argument}, we find that $S^w_{n-1,\beta}$ converges in probability to $-\frac{\beta t}{2} \int \frac{d\mu(x)}{u-x}$, and \eqref{eq:2nd_support_of_cor:outlier_micro} is proved. Assuming \eqref{eq:formula_of_mean_argument_w_0} and \eqref{eq:formula_of_mean_fluctuation_argument_w_0}, we find that $S^w_{n-1,\beta} + \frac{\beta t}{2} \sqrt{n} \int \frac{d\mu(x)}{w_0-x}$ converges in probability to $0$, and \eqref{eq:2nd_support_of_cor:outlier_micro_c} is proved.

To prove \eqref{eq:formula_of_mean_argument}, we denote for $x < \bar{c}$ the function
\begin{equation} \label{eq:defn_of_g_beta}
g_{\beta}(x;w) = n \arg \frac{1}{(w-x)^{\beta/2}}.
\end{equation}
$g_{\beta}(x;w)$ depends on $n$, but we suppress that dependence to economize on notation. Let $w$ be given by \eqref{eq:condition_enu:cor:outlier_micro:a}, uniformly for all $x < \bar{c}$
\begin{equation} \label{eq:asy_of_g_beta}
g_{\beta}(x;w) = -\frac{\beta t}{2} \frac{1}{u-x} + O(n^{-1}).
\end{equation}
Define the $(n-1)$-variable function
\begin{equation}
  G_{r,\beta}(x_1, \dots, x_{n-1};u,w;\bar{c}) :=
  \begin{cases}
    P_{n-1,\beta}(x_1, \dots, x_{n-1};u,w) e^{r \sum^{n-1}_{j=1} g_{\beta}(x_j;w)} & \textnormal{if $\max_{1 \leq j \leq n-1} x_j < \bar{c}$,} \\
    0 & \textnormal{otherwise.}
  \end{cases}
\end{equation}
We have (comparing with \eqref{eq:defn_of_Z_hat_(uwc)_with_c})
\begin{align}
  \E(e^{rnS^w_{n-1,\beta}}) = & \frac{\E_{n-1,\beta}\left( G_{r,\beta}(x_1, \dots, x_{n-1};u,w;\bar{c}) \right)}{\hat{Z}_{n-1,\beta}(u,w; \bar{c})}, \label{eq:expression_of_exponential_mean} \\
  n\E(S^w_{n-1,\beta}) = & \left. \frac{d}{dr} \E(e^{rnS^w_{n-1,\beta}}) \right\rvert_{r=0}. \label{eq:expression_of_mean_by_derivative_of_ln}
\end{align}
For any $r \in \realR$, analogous to \eqref{eq:average_pf_F_beta_by sandwich} we have
\begin{equation} \label{eq:evaluation_of_ln_of_ave_for_r<>0}
\begin{split}
  & \log \E_{n-1,\beta}\left( G_{r,\beta}(x_1, \dots, x_{n-1};u,w;\bar{c}) \right) \\
  = & (n-1) \int r g_{\beta}(x;w) + p(x;u,w) d\mu(x) 
   + \left( \frac{2}{\beta} - 1 \right) \int r g_{\beta}(x;w) + p(x;u,w) d\nu(x) \\
  & + \frac{2}{\beta} A(r g_{\beta}(x;w) + p(x;u,w)) + o(1) \\
  = & \log \hat{Z}_{n-1,\beta}(u,w;\bar{c}) + r \left[ (n-1) \int g_{\beta}(x;w) d\mu(x) \right. \\
  & + \left. \left( \frac{2}{\beta} - 1 \right) \int g_{\beta}(x;w) d\nu(x) + \frac{4}{\beta} \langle g_{\beta}(x;w), p(x;u,w) \rangle_A\right] 
   + \frac{2r^2}{\beta} A(g_{\beta}(x;w)) + o(1),
\end{split}
\end{equation}
where we use \eqref{eq:1st_support_of_cor:outlier_micro_2_c} in the last line. Using the Cauchy-Schwartz inequality, we find (for notational simplicity, we write $G_{r,\beta}(x_1, \dots, x_{n-1};u,w; \bar{c})$ as $G_{r,\beta}$ if there is no confusion)
\begin{equation}
\begin{split}
& \frac{d^2}{dr^2} \log \E_{n-1,\beta}(G_{r,\beta}(x_1, \dots, x_{n-1};u,w; \bar{c})) \\
= & \E_{n-1,\beta}(G_{r,\beta})^{-2}  \left[ \E_{n-1,\beta} \left( \left( \sum^{n-1}_{k=1} g_{\beta}(x;w) \right)^2 G_{r,\beta} \right) \E_{n-1,\beta}(G_{r,\beta}) \right. \\
& - \left. \E_{n-1,\beta} \left( \left( \sum^{n-1}_{k=1} g_{\beta}(x;w) \right) G_{r,\beta} \right) \E_{n-1,\beta} \left( \left( \sum^{n-1}_{k=1} g_{\beta}(x;w) \right) G_{r,\beta} \right) \vphantom{\left( \left( \sum^{n-1}_{k=1} g_{\beta}(x;w) \right)^2 G_{r,\beta} \right) \E_{n-1,\beta}(G_{r,\beta}) } \right] \\
> & 0.
\end{split}
\end{equation}
Hence $\log \E_{n-1,\beta}(G_{r,\beta}(x_1, \dots, x_{n-1};u,w; \bar{c}))$ is a convex function in $r$. For any $\epsilon > 0$, by \eqref{eq:expression_of_mean_by_derivative_of_ln}
\begin{equation} \label{eq:sandwich_of_derivative_of_ln}
\begin{split}
  \frac{\log E_{n-1,\beta}(G_{-\epsilon,\beta}) - \hat{Z}_{n-1,\beta}(u,w; \bar{c})}{-\epsilon} < & \left. \frac{d \log E_{n-1,\beta}(G_{r,\beta})}{dr} \right\rvert_{r=0} = n\E(S^w_{n-1,\beta}) \\
  < & \frac{\log E_{n-1,\beta}(G_{\epsilon,\beta}) - \hat{Z}_{n-1,\beta}(u,w; \bar{c})}{\epsilon}.
\end{split}
\end{equation}
Taking $\epsilon \to 0$, by \eqref{eq:evaluation_of_ln_of_ave_for_r<>0}, \eqref{eq:sandwich_of_derivative_of_ln} and \eqref{eq:asy_of_g_beta} we have 
\begin{equation} \label{eq:asymptotics_of_n_E(Sigma)}
n \E(S^w_{n-1,\beta}) = (n-1) \int g_{\beta}(x;w) d\mu(x) + \left( \frac{2}{\beta} - 1 \right) \int g_{\beta}(x;w) d\nu(x) + \frac{4}{\beta} \langle g_{\beta}(x;w), p(x;u,w) \rangle_A + o(1), 
\end{equation}
and by \eqref{eq:defn_of_g_beta}
\begin{equation}
\E(S^w_{n-1,\beta}) = \int g_{\beta}(x;w) d\mu(x) + o(1) = -\frac{\beta t}{2} \int \frac{d\mu(x)}{u-x} + o(1).
\end{equation}

To prove \eqref{eq:formula_of_mean_fluctuation_argument}, we consider the moment-generating function of $n(S^w_{n-1,\beta} - \E(S^w_{n-1,\beta}))$. By \eqref{eq:expression_of_exponential_mean}  we have
\begin{equation} \label{eq:representation_of_moment_gen_u}
\begin{split}
M_{n(S^w_{n-1,\beta} - \E(S^w_{n-1,\beta}))}(\xi) = & \E(\exp[\xi n(S^w_{n-1,\beta} - \E(S^w_{n-1,\beta}))]) \\
= & \frac{\E_{n-1,\beta}(G_{\xi,\beta}(\lambda_1, \dots, \lambda_{n-1};u,w; \bar{c}))}{\exp[\xi n \E(S^w_{n-1,\beta})] \hat{Z}_{n-1,\beta}(u,w; \bar{c})}.
\end{split}
\end{equation}
Then by \eqref{eq:evaluation_of_ln_of_ave_for_r<>0} and \eqref{eq:asymptotics_of_n_E(Sigma)} we have 
\begin{equation} \label{eq_convergence_moment_generating_u}
\begin{split}
M_{n(S^w_{n-1,\beta} - \E(S^w_{n-1,\beta}))}(\xi) = & \exp \left[ \frac{2\xi^2}{\beta} A(g_{\beta}(x;w)) + o(1) \right] \\
= & \exp \left[ \frac{\beta t^2\xi^2}{2} A \left( \frac{1}{u-x} \right) + o(1) \right],
\end{split}
\end{equation}
where in the last step we use \eqref{eq:asy_of_g_beta}. The convergence of moment-generating function \eqref{eq_convergence_moment_generating_u} implies \eqref{eq:formula_of_mean_fluctuation_argument}.

To prove \eqref{eq:formula_of_mean_argument_w_0}, we denote for $x < \bar{c}$ the function
\begin{equation} \label{eq:defn_of_tilde_g_beta}
\tilde{g}_{\beta}(x;w) := \frac{1}{\sqrt{n}} g_{\beta}(x;w) = \sqrt{n} \arg \frac{1}{(w-x)^{\beta/2}}.
\end{equation}
$\tilde{g}_{\beta}(x;w)$ depends on $n$, and we suppress the dependence to economize on notation. Let $w$ be given by \eqref{eq:condition_enu:cor:outlier_micro:c}, uniformly for all $x < \bar{c}$
\begin{equation} \label{eq:asy_of_g_beta_away}
\tilde{g}_{\beta}(x;w) = -\frac{\beta t}{2} \frac{1}{w_0-x} + O(n^{-1}).
\end{equation}
Like \eqref{eq:asymptotics_of_n_E(Sigma)}, we have
\begin{multline} \label{eq:asymptotics_of_n_E(Sigma)_w_0}
\sqrt{n} \E(S^w_{n-1,\beta}) = (n-1) \int \tilde{g}_{\beta}(x;w) d\mu(x) + \left( \frac{2}{\beta} - 1 \right) \int \tilde{g}_{\beta}(x;w) d\nu(x) \\
+ \frac{4}{\beta} \langle \tilde{g}_{\beta}(x;w), p(x;u,w) \rangle_A + o(1), 
\end{multline}
and by \eqref{eq:asymptotics_of_n_E(Sigma)_w_0} we obtain
\begin{equation}
\E(S^w_{n-1,\beta}) = \sqrt{n} \int \tilde{g}_{\beta}(x;w) d\mu(x) + o(1) = -\frac{\beta t}{2}\sqrt{n} \int \frac{d\mu(x)}{w_0-x} + o(1)
\end{equation}
and complete the proof of \eqref{eq:formula_of_mean_argument_w_0}.

To prove \eqref{eq:formula_of_mean_fluctuation_argument_w_0}, we consider the moment-generating function of $\sqrt{n}(S^w_{n-1,\beta} - \E(S^w_{n-1,\beta}))$. Like \eqref{eq:representation_of_moment_gen_u} and \eqref{eq_convergence_moment_generating_u}, we have the convergence of moment-generating function
\begin{equation}
\begin{split}
M_{\sqrt{n}(S^w_{n-1,\beta} - \E(S^w_{n-1,\beta}))}(\xi) = & \E(\exp[\xi \sqrt{n}(S^w_{n-1,\beta} - \E(S^w_{n-1,\beta}))]) \\
= &  \exp \left[ \frac{2\xi^2}{\beta} A(\tilde{g}_{\beta}(x;w)) + o(1) \right] \\
= & \exp \left[ \frac{\beta t^2\xi^2}{2} A \left( \frac{1}{w_0-x} \right) + o(1) \right],
\end{split}
\end{equation}
which implies \eqref{eq:formula_of_mean_fluctuation_argument_w_0}.
\end{proof}

\paragraph{Acknowledgments}

The author thanks Mark Adler, Jinho Baik, Kenneth D.~T-R McLaughlin and Peter J.~Forrester for helpful comments, and anonymous referees for careful reading and valuable suggestions on presentation.

\appendix

\section{Contour integral formula of the joint \pdf{s} of the eigenvalues in $\beta$-external source models with $\beta = 1,2,4$, and $\beta$-external source model for all $\beta > 0$}  \label{sec:external_source_model_with_odd_n_and_general_beta_external_source_model}

The goal of this appendix is two-fold. We prove Proposition \ref{prop:contour_rep_of_pdf_of_largest_eigenvalue} and also propose the definition of the $\beta$-external source model. 

\begin{rmk} \label{rmk:relation_to_Mo}
  The strategy in this appendix has appeared in \cite[Appendix]{Mo11} independently for the purpose of proof of \cite[Theorem 1]{Mo11}. Since we are concerned with $\beta = 1, 2, 4$ cases and furthermore all $\beta > 0$, we give full detail in this appendix.
\end{rmk}

By change of variables and calculation of Jacobian (\cf\ \cite[Chapter 3]{Mehta04}), it follows from \eqref{eq:pdf_of_matrix_beta=1}, \eqref{eq:pdf_of_matrix_beta=2} and \eqref{eq:quaternion_form_of_pdf_of_beta=4} that the joint \pdf{s} of the eigenvalues $\lambda_1, \dots, \lambda_n$ of $M$ in the three $\beta$-external source models ($\beta = 1,2,4$) are given by
\begin{equation} \label{eq:pdf_of_eigenvalues}
  p_{n,\beta}(\lambda_1,\dots,\lambda_n) = \frac{1}{C_{n,\beta}} \lvert \Delta(\lambda_1, \dots, \lambda_n) \rvert^{\beta} \prod^n_{j=1} e^{-\frac{\beta}{2}nV(\lambda_j)} \int_{Q \in G_{\beta}(n)} e^{n\Re\Tr(\A_{n,\beta}Q\Lambda_nQ^{-1})} dQ,
\end{equation}
where $C_{n,\beta}$ is a normalization constant and $C_{n,\beta}/\tilde{C}_{n,\beta}$ is a constant depending only on $n$. The integral in \eqref{eq:pdf_of_eigenvalues} is with respect to the Haar measure of the compact group $G_{\beta}(n)$, which is the orthogonal group $O(n)$, the unitary group $U(n)$ and the compact symplectic group $Sp(n)$ for $\beta = 1,2,4$ respectively. The matrix $\A_{n,\beta}$ is defined in \eqref{eq:defination_of_matrix_A}, \eqref{eq:relation_of_A_beta_to_A} and \eqref{eq:diag_formula_of_A_n4}, and
\begin{equation} \label{eq:definition_of_Lambda_beta}
\Lambda_n = \diag(\lambda_1, \lambda_2, \dots, \lambda_n).
\end{equation} 

Recall that in combinatorics, a partition $\kappa = (\kappa_1, \kappa_2, \dots)$ is a sequence of non-negative integers in decreasing order, and containing only finitely many non-zero terms. We denote $l(\kappa)$ as the number of non-zero terms of $\kappa$, and write $\kappa \vdash k$ if $\sum^{l(\kappa)}_{i=1} \kappa_i = k$. 

Jack polynomials $C^{(\alpha)}_{\kappa}(x_1, \dots, x_n)$ are $n$-variable symmetric polynomials indexed by partition $\kappa$ and the parameter $\alpha$. For general references of Jack polynomials, see \cite{Macdonald95} and \cite{Stanley89}. In this paper, we take the ``C''-normalization of Jack polynomials \cite{Dumitriu-Edelman-Shuman07}, such that
\begin{equation} \label{eq:normalization_of_Jack_polynomial}
  \sum_{\kappa \vdash k, l(\kappa) \leq n} C^{(\alpha)}_{\kappa}(x_1, \dots, x_n) = (x_1 + \dots + x_n)^k.
\end{equation}
The Jack polynomials with parameters $2, 1, \frac{1}{2}$ are Zonal spherical functions. See \cite[Chapter VII]{Macdonald95}. $C^{(2)}_{\kappa}$ are the well known Zonal polynomial in statistics \cite{Muirhead82}, $C^{(1)}_{\kappa}$ are the complex Zonal polynomials, and are better known as Schur polynomials, and $C^{(1/2)}_{\kappa}$ are the quaternionic Zonal polynomials.

The integral in \eqref{eq:pdf_of_eigenvalues} can be expanded in Jack polynomials:
\begin{prop}
Let $\beta = 1,2,4$ and $G_{\beta}(n)$ be $O(n)$, $U(n)$ and $Sp(n)$ respectively. If $\A_{n,\beta}$ is defined by \eqref{eq:defination_of_matrix_A}, \eqref{eq:relation_of_A_beta_to_A} and \eqref{eq:diag_formula_of_A_n4} and $\Lambda_n$ is defined by \eqref{eq:definition_of_Lambda_beta}, then
\begin{equation} \label{eq:general_integral_formula_of_Jack}
\int_{Q \in G_{\beta}(n)} e^{n\Re\Tr(\A_{n,\beta}Q\Lambda_nQ^{-1})} dQ = \sum^{\infty}_{k=0} \frac{(\frac{\beta}{2}n)^k}{k!} \sum_{\kappa \vdash k, l(\kappa) \leq n} \frac{C^{(2/\beta)}_{\kappa}(a_1, \dots, a_n) C^{(2/\beta)}_{\kappa}(\lambda_1, \dots, \lambda_n)}{C^{(2/\beta)}_{\kappa}(1, \dots, 1)}.
\end{equation}
\end{prop}
\begin{proof}
  Any $n$-variable symmetric polynomial $f$ can be regarded as a polynomial function from the spaces of $n \times n$ matrices $M_n(F)$ to $F$, where $F$ stands for the division algebras $\realR$, $\compC$ and $\quatH$. For $F = \realR$ and $F = \compC$, the definition is simple: If $M \in M_n(\realR)$ or $M \in M_n(\compC)$ and the eigenvalues of $M$ are $\xi_1, \dots, \xi_n$, then \cite[Pages 420 and 443]{Macdonald95} 
  \begin{equation}
    f(M) = f(\xi_1, \dots, \xi_n).
  \end{equation}
  For $F = \quatH$, the definition is more complicated and the reader is referred to \cite[Page 452]{Macdonald95}. In all the three cases, identity \eqref{eq:normalization_of_Jack_polynomial} implies that
  \begin{equation} \label{eq:matrix_representation_of_Jack_normalization}
    \Re \Tr(\A_{n,\beta}Q\Lambda_nQ^{-1})^k = \left( \frac{\beta}{2} \right)^k \sum_{\kappa \vdash k, l(\kappa) \leq n} C^{(2/\beta)}_{\kappa}(\A_nQ\Lambda_nQ^{-1}),
  \end{equation}
  where $\A_n = \frac{2}{\beta} \A_{n,\beta} = \diag(a_1, \dots, a_n)$ as defined in \eqref{eq:defination_of_matrix_A}. Furthermore, by general theory of Zonal spherical functions (\eg\ \cite[Proposition 5.5]{Gross-Richards87})
  \begin{equation} \label{eq:standard_result_of_Zonal_spherical_poly}
    \int_{Q \in G_{\beta}(n)} C^{(2/\beta)}_{\kappa}(\A_nQ\Lambda_nQ^{-1}) = \frac{C^{(2/\beta)}_{\kappa}(a_1, \dots, a_n) C^{(2/\beta)}_{\kappa}(\lambda_1, \dots, \lambda_n)}{C^{(2/\beta)}_{\kappa}(1, \dots, 1)}.
  \end{equation}
  After expanding $e^{n\Re\Tr(\A_{n,\beta}Q\Lambda_nQ^{-1})}$ into power series of $\Re\Tr(\A_{n,\beta}Q\Lambda_{n,\beta}Q^{-1})$, we prove \eqref{eq:general_integral_formula_of_Jack} by \eqref{eq:matrix_representation_of_Jack_normalization} and \eqref{eq:standard_result_of_Zonal_spherical_poly}.
\end{proof}

In case that $\A_n = \diag(a, 0, \dots, 0)$, \eqref{eq:general_integral_formula_of_Jack} is much simplified by the property of Jack polynomials:
\begin{prop}{\cite[Proposition 2.5]{Stanley89}} \label{prop:Stanley_89}
If the number of nonzero variables among $a_1, \dots, a_n$ is less than $l(\kappa)$, then $C^{(\alpha)}_{\kappa}(a_1, \dots, a_n) = 0$ for any $\alpha > 0$.
\end{prop}
Therefore, in the case $\A_n = \diag(a, 0, \dots, 0)$,
\begin{equation} \label{eq:Jack_poly_representation_of_group_integral}
  \int_{Q \in G_{\beta}(n)} e^{n\Re\Tr(\A_{n,\beta}Q\Lambda_nQ^{-1})} dQ = \sum^{\infty}_{k=0} \frac{(\frac{\beta}{2}n)^k}{k!}\frac{C^{(2/\beta)}_{(k)}(a, 0, \dots, 0)C^{(2/\beta)}_{(k)}(\lambda_1, \dots, \lambda_n)}{C^{(2/\beta)}_{(k)}(1, \dots, 1)}.
\end{equation}
$C^{(2/\beta)}_{(k)}(a, 0, \dots, 0)$ and $C^{(2/\beta)}_{(k)}(1, \dots, 1)$ can be calculated explicitly \cite[Table 5]{Dumitriu-Edelman-Shuman07}
\begin{equation} \label{eq:special_values_of_Jack_polynomial}
  C^{(2/\beta)}_{(k)}(a, 0, \dots, 0) = a^k, \quad C^{(2/\beta)}_{(k)}(1, \dots, 1) = \prod^{k-1}_{j=0} \frac{n+\frac{2}{\beta}j}{1+\frac{2}{\beta}j}.
\end{equation}
Thus we have
\begin{equation}
\int_{Q \in G_{\beta}(n)} e^{n\Re\Tr(\A_{n,\beta}Q\Lambda_nQ^{-1})} dQ = \sum^{\infty}_{k=0} \prod^{k-1}_{j=0} \frac{1+\frac{2}{\beta}j}{n+\frac{2}{\beta}j} \frac{(\frac{\beta}{2}an)^k}{k!} C^{(2/\beta)}_{(k)}(\lambda_1, \dots, \lambda_n).
\end{equation}

By \cite[Proposition 2.1]{Stanley89} and the conversion between the ``J''-normalization and ``C''-normalization \cite[Table 6]{Dumitriu-Edelman-Shuman07}, we have the identity of formal series in $a$
\begin{equation} \label{eq:combinatorial_from_Stanley}
  \sum^{\infty}_{k=0} \frac{\prod^{k-1}_{j=0}(1+\frac{2}{\beta}j)}{(\frac{2}{\beta})^k k!} C^{(2/\beta)}_{(k)}(a, 0, \dots, 0) C^{(2/\beta)}_{(k)}(\lambda_1, \dots, \lambda_n) = \prod^n_{j=1} \frac{1}{(1-a\lambda_j)^{\beta/2}}.
\end{equation}
Hence we obtain by Cauchy's integral formula and \eqref{eq:special_values_of_Jack_polynomial}
\begin{equation} \label{eq:Jack_poly_evaluated_by_residue}
  \frac{\prod^{k-1}_{j=0}(1+\frac{2}{\beta}j)}{(\frac{2}{\beta})^k k!} C^{(2/\beta)}_{(k)}(\lambda_1, \dots, \lambda_n) = \frac{1}{2\pi i} \oint_0 \prod^n_{j=1} \frac{1}{(1-z\lambda_j)^{\beta/2}} \frac{dz}{z^{k+1}},
\end{equation}
where the contour is taken to be a small circle around $0$ such that all $\lambda^{-1}_j$ ($j=1, \dots, n$) are in the exterior of the contour. By \eqref{eq:special_values_of_Jack_polynomial} and \eqref{eq:Jack_poly_evaluated_by_residue}, we obtain
\begin{multline} \label{eq:integral_over_G_beta_into_contour_integral}
  \sum^{\infty}_{k=0} \frac{(\frac{\beta}{2}n)^k}{k!}\frac{C^{(2/\beta)}_{(k)}(a, 0, \dots, 0)C^{(2/\beta)}_{(k)}(\lambda_1, \dots, \lambda_n)}{C^{(2/\beta)}_{(k)}(1, \dots, 1)} = \\
  \frac{1}{2\pi i} \oint_0 \prod^n_{j=1} \frac{1}{(1-z\lambda_j)^{\beta/2}} \left( \sum^{\infty}_{k=0} (\frac{\beta}{2}an)^k \prod^{k-1}_{j=0} \left( \frac{1}{\frac{\beta}{2}n+j} \right) \frac{1}{z^{k+1}} \right) dz.
\end{multline}
Note that \eqref{eq:integral_over_G_beta_into_contour_integral} is valid for all $\beta > 0$.

Suppose $\beta$ is a positive number. Let $m$ be an integer and $\xi \in (0,1]$, such that
\begin{equation} \label{eq:defn_of_m,xi,alpha}
  \frac{\beta}{2} n = m+\xi,
\end{equation}
we have
\begin{equation} \label{eq:contour_integral_identity_of_general_beta}
  \begin{split}
    & \frac{1}{2\pi i} \oint_0 \prod^n_{j=1} \frac{1}{(1-z\lambda_j)^{\beta/2}} \left( \sum^{\infty}_{k=0} (\frac{\beta}{2} an)^k \prod^{k-1}_{j=0} \left( \frac{1}{\frac{\beta}{2} n+j} \right) \frac{1}{z^{k+1}} \right) dz \\
    = & \frac{1}{2\pi i} \oint_0 \prod^n_{j=1} \frac{1}{(1-z\lambda_j)^{\beta/2}} \left( \frac{(\xi)_m}{(\frac{\beta}{2} an)^m} z^{m-1} \sum^{\infty}_{k=m} \frac{1}{(\xi)_k} \left( \frac{\frac{\beta}{2} an}{z} \right)^k \right) dz \\
    = & \frac{1}{2\pi i} \oint_0 \prod^n_{j=1} \frac{1}{(1-z\lambda_j)^{\beta/2}} \left( \frac{(\xi)_m}{(\frac{\beta}{2} an)^m} z^{m-1} \left( M(1, \xi, \frac{\frac{\beta}{2} an}{z}) - \sum^{m-1}_{k=0} \frac{1}{(\xi)_k} \left( \frac{\frac{\beta}{2} an}{z} \right)^k \right) \right) dz \\
    = & \frac{(\xi)_m}{(\frac{\beta}{2} an)^m} \frac{1}{2\pi i} \oint_0 \prod^n_{j=1} \frac{1}{(1-z\lambda_j)^{\beta/2}} z^{m-1} M(1, \xi, \frac{\frac{\beta}{2} an}{z}) dz \\
  = & \frac{(\xi)_m}{(\frac{\beta}{2} an)^m} \frac{1}{2\pi i} \oint_{\infty} \left( \prod^n_{j=1} \frac{1}{(w-\lambda_j)^{\beta/2}} \right) w^{\xi-1} M(1, \xi, \frac{\beta}{2} anw) dw,
  \end{split}
\end{equation}
where the contour is large enough so that all $\lambda_j$ are in its interior, and $0$ is in its interior if $\xi \neq 1$. Here
\begin{equation}
  (c)_i := c(c+1)(c+2) \dots (c+n-1)
\end{equation}
is the Pochhammer symbol (``rising factorial''), and
\begin{equation}
  M(c_1,c_2,z) := \sum^{\infty}_{i=0} \frac{(c_1)_i}{(c_2)_i i!} z^i
\end{equation}
is the Kummer's (confluent hypergeometric) function. See \cite[13.1.2]{Abramowitz-Stegun64}. Alternatively, for $\xi \neq 1$
\begin{equation}
  M(1, \xi, z) = (\xi-1)z^{1-\xi} e^z \gamma(\xi-1, z)
\end{equation}
where $\gamma(s,z)$ is the incomplete gamma function (\cf\ \cite[6.5.12]{Abramowitz-Stegun64}), and for $\xi = 1$
\begin{equation}
  M(1,1,z) = e^z.
\end{equation}

We note that in the cases $\beta = 2,4$, or in the case that $\beta = 1$ and $n$ is even, $\xi = 1$ and $m = \frac{\beta}{2}n-1$. Thus by \eqref{eq:pdf_of_eigenvalues}, \eqref{eq:Jack_poly_representation_of_group_integral}, \eqref{eq:integral_over_G_beta_into_contour_integral} and \eqref{eq:contour_integral_identity_of_general_beta}, we have
\begin{equation} \label{eq:pdf_in_contour_integral_general}
  p_{n,\beta}(\lambda_1,\dots,\lambda_n) = \bar{C}_{n,\beta} \lvert \Delta(\lambda_1, \dots, \lambda_n) \rvert^{\beta} \prod^n_{j=1} e^{-\frac{\beta}{2}nV(\lambda_j)} \frac{1}{2\pi i} \oint_{\infty} \left( \prod^n_{j=1} \frac{1}{(w-\lambda_j)^{\beta/2}} \right) e^{\frac{\beta}{2}anw} dw
\end{equation}
where $\bar{C}_{n,\beta}$ is a constant, and the contour encloses all $\lambda_j$ in its interior. 

If $a > 0$ and $\max_{1 \leq j \leq n} \lambda_j \leq u$, the contour in \eqref{eq:pdf_in_contour_integral_general} can be taken as $\Sigma^z_{s_1,s_2}$ defined in \eqref{eq:parametrization_of_Sigma} or $\Pi^z_s$ defined in \eqref{eq:parametrization_of_Pi}, where $z > u$. From the joint \pdf\ of $\lambda_j$, it is straightforward to find the \pdf\ formula \eqref{eq:pdf_of_largest_eigen_supercritical} of the largest eigenvalue $\ximax$. Thus Proposition  \ref{prop:contour_rep_of_pdf_of_largest_eigenvalue} is proved.

\bigskip

If $\beta = 1$ and $n$ is odd, similarly we obtain that the joint \pdf\ of The eigenvalues in $n$-dimensional $1$-external source model is
\begin{multline} \label{eq:pdf_in_contour_integral_general_beta=1_n_odd}
  p_{n,1}(\lambda_1,\dots,\lambda_n) = \\
  \bar{C}_{n,1} \lvert \Delta(\lambda_1, \dots, \lambda_n) \rvert \prod^n_{j=1} e^{-\frac{1}{2}nV(\lambda_j)} \frac{1}{2\pi i}\oint_{\infty} \left( \prod^n_{j=1} \frac{1}{(w-\lambda_j)^{1/2}} \right) w^{-1/2} M(1, \frac{1}{2}, \frac{1}{2}anw).
\end{multline}

In Section \ref{sec:the_pdf_of_the_largest_eigenvalue} we compute the limiting distribution of $\ximax$ based on \eqref{eq:pdf_in_contour_integral_general}. Since the asymptotic property of $M(1, \frac{1}{2},z)$ is similar to that of $e^z = M(1,1,z)$ for large $z$, we can compute the limiting distribution of $\ximax$ based on \eqref{eq:pdf_in_contour_integral_general_beta=1_n_odd} by the same method that we use in Section \ref{sec:the_pdf_of_the_largest_eigenvalue}. Hence we can prove that Theorems \ref{thm:limiting_position}, \ref{thm:limiting_distr_nondegenerate}, \ref{thm:two_normal_max_points} and \ref{thm:one_abnormal_max_points} hold when $\beta = 1$ and $n$ is odd.

\bigskip

Inspired by the Coulomb gas interpretation of the distribution of eigenvalues in random matrix models (see \cite{Forrester10}), we generalize the $\beta$-external source model to any $\beta > 0$ as the probability distribution of $n$ points on the real line, such that 
\begin{multline} \label{eq:general_beta_external_source_model}
  p_{n,\beta}(\lambda_1, \dots, \lambda_n) = \frac{1}{C_{n,\beta}} \lvert \Delta(\lambda_1, \dots, \lambda_n) \rvert^{\beta} \prod^n_{j=1} e^{-\frac{\beta}{2}nV(\lambda_j)} \\
  \times \sum^{\infty}_{k=0} \frac{(\frac{\beta}{2}n)^k}{k!} \sum_{\kappa \vdash k, l(\kappa) \leq n} \frac{C^{(2/\beta)}_{\kappa}(a_1, \dots, a_n) C^{(2/\beta)}_{\kappa}(\lambda_1, \dots, \lambda_n)}{C^{(2/\beta)}_{\kappa}(1, \dots, 1)},
\end{multline}
where $V$ is the potential and $a_1, \dots, a_n$ are external source parameters. By \eqref{eq:pdf_of_eigenvalues} and \eqref{eq:general_integral_formula_of_Jack}, \eqref{eq:general_beta_external_source_model} gives the distribution of eigenvalues of the random matrix models with external source with $\beta = 1,2,4$. But for other value of $\beta$, it has no matrix interpretation. By Proposition \ref{prop:Stanley_89}, \eqref{eq:special_values_of_Jack_polynomial} \eqref{eq:combinatorial_from_Stanley} and \eqref{eq:contour_integral_identity_of_general_beta}, we find that if one external source parameter is $a$ and all others are $0$, the distribution of the right-most point in the general $\beta$-external source model is
\begin{multline} \label{eq:pdf_in_contour_integral_general_beta=1_n_odd}
  p_{n,\beta}(\lambda_1,\dots,\lambda_n) = \bar{C}_{n,\beta} \lvert \Delta(\lambda_1, \dots, \lambda_n) \rvert^{\beta} \prod^n_{j=1} e^{-\frac{\beta}{2}nV(\lambda_j)} \\
  \times \frac{1}{2\pi i}\oint_{\infty} \left( \prod^n_{j=1} \frac{1}{(w-\lambda_j)^{\beta/2}} \right) w^{\xi-1} M(1, \xi, \frac{\beta}{2} anw),
\end{multline}
where $\xi$ is defined by \eqref{eq:defn_of_m,xi,alpha}. It is of interest to compare this formula with the rank $1$ spiked Gaussian and Laguerre $\beta$ ensembles studied in \cite{Bloemendal-Virag11}.  In the very recent preprint \cite{Forrester11}, Forrester obtained similar formulas for $\beta$-Wishart ensembles.

\section{Computation of $\M_2(u)$} \label{sec:simplification_of_M_2(x)}

In this Appendix, we show that when $\beta = 2$, the $p_{j,2}(\alpha)$ in Theorem \ref{thm:two_normal_max_points} are the same as the $p^{(j)}_{1,n}(\alpha)$ in \cite[Formula (52)]{Baik-Wang10a}, and the $\tilde{p}_{j,2}(\alpha)$ in Theorem \ref{thm:one_abnormal_max_points} are the same as the $p^{(j)}_{1,n}(\alpha)$ in \cite[Formula (63)]{Baik-Wang10a}. Hence we verify that the result obtained in this paper agrees with the result in \cite{Baik-Wang10a}. Since $p_{j,2}(\alpha)$ and $\tilde{p}_{j,2}(\alpha)$ in our paper are defined by $\Gfn(u;a)$ and $\M_2(u)$ in the same way that $p_{j,2}(\alpha)$ in \cite[Formulas (52) and (63)]{Baik-Wang10a} are defined by $\Gfn(u;a)$ and $\M_{1,n}(u)$, (see \eqref{eq:defn_of_B_i_two_normal_points}, \eqref{eq:defn_of_p_i_two_normal_points}, \eqref{eq:formula_of_D_1,beta}, \eqref{eq:formula_of_D_2,beta} and \eqref{eq:defn_of_tilde_p} in this paper, and \cite[Formulas (177), (178) and (182)]{Baik-Wang10a} and the comments in \cite{Baik-Wang10a} below \cite[Formula (52)]{Baik-Wang10a}), we need only to show that $\M_2(u)/\M_{1,n}(u)$ is a nonzero constant for $u > \redge$. ($\redge$ is the right edge of $J$, the support of the equilibrium measure $\mu$.) The explicit formula of $\M_{1,n}(u)$ is given in \cite[Formula (315)]{Baik-Wang10a} in the case that the support of equilibrium measure is one interval. In this appendix we obtain that for $u > \redge$
\begin{equation} \label{eq:final_formula_of_M_2}
  \M_2(u) = C(\gamma(u)-\gamma(u)^{-1}),
\end{equation}
where $C$ is a constant independent of $u$, and, if $J = [b_1, b_2]$,
\begin{equation}
  \gamma(u) = \left( \frac{u-b_1}{u-b_2} \right)^{1/4}.
\end{equation}
Thus we prove the statements above.

To make the notations simpler, we assume $J = [-1, 1]$ in the proof of \eqref{eq:final_formula_of_M_2}. The generalization to arbitrary $J$ is straightforward.

From formula \eqref{eq:defn_of_M_beta}, \eqref{eq:defn_of_R_beta(u)} and \eqref{eq:defn_of_quadratic_A}, we have that
\begin{equation} \label{eq:explicit_formula_of_M_2}
  \begin{split}
    \M_2(u) = & \exp \left[ A(-V(x) + \log(u-x)) - \int^1_{-1} -V(x) + \log(u-x) d\mu(x) \right] \\
    = & e^{\int^1_{-1} V(x) d\mu(x) - \frac{1}{4\pi^2} \int^1_{-1} \frac{V(x)}{\sqrt{1-x^2}} \pv \int^1_{-1} \frac{V'(s)\sqrt{1-s^2}}{s-x} dsdx} \\
  & \times e^{ \left\{ \begin{split} & \scriptstyle \frac{1}{4\pi^2} \int^1_{-1} \frac{\log(u-x)}{\sqrt{1-x^2}} \pv \int^1_{-1} \frac{V'(s)\sqrt{1-s^2}}{s-x} dsdx \\ & \scriptstyle - \frac{1}{4\pi^2} \int^1_{-1} \frac{V(x)}{\sqrt{1-x^2}} \pv \int^1_{-1} \frac{\sqrt{1-s^2}}{(u-s)(s-x)} dsdx  - \int^1_{-1} \log(u-x) d\mu(x) \end{split} \right\}} \\
  & \times e^{\frac{1}{4\pi^2} \int^1_{-1} \frac{\log(u-x)}{\sqrt{1-x^2}} \pv \int^1_{-1} \frac{\sqrt{1-s^2}}{(u-s)(s-x)} dsdx}.
  \end{split}
\end{equation}
The right-hand side of \eqref{eq:explicit_formula_of_M_2} is divided into the product of three terms. The first one is a constant, and we compute the other two terms below.

First we compute the third term in \eqref{eq:explicit_formula_of_M_2}. Exchanging the order of integration, we have
\begin{equation} \label{eq:integral-written_in_F(us)}
  \frac{1}{4\pi^2} \int^1_{-1} \log(u-x) \frac{1}{\sqrt{1-x^2}} \pv \int^1_{-1} \frac{\sqrt{1-s^2}}{(u-s)(s-x)} dsdx =\frac{1}{4\pi^2} \int^1_{-1} \frac{\sqrt{1-s^2}}{u-s} F(u,s) ds,
\end{equation}
where
\begin{equation}
F(u,s) = \pv \int^1_{-1} \log(u-x) \frac{1}{(s-x)\sqrt{1-x^2}} dx.
\end{equation}
To evaluate $F(u,s)$, we note (with the change of variable $x = \sin\theta$)
\begin{equation} \label{eq:evaluation_of_partial_F(us)}
\begin{split}
  \frac{\partial}{\partial u} F(u,s) = & \pv \int^1_{-1} \frac{1}{(u-x)(s-x)\sqrt{1-x^2}} dx \\
  = & \frac{1}{s-u} \left[ \int^{\frac{\pi}{2}}_{-\frac{\pi}{2}} \frac{1}{u-\sin\theta} d\theta - \pv \int^{\frac{\pi}{2}}_{-\frac{\pi}{2}} \frac{1}{s-\sin\theta} d\theta \right] \\
  = & \frac{1}{s-u} \left[ \frac{\pi}{\sqrt{u^2-1}} - 0 \right] = \frac{\pi}{(s-u)\sqrt{u^2-1}}.
\end{split}
\end{equation}
On the other hand,
\begin{equation} \label{eq:separation_of_F(us)}
\begin{split}
F(u,s) = & \log u \pv \int^1_{-1} \frac{1}{(s-x)\sqrt{1-x^2}} dx + \pv \int^1_{-1} \log\frac{u-x}{u} \frac{1}{(s-x)\sqrt{1-x^2}} dx \\
= & 0 + \pv \int^1_{-1} \log\frac{u-x}{u} \frac{1}{(s-x)\sqrt{1-x^2}} dx,
\end{split}
\end{equation}
and it implies that $F(u,s) \to 0$ as $u \to \infty$. Thus from \eqref{eq:evaluation_of_partial_F(us)} and \eqref{eq:separation_of_F(us)}
\begin{equation} \label{eq:explicit_formula_of_F(u,s)}
\begin{split}
F(u,s) = \frac{\pi}{\sqrt{1-s^2}} \left[ \arcsin s - \arcsin \frac{us-1}{u-s} \right].
\end{split}
\end{equation}
Because
\begin{align} 
  \int^1_{-1} \frac{\arcsin s}{u-s} ds = & \frac{\pi}{2} \log(u^2-1) + \pi \log \frac{u+\sqrt{u^2-1}}{2}, \label{eq:evaluation_of_easier_integral_involving_arcsin} \\
  \int^1_{-1} \frac{\arcsin \frac{us-1}{u-s}}{u-s} ds = & -\frac{\pi}{2} \log(u^2-1) - \pi \log \frac{u+\sqrt{u^2-1}}{2}. \label{eq:evaluation_of_harder_integral_involving_arcsin}
\end{align}
By \eqref{eq:integral-written_in_F(us)}, \eqref{eq:explicit_formula_of_F(u,s)}, \eqref{eq:evaluation_of_easier_integral_involving_arcsin} and \eqref{eq:evaluation_of_harder_integral_involving_arcsin} we find
\begin{equation} \label{eq:the_2nd_term_of_M_2}
  e^{\frac{1}{4\pi^2} \int \frac{\log(u-x)}{\sqrt{1-x^2}} \pv \int \frac{\sqrt{1-s^2}}{(u-s)(s-x)} dsdx} = \frac{1}{2} \left[ \left( \frac{u+1}{u-1} \right)^{1/4} + \left( \frac{u-1}{u+1} \right)^{1/4} \right].
\end{equation}

Next we compute the second term in \eqref{eq:explicit_formula_of_M_2}. 
For the equilibrium measure $d\mu(x) = \Psi(x)dx$ on its support $[-1,1]$, By \cite[Formula 6.135]{Deift99}, we have 
\begin{equation} \label{eq:relation_between_psi_and_G}
\Psi(x) = \Re \tilde{G}_+(x) = -\Re \tilde{G}_-(x) = \frac{1}{2}(\tilde{G}_+(x) - \tilde{G}_-(x)),
\end{equation}
where $\tilde{G}(z)$ is an analytic function \cite[Formula 6.141]{Deift99}
\begin{equation} \label{eq:formula_of_tilde_G}
  \tilde{G}(z) = \frac{\sqrt{z^2-1}}{2\pi^2 i} \int^1_{-1} \frac{V'(s)}{(s-z)\sqrt{1-s^2}} ds
\end{equation}
is an analytic function in $\compC \setminus [-1,1]$ and $\sqrt{z^2-1}$ is analytic in $\compC \setminus [-1,1]$ with $\sqrt{z^2-1} \sim z$, $z \to \infty$.

Thus by \eqref{eq:relation_between_psi_and_G}, \eqref{eq:formula_of_tilde_G} and exchanging the order of integration,
\begin{equation}
  \begin{split}
    \int^1_{-1} \log(u-s) d\mu(x) = & \frac{1}{2\pi^2} \int^1_{-1} \log(u-s) \sqrt{1-x^2} \pv \int^1_{-1} \frac{V'(s)}{\sqrt{1-s^2}(s-x)} dsdx \\
    = & \frac{1}{2\pi^2} \int^1_{-1} \frac{V'(s)}{\sqrt{1-s^2}} G(u,s) ds,
  \end{split}
\end{equation}
where
\begin{equation}
  G(u,s) := \pv \int \log(u-x)\sqrt{1-x^2} \frac{1}{s-x} dx.
\end{equation}
Similar to \eqref{eq:explicit_formula_of_F(u,s)}, we have
\begin{multline} \label{eq:evaluation_of_G(u,s)}
  G(u,s) = \pi \left[ \sqrt{u^2-1}-u  + s\log \frac{u+\sqrt{u^2-1}}{2} \right. \\
\left. -\sqrt{1-s^2} \arctan \frac{su-1}{\sqrt{(1-s^2)(1-u^2)}} + \sqrt{1-s^2} \arctan \frac{s}{\sqrt{1-s^2}} \right]
\end{multline}
Therefore
\begin{multline} \label{eq:evaluation_of_integral_of_equi_measure}
  \int^1_{-1} \log(u-s) d\mu(x)  = \frac{1}{2\pi} \sqrt{u^2-1} \int^1_{-1} \frac{V'(s)}{\sqrt{1-s^2}} ds + \frac{1}{2\pi} \log \frac{u+\sqrt{u^2-1}}{2} \int^1_{-1} \frac{sV'(s)}{\sqrt{1-s^2}} ds \\
  + \frac{1}{2\pi} \int^1_{-1} V'(s) \left( \arctan \frac{s}{\sqrt{1-s^2}} - \arctan \frac{su-1}{\sqrt{(1-s^2)(1-u^2)}} \right) ds.
\end{multline}
Using \eqref{eq:evaluation_of_G(u,s)} and integration by parts,
\begin{equation} \label{eq:evaluation_of_integral_in_V}
  \begin{split}
    & \frac{1}{4\pi^2} \int^1_{-1} \frac{V(x)}{\sqrt{1-x^2}} \pv \int^1_{-1} \frac{\sqrt{1-s^2}}{(u-s)(x-s)} dsdx \\
    = & \frac{1}{4\pi^2} \int^1_{-1} \frac{V(s)}{\sqrt{1-s^2}} \frac{\partial}{\partial u} G(u,s) ds \\
    = & -\frac{1}{8}(V(1)-V(-1)) + \frac{1}{4\pi} \int^1_{-1} V'(s) \left( \arctan \frac{s}{\sqrt{1-s^2}} + 2 \arctan \sqrt{\frac{(u+1)(1-s)}{(u-1)(1+s)}} \right) ds.
  \end{split}
\end{equation}
Similarly,
\begin{equation} \label{eq:evaluation_of_integral_in_V'}
  \begin{split}
    \frac{1}{4\pi^2} \int^1_{-1} \frac{\log(u-x)}{\sqrt{1-x^2}} \pv \int^1_{-1} \frac{V'(s)\sqrt{1-s^2}}{s-x} dsdx = & \frac{1}{4\pi^2} \int^1_{-1} V'(s)\sqrt{1-s^2} F(u,s) ds \\
    = & \frac{1}{4\pi^2} \int^1_{-1} V'(s) \left( \arcsin s - \arcsin \frac{us-1}{u-s} \right) ds.
  \end{split}
\end{equation}
Thus by \eqref{eq:evaluation_of_integral_of_equi_measure}, \eqref{eq:evaluation_of_integral_in_V}, \eqref{eq:evaluation_of_integral_in_V'} and the identity
\begin{multline}
  \arcsin s - \arcsin \frac{us-1}{u-s} + \arctan \frac{s}{\sqrt{1-s^2}} + 2\arctan \sqrt{\frac{(u+1)(1-s)}{(u-1)(1+s)}} \\
  - 2\arctan \frac{s}{\sqrt{1-s^2}} + 2 \arctan \frac{su-1}{\sqrt{(1-s^2)(u^2-1)}} = -\frac{\pi}{2},
\end{multline}
we have
\begin{equation}
  \begin{split}
    & \frac{1}{4\pi^2} \int^1_{-1} \frac{\log(u-x)}{\sqrt{1-x^2}} \pv \int^1_{-1} \frac{V'(s)\sqrt{1-s^2}}{s-x} dsdx \\
    & - \frac{1}{4\pi^2} \int^1_{-1} \frac{V(x)}{\sqrt{1-x^2}} \pv \int^1_{-1} \frac{\sqrt{1-s^2}}{(u-s)(s-x)} dsdx  - \int^1_{-1} \log(u-x) d\mu(x) \\
    = & -\frac{1}{2\pi} \sqrt{u^2-1} \int^1_{-1} \frac{V'(s)}{\sqrt{1-s^2}} ds - \frac{1}{2\pi} \log \frac{u+\sqrt{u^2-1}}{2} \int^1_{-1} \frac{sV'(s)}{\sqrt{1-s^2}} ds.
  \end{split}
\end{equation}
By the property \cite[Formulas 6.143 and 6.144]{Deift99}
\begin{gather}
\int^1_{-1} \frac{V'(x)}{\sqrt{1-x^2}} d\mu(x) = 0, \\
\int^1_{-1} \frac{xV'(x)}{\sqrt{1-x^2}} d\mu(x) = 2\pi,
\end{gather}
we further simplify the second factor on the right-hand of \eqref{eq:explicit_formula_of_M_2} as
\begin{equation} \label{eq:the_3rd_term_of_M_2}
  e^{ \left\{ \begin{split} & \scriptstyle \frac{1}{4\pi^2} \int^1_{-1} \frac{\log(u-x)}{\sqrt{1-x^2}} \pv \int^1_{-1} \frac{V'(s)\sqrt{1-s^2}}{s-x} dsdx \\ & \scriptstyle - \frac{1}{4\pi^2} \int^1_{-1} \frac{V(x)}{\sqrt{1-x^2}} \pv \int^1_{-1} \frac{\sqrt{1-s^2}}{(u-s)(s-x)} dsdx  - \int^1_{-1} \log(u-x) d\mu(x) \end{split} \right\}} = 2(u-\sqrt{u^2-1}).
\end{equation}

Substituting \eqref{eq:the_2nd_term_of_M_2} and \eqref{eq:the_3rd_term_of_M_2} into \eqref{eq:explicit_formula_of_M_2}, we obtain
\begin{equation}
  M_2(u) = \left[ \left( \frac{u+1}{u-1} \right)^{1/4} - \left( \frac{u-1}{u+1} \right)^{1/4} \right] 2 e^{\int^1_{-1} V(x) d\mu(x) - \frac{1}{4\pi^2} \int^1_{-1} \frac{V(x)}{\sqrt{1-x^2}} \pv \int^1_{-1} \frac{V'(s)\sqrt{1-s^2}}{s-x} dsdx},
\end{equation}
and prove \eqref{eq:final_formula_of_M_2} in the case that $J = [-1, 1]$. The general case can be proved by a simple rescaling.

\bibliographystyle{abbrv}
\bibliography{bibliography}

\end{document}